\documentclass[11pt]{article}

\usepackage[T1]{fontenc}
\usepackage[full]{textcomp}
\usepackage{newtxtext}
\usepackage{cabin} %
\usepackage[varqu,varl]{inconsolata} %
\usepackage[english]{babel}
\usepackage{amsmath}
\usepackage{amssymb}

\usepackage{scalefnt}
\usepackage{letltxmacro}
\LetLtxMacro{\oldtextsc}{\textsc}
\renewcommand{\textsc}[1]{\oldtextsc{\scalefont{1.1}#1}}

\usepackage{amsfonts}       %
\usepackage[nice]{nicefrac}

\usepackage[usenames,dvipsnames]{xcolor}
\definecolor{shadecolor}{gray}{0.9}

\usepackage{afterpage}
\usepackage{framed}

{\endMakeFramed}

\DeclareRobustCommand{\parhead}[1]{\textbf{#1}~}

\usepackage{lineno}

\usepackage{ragged2e}

\newcounter{parcount}

\usepackage{graphicx}
\usepackage[labelfont=bf]{caption}
\usepackage[format=hang]{subcaption}

\usepackage{booktabs}

\usepackage[algoruled]{algorithm2e}
\usepackage{listings}
\usepackage{fancyvrb}
\fvset{fontsize=\normalsize}

\usepackage{natbib}

\usepackage[colorlinks,linktoc=all]{hyperref}
\usepackage[all]{hypcap}
\hypersetup{citecolor=MidnightBlue}
\hypersetup{linkcolor=MidnightBlue}
\hypersetup{urlcolor=MidnightBlue}

\usepackage[nameinlink]{cleveref}

\usepackage[acronym,smallcaps,nowarn]{glossaries}

\lstdefinestyle{mystyle}{
    commentstyle=\color{OliveGreen},
    keywordstyle=\color{BurntOrange},
    numberstyle=\tiny\color{black!60},
    stringstyle=\color{MidnightBlue},
    basicstyle=\ttfamily,
    breakatwhitespace=false,
    breaklines=true,
    captionpos=b,
    keepspaces=true,
    numbers=left,
    numbersep=5pt,
    showspaces=false,
    showstringspaces=false,
    showtabs=false,
    tabsize=2
}
\usepackage{amsthm}
\usepackage{amssymb}
\newtheorem{claim}{Claim}[section]

\DeclareRobustCommand{\mb}[1]{\ensuremath{\boldsymbol{\mathbf{#1}}}}

\newacronym{KL}{kl}{Kullback-Leibler}
\newacronym{ELBO}{elbo}{\emph{evidence lower bound}}
\newacronym{POPELBO}{pop-elbo}{\emph{population evidence lower bound}}

\newacronym{SVI}{svi}{stochastic variational inference}
\newacronym{BUMPVI}{bump-vi}{bumping variational inference}

\newacronym{GMM}{gmm}{Gaussian mixture model}
\newacronym{LDA}{lda}{latent Dirichlet allocation}

\newacronym{SUTVA}{sutva}{stable unit treatment value assumption}

\usepackage[makeroom]{cancel}

\usepackage{xspace}
\newcommand{\ziggy}{\textsc{ziggy}\xspace}

\usepackage{tikz}
\usetikzlibrary{fit,positioning,calc}
\definecolor{light}{RGB}{220, 188, 188}
\definecolor{mid}{RGB}{185, 124, 124}
\definecolor{dark}{RGB}{143, 39, 39}
\definecolor{highlight}{RGB}{0, 255, 0}
\definecolor{gray10}{gray}{0.1}
\definecolor{gray20}{gray}{0.2}
\definecolor{gray30}{gray}{0.3}
\definecolor{gray40}{gray}{0.4}
\definecolor{gray60}{gray}{0.6}
\definecolor{gray70}{gray}{0.7}
\definecolor{gray80}{gray}{0.8}
\definecolor{gray90}{gray}{0.9}
\definecolor{gray95}{gray}{0.95}
\definecolor{comment}{gray}{0.50}

\newcommand{\ba}{\mb{a}}

\newcommand{\bd}{\mb{d}}
\newcommand{\be}{\mb{e}}

\newcommand{\bm}{\mb{m}}

\newcommand{\bu}{\mb{u}}

\newcommand{\bx}{\mb{x}}
\newcommand{\by}{\mb{y}}
\newcommand{\bz}{\mb{z}}

\newcommand{\bK}{\mb{K}}
\newcommand{\bL}{\mb{L}}

\newcommand{\bS}{\mb{S}}

\newcommand{\blambda}{{\mb{\lambda}}}

\newcommand{\btheta}{{\mb{\theta}}}
\newcommand{\bmu}{{\mb{\mu}}}

\newcommand{\bSigma}{\mb{\Sigma}}
\newcommand{\bsigma}{\mb{\sigma}}

\newcommand{\brho}{\mb{\rho}}
\newcommand{\etab}{{\mb{\eta}}}

\newcommand{\given}{\,|\,}

\newcommand{\matern}{{Mat\'ern}}

\newcommand{\Cov}{\mathrm{Cov}}

\usepackage{todonotes}
\presetkeys{todonotes}{color=blue!30,size=\small}{}
\usepackage{enumitem}
\usepackage[final]{changes}

\definechangesauthor[name={Andy}, color=BrickRed]{AM}

\usepackage{graphbox}

\usepackage[top=1in, bottom=1in, left=1.3in, right=1.3in]{geometry}
\usepackage{tcolorbox}

\usepackage{setspace}
\allowdisplaybreaks

\title{\textbf{Mapping Interstellar Dust with Gaussian Processes}}

\author{ acm, la, bl, jpc, dh, dmb } \author{
   Andrew C.~Miller \\
   \texttt{am5171@columbia.edu}\thanks{Work done while at the Data Science Institute at Columbia University.} \and
   Lauren Anderson \\
   \texttt{anders.astro@gmail.com} \and
   Boris Leistedt \\
   \texttt{boris.leistedt@gmail.com} \and
   John P.~Cunningham \\
   \texttt{jpc2181@columbia.edu} \and
   David W.~Hogg \\
   \texttt{david.hogg@nyu.edu} \and
   David M.~Blei \\
   \texttt{david.blei@columbia.edu} } \date{\today}

\begin{document}
\maketitle
\begin{abstract}
  \noindent
  Interstellar dust corrupts nearly every stellar observation, and accounting for
it is crucial to measuring physical properties of stars.  We model the dust
distribution as a spatially varying latent field with a Gaussian process (GP)
and develop a likelihood model and inference method that scales to millions of
astronomical observations.  Modeling interstellar dust is complicated by two
factors.  The first is \textit{integrated observations}. The data come from a
vantage point on Earth and each observation is an integral of the unobserved
function along our line of sight, resulting in a complex likelihood and a more
difficult inference problem than in classical GP inference.  The second
complication is \textit{scale}; stellar catalogs have millions of observations.
To address these challenges we develop \ziggy, a scalable approach to GP
inference with integrated observations based on stochastic variational
inference. We study \ziggy on synthetic data and the Ananke dataset, a
high-fidelity mechanistic model of the Milky Way with millions of stars. \ziggy
reliably infers the spatial dust map with well-calibrated posterior
uncertainties.

\end{abstract}

\section{Introduction}
\label{sec:intro}

The Milky Way galaxy is primarily comprised of dark matter, stars, and
gas.  Within the gas, in its densest and coldest regions, dust
particles form. The stars of the Milky Way are embedded in this field
of dust.

Back on Earth, astronomers try to map the stars, measuring the
location, apparent brightness, and color of each. But the dust between
Earth and a star obscures the light, corrupting the astronomers'
observation.  Relative to the star's true brightness and color, the
dust \textit{dims} the brightness and \textit{reddens} the color.
This corruption is called \textit{extinction}, and it complicates
inferences about a star's true distance and other true properties.
Stellar extinction has hindered many studies of stars in the Milky Way
disk, which is where most dust lies, and therefore the largest
extinctions~\citep{mathis1990interstellar}.

To cope with these corruptions, astronomers need a map of interstellar
dust, an estimate of the density of dust at each location in the
Galaxy.  An accurate dust map could be used to correct our
astronomical measurements and sharpen our knowledge of the Galaxy's
stars.  Moreover, a dust map may be of independent scientific
interest---for example, it may reveal macroscopic properties of the
shape of the Milky Way, such as spiral arms.\footnote{Though hints of
  spiral structures have been inferred from other observations
  \citep{yg1976spiral, bc2019spiral}, whether the Milky Way galaxy is
  a grand design spiral or if the spiral structures are more
  flocculent remains an open question. An accurate dust map will shed
  light on the recent dynamical history of the Milky Way.  }

But constructing such a map is difficult.  We are embedded in our own
dust field, and so we cannot directly observe it.  Rather, we can only
observe a noisy \emph{integral} of the field along the line of sight
between Earth and a star --- the starlight extinction. (Here, the line of sight 
is the straight line through space between the star and the earth.) Thus the
inference problem is to calculate a single, coherent spatial field of
dust that can explain the observed extinction of millions of spatially
distributed stars.  In this paper, we use a large data set of
astronomical measurements to infer a three-dimensional map of
interstellar dust.

The data comes from standard practice in astronomy, which is to
estimate the extinction of an individual star from its observed color
and brightness using a physical model of star formation and
evolution.\footnote{The observed color and brightness are directly
  derived from telescope images (e.g.~photometry).  Extinction
  measurements are backed out from a theoretical distribution of
  dust-free star colors and luminosities, where luminosity is the
  intrinsic brightness of a star.  This distribution can be based on
  isochrones or empirically derived from a region of the sky known to
  have no dust (e.g.~the Milky Way halo).  The extinction corresponds
  to how far the observed color and brightness are from the set of
  theoretically plausible colors and luminosities. The extinction
  uncertainty incorporates both noise in the photometric measurement
  and prior uncertainty over the range of plausible colors and
  luminosities.  } This procedure results in an estimate of the
extinction and an approximate variance of the estimate about the true
extinction value.  However, this estimate is derived from a single
star's color and brightness---it ignores information about the spatial
structure of the Milky Way and the fact that all stars are observed
through the same three-dimensional density of dust.

The dust map is a spatial distribution of dust---an unobserved spatial
function---that can help refine these noisy measurements.  Each
measurement is modeled conditional on the function and the
three-dimensional location of the star.\footnote{The spatial locations
  of stars can be derived from parallax measurements.  In this work we
  consider them fixed, but in future work will additionally consider
  their uncertainty.}  Such a model shrinks estimates toward their
true values and reduces uncertainty about them.  More formally, the
unknown dust map is a function
${\rho : \mathbb{R}^D \mapsto \mathbb{R}}$ from a three-dimensional
location in the universe $x$ to the density of dust at that location.
Our goal is to estimate this function from data.

We take a Bayesian nonparametric approach.  We place a Gaussian
process (GP) prior on the dust map and posit a likelihood function for
how astronomical observations arise from it. We develop a scalable
approximate posterior inference algorithm to estimate the posterior
dust map from large-scale astronomical data.  In addition to the scale
of the data, the main challenge is the likelihood function.  Typical
spatial analysis involves data that are noisy evaluations of an
unobserved function, and the likelihood is simple.  Astronomical data,
however, comes from a limited vantage point---we only observe the
latent dust map as an \emph{integrated process} along a line of site
to a star, and this integral is baked into the likelihood of the
observation.

More formally, let $e_n$ be the true extinction of starlight for star
$n$ at location $x_n$, let $a_n$ be the noisily measured extinction
with uncertainty $\sigma_n$.
Given the unknown dust map, we
model the noisy measurement as a Gaussian whose mean is an integral
from Earth to $x_n$ \citep{kh2017inferring}.  With covariance function
$k_{\btheta}(\cdot, \cdot)$, the full model is
\begin{align}
  \rho(\cdot) &\sim GP(0, k_{\theta}(\cdot, \cdot))
  \label{eq:prior} \\
  a_n &\sim \mathcal{N}\left(e_n, \sigma_n^2\right), \quad
        \textrm{where \,} e_n \triangleq \int_{x \in R_n} \rho(x) dx.
  \label{eq:likelihood}
\end{align}
Here $R_n$ is the set of points (i.e. the ray) from Earth to $x_n$;
the extinction $e_n$ is in an integral of the latent dust map
$\rho(\cdot)$, and it is through this integral that the map enters the
likelihood. Figure~\ref{fig:pointwise-integrated-comparison}
graphically depicts the distinction between pointwise and
\textit{integrated} observations.

The data are $N$ locations, extinctions, and measurement errors, denoted
$\mathcal{D} = \{a_n, x_n, \sigma_n^2\}_{n=1}^{N}$.  Conditional on
the data, we want to calculate the posterior dust map
$\pi(\rho \mid \mathcal{D})$.  This posterior can estimate different
properties of the latent dust, e.g., the density of dust $\rho(\cdot)$
at a new location, the integral of dust $\rho(\cdot)$ over new sets,
and posterior uncertainty about these values.  Such inferences can aid
many stellar studies.

But the posterior is difficult to compute, complicated by both the
integrated likelihood and the large scale of the data.  (Modern
catalogs of astronomical data contain observations of millions of
stars \citep{brown2018gaia, aguado2019fifteenth}.)  In theory the
scale is helpful---each star provides information about the unobserved dust map.
But scaling a Gaussian process to millions of observations is a
significant computational challenge.

We overcome these challenges with a scalable algorithm for Gaussian process
inference with integrated observations, which we call \ziggy.\footnote{\ziggy
is named for Ziggy Stardust, David Bowie's alter ego.}  In particular, \ziggy
is a stochastic variational inference algorithm, one that uses stochastic
optimization, inducing points, and variational inference to approximate the
posterior.  It handles general covariance functions and scales to millions of
data points.  We study \ziggy on both synthetic data and the Ananke data, which
comes from a high-fidelity mechanistic model of the Milky Way.  We find that
\ziggy accurately reconstructs the dust map and accuracy continues to improve
as the number of observations grows above a million.

In our applied setting, the integrated observations introduce technical
challenges.  \ziggy builds upon an existing scalable Bayesian inference framework
\citep{hoffman2013stochastic, hensman2013gaussian}, and introduces additional
approximations and computational techniques to address the challenges created
by integrated observations.

The paper is organized as follows.  Section~\ref{sec:related-work}
describes related work for estimating interstellar dust.
Section~\ref{sec:exact-inference} formally sets up the problem and
describes exact inference in GP models with integrated observations.
Section~\ref{sec:scaling-inference} develops a stochastic variational
inference algorithm that scales to millions of stellar observations.
Appendix~\ref{sec:synthetic-experiments} studies \ziggy on a synthetic
two-dimensional example, comparing various settings of the algorithm
that trade computation for accuracy and flexibility.
Section~\ref{sec:domain-simulation} studies \ziggy with the Ananke
data set, showing that it recovers a well-calibrated three-dimensional
dust map that accurately predicts extinctions and global dust
structure. Section~\ref{sec:conclusion} concludes the paper and
discusses directions for future research.

\begin{figure}[t!]
\centering
\begin{subfigure}[b]{.46\textwidth}
\centering
\includegraphics[width=\textwidth]{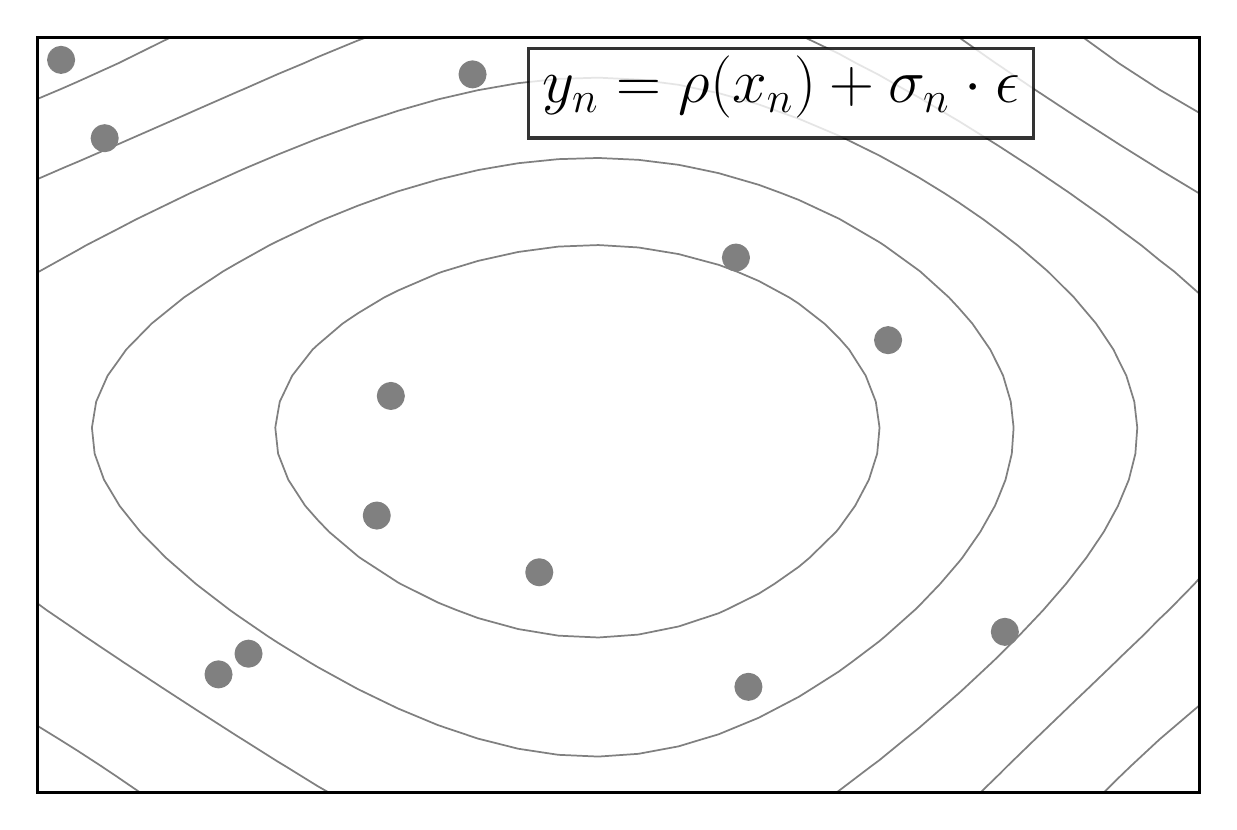}
\caption{Pointwise}
\end{subfigure}
~
\begin{subfigure}[b]{.46\textwidth}
\centering
\includegraphics[width=\textwidth]{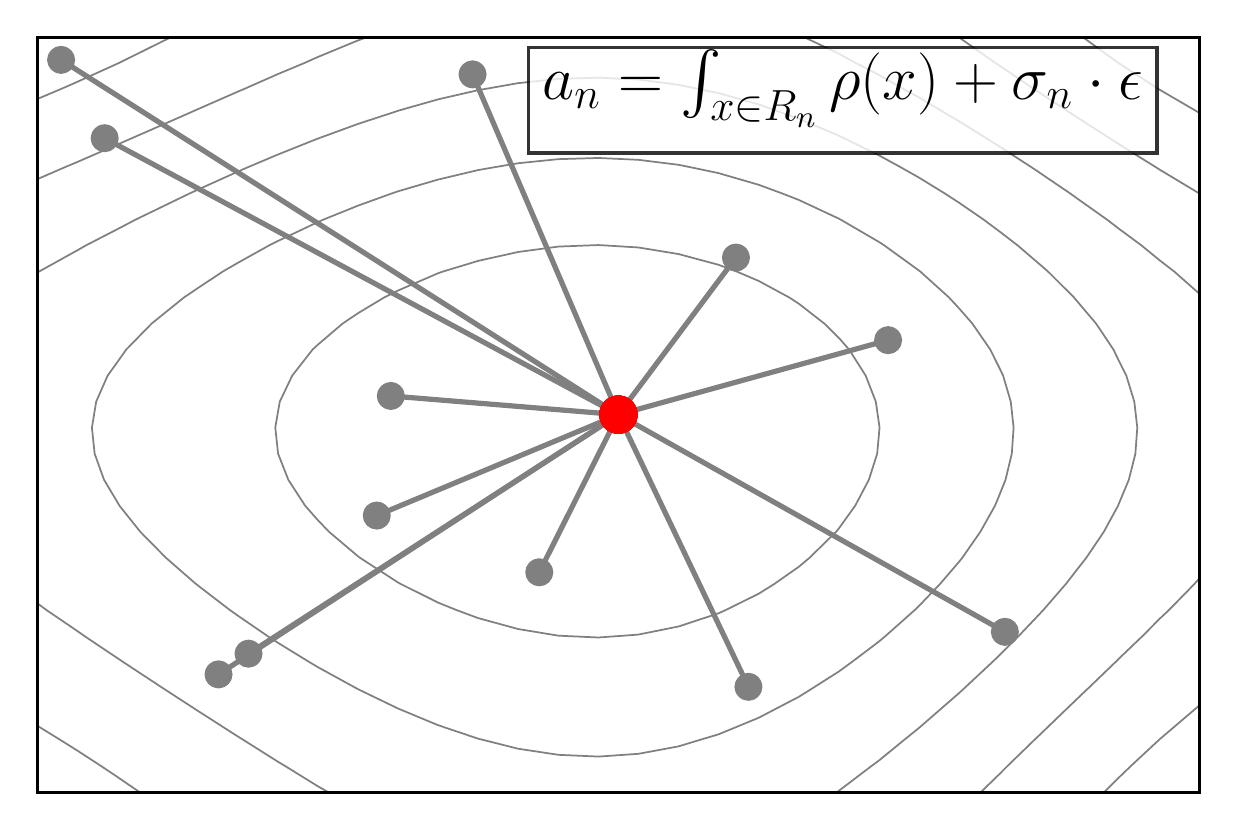}
\caption{Integrated}
\end{subfigure}
\caption{Reconstructing an unobserved function from pointwise (left) 
and integrated (right) observations.  The latent function governs
noisy pointwise observations, $y_n$, and integrated
observations, $a_n$.  The task is to reconstruct the unobserved 
function $\rho(x)$, depicted by the grey contours, everywhere in the domain.
Our perspective is limited to the origin (red dot); our observations of this
process are integrated along a compact set (i.e.~a ray).}
\label{fig:pointwise-integrated-comparison}
\end{figure}

\section{Related research}
\label{sec:related-work}
This work builds on a foundation of research in both astronomy and
statistics.

\parhead{Estimating the latent dust map.} The seminal work of
\citet{schlegel1998maps} estimates the two-dimensional map of dust
across the full sky using dust emission (rather than dust absorption, which we use here).
Emission is a more direct estimate of the dust, but doesn't allow for direct inferences of 3D structure.
Stars within the Milky Way,
however, are embedded in the three-dimensional dust field.  Estimating
per-star extinctions requires characterizing the dust field as a
function of distance, motivating the construction of three-dimensional
dust maps.

More recent approaches use hierarchical spatial models of noisy
integrated observations to reconstruct three-dimensional maps.  These
approaches, however, rely on \emph{discretizing space} and modeling
integrated observations as finite sums.  Various discretization
strategies have been proposed, each with different computational
demands.

One approach models discretized lines of sight to every star
\citep{kh2017inferring}, but this approach can only accommodate a few
thousand stellar observations.  Another approach uses a fine mesh of
cubic voxels to describe the dust distribution.  This introduces a
different computational trade off --- larger voxels introduce
unrealistic spatial artifacts \citep{green2018galactic, green20193d},
while smaller voxels significantly increase the computational burden
\citep{leike2019charting}.  Moreover, the theoretical length scale of
the dust distribution is small; discretizing three-dimensional space
with fine enough resolution to capture this small length scale will
have to rely on additional approximations to scale to the entire
volume of the Milky Way.  The approach we develop here avoids
discretization, works directly with continuous space, and scales to
millions of observations.

\parhead{Spatial statistics and Gaussian processes.} The field of spatial
statistics has developed many tools to estimate unobserved functions
from noisy measurements \citep{cressie1992statistics}.  One ubiquitous and
fundamental method is Gaussian process regression or \emph{kriging}, which
interpolates or smooths a latent function given noisy observation at nearby
locations \citep{krige1951statistical,
  matheron1963principles, matheron1973intrinsic,
  cressie1990origins,rasmussen2006gaussian}.  A GP defines a
probability distribution over the unobserved function that encodes prior
assumptions about some properties --- continuity, smoothness, and amplitude
--- while remaining flexible.  GPs admit analytically tractable inference
routines in typical settings, making them a useful prior for unobserved
functions.

Scaling GPs to massive datasets is a more recent challenge.  One
approach is to approximate the GP with inducing point methods
\citep{quinonero2005unifying, snelson2006sparse}, where 
process values at specific places in the space, the inducing points,
are learned from the data, and predictions are then produced using the 
inducing point process values. Inference using
this approximation scales better than exact GP inference.  Bolstering
their use, recent theoretical work has characterized the error of
inducing point GP approximations \citep{burt2019rates}.  However,
inducing point methods still require analyzing the entire dataset
simultaneously; this requirement limits their applicability to
datasets where the number of observations is in the millions or
billions.

Another thread of GP research in spatial statistics embeds them in
more complicated models and scales them \citep{cressie1992statistics,
  banerjee2014hierarchical}.  For example, nearest-neighbor GP's are a
scalable approximation for estimating the posterior of a latent
spatial field~\citep{datta2016hierarchical}.  Similarly, the
integrated nested Laplace approximations (INLA) framework is an
approximate inference methodology that computes posterior marginal
uncertainty in latent Gaussian models, including Gaussian process
models with computationally tractable precision matrices
\citep{rue2009approximate}.  Several scalable Gaussian process approximations
are reviewed in \citet{heaton2019case}.

Gaussian processes have also been used within astronomy applications beyond
building spatial maps of interstellar dust.  Methods to scale
one-dimensional GP regression to millions of observations have been
developed and deployed with success \citep{ambikasaran2014fast,
foreman2017fast}.

\parhead{Scalable Bayesian inference.} We build on methods
for scaling Bayesian inference to massive datasets.  Variational
methods \citep{jordan1999introduction, wainwright2008graphical,
  blei2017variational} are a computationally efficient alternative to
Monte Carlo methods for posterior approximation. Variational inference
treats approximate inference as an optimization problem, fitting a
parameterized family to be close to the exact posterior.  Stochastic
optimization~\citep{Robbins:1951} scales variational inference to
large datasets, iteratively subsampling from the data to produce
cheap, noisy gradients of the objective; this strategy is called
stochastic variational inference (SVI)~\citep{hoffman2013stochastic}.
Building off of SVI, stochastic variational Gaussian processes (SVGP)
bring in inducing point methods to scale Gaussian processes to massive
datasets \citep{hensman2013gaussian}.  We build on the SVGP framework
here, adapting it to fit the GP model of astronomical data.

\section{A Gaussian process model of starlight extinctions}
\label{sec:method}

The data are $N$ stars, each one a tuple of its location $x_n$, a
noisy measurement of the extinction $a_n$, and the variance of the
observation $\sigma^2_n$; denote the data set
$\mathcal{D} \triangleq \{x_n, a_n, \sigma^2_n \}_{n=1}^N$.  Given the
latent dust map $\rho(x)$, the likelihood of each observation $a_n$ is
defined in Equation~\ref{eq:likelihood}, where the region of
integration $R_n$ is the ray that originates at the origin $O$
(i.e.~the Earth) and ends at the spatial location of the star, $x_n$,
\begin{align}
  R_n
  &=
    \{ \alpha \cdot O +
    (1-\alpha) \cdot x_n : \alpha \in [0, 1] \} \,.
\end{align}
To complete the model, \Cref{eq:prior} places a Gaussian process prior
on the dust map.

Given the data, the posterior distribution,
$p(\rho \given \mathcal{D})$, summarizes the evidence about the latent
dust map.
\added[id=AM]{
Via the relationship between $\rho(x_n)$ and $e_n$ defined
in \Cref{eq:likelihood}, we use this posterior over $\rho$ to form estimates
of the extinction for observed
stars $e_n$, the extinction at new locations $e_*$, and the dust map
itself at new locations $\rho(x_*)$}.  Specifically, the compute the posterior expectations
\begin{align}
  \hat{\rho}(x_*)
  \triangleq \mathbb{E}\left[ \rho(x^*) \given \mathcal{D} \right] \,,\quad
  \hat{e}_n
  \triangleq
    \mathbb{E}\left[ e_n \given x_n, \mathcal{D} \right] \,,\quad
  \hat{e}_*
  \triangleq
    \mathbb{E}\left[ e_* \given x_*, \mathcal{D} \right] \, .
\end{align}
The posterior variance, e.g.,
$\mathbb{V}\left[ e_* \given x_* , \mathcal{D} \right]$, describes the
uncertainty of the estimates.

The posterior dust map $\rho(\cdot)$ synthesizes information from nearby
sources to de-noise or shrink an individual extinction $e_n$, resulting in
more accurate inferences with smaller estimator variance.  Moreover it
enables estimating the dust density at a new location $x_*$ or to estimate
estimate the extinction along a path to the new point.

In the following sections we discuss how to calculate the posterior
and how to approximate it with large data sets of stellar
observations.  Before that, however, we discuss the choice of a
Gaussian process prior.

\parhead{Why use a Gaussian process?}  First, Gaussian processes can
flexibly model complex unobserved latent functions.  The dust map
$\rho(\cdot)$ is not easily described by a parametric functional form,
but GPs can adapt highly complex, non-linear functions to describe
data.
Second, high-level properties of the unobserved dust map can be
captured by the choice of the GP's covariance function.  In
particular, we assume the dust map has specific structure---it is
continuous and two nearby locations are more likely to have similar
dust density values than two distant locations.  GPs naturally capture
such spatial coherence.
Third, all stellar observations are derived from the same underlying
dust map $\rho(x)$.  Two noisy extinction measurements with nearby
integration paths must be coherently described by a single estimate of
$\rho(\cdot)$.  The continuity, smoothness, and spatial coherence
across observations give traction in forming an accurate statistical
estimate of the unobserved dust map.

\subsection{GP Inference with Integrated Observations}
\label{sec:exact-inference}

In a GP, any finite set of function evaluations are multivariate
normal distributed with covariance defined by the covariance function,
$\text{cov}(\rho(x_i), \rho(x_j)) = k_{\btheta}(x_i, x_j)$.  When
observations are corrupted with Gaussian noise---the typical setting
for GP models---the joint distribution between all observations and
latent process quantities remains multivariate normal and the
posterior at a new point $\rho(x^*)$ can be expressed as a few matrix
operations \citep[Chapter 2]{rasmussen2006gaussian}.
Inference in Gaussian process models with noisy \emph{integrated
  observations}--- defined in Equation~\ref{eq:likelihood}---closely
mirrors inference with the more typical, noisy pointwise observations.

Consider the model in \Cref{eq:prior,eq:likelihood}.  We derive the
posterior distribution by first characterizing the joint distribution
over $\rho(x_*)$ and $\ba = (a_1, \dots, a_N)$
\citep{rasmussen2006gaussian}.  Observe that linear operations of
multivariate normal random variables remain multivariate normal in
distribution.  As noted in \citet[Section~9.8]{rasmussen2006gaussian},
because integration is a linear operation, the value of the integral
of $\rho$ over the domain $\mathcal{X}$ will remain Gaussian.  This is
also true for any collection of definite integrals and the
corresponding (Gaussian) noisy observations.

The joint distribution over $\ba$ and $\rho(x_*) \triangleq \rho_*$ remains
multivariate normal.  Given the mean and covariance of this joint
distribution, we can compute the posterior distribution over $\rho_*$ (or
any set of $\rho_*$ points) given observations $\ba$, $\bx$, and
$\bsigma^2$.  The joint distribution over $\ba$ and $\rho_*$ mimics the
standard Gaussian process setup
\begin{align}
  \begin{pmatrix}
  \ba \\ \rho_*
  \end{pmatrix}
  &\sim \mathcal{N}\left( \mb{0}, 
  \begin{pmatrix}
    {\bK}_{\be} + \bSigma_{\ba} & \bK_{\be,*} \\
    {\bK}_{*, \be} & \bK_* \\
  \end{pmatrix}
  \right) \, ,
  \label{eq:integrated-joint}
\end{align}
where the covariance matrix blocks are defined
\begin{align}
    ({\bK}_{\be})_{ij}   &= \Cov(e_i, e_j) \,,\quad
    ({\bK}_{*, \be})_{*j} = \Cov(\rho_*, e_j) \,,\quad
    \bK_*                 = \Cov(\rho_*, \rho_*) \, \,,
    \label{eq:pointwise-cov}
\end{align}
\added[id=AM]{and the observation noise matrix $\bSigma_{\ba} \equiv \texttt{diag}(\sigma_1^2, \dots, \sigma_N^2)$.  We note that because $\rho_*$ is scalar, $\bK_{*}$ is a $1\times 1$ matrix. }
Computing these covariance entries enables us to populate the joint
covariance matrix and use standard multivariate normal conditioning
for posterior inference.  This conditioning mimics the equations for
standard Gaussian process inference,
\begin{align}
  \rho(x_*) &\sim \mathcal{N}( \mu_*, \Sigma_*)  \\
  \mu_* &= \bK_{*,\be} \left( \bK_{\be} + \bSigma_{\ba} \right)^{-1} \ba \\
  \Sigma_* &= \bK_{*} - \bK_{*,\be} \left( \bK_{\be} + \bSigma_{\ba} \right)^{-1}\bK_{\be, *}.
\end{align}
To compute this posterior, we need to compute the covariances in
Equation~\ref{eq:pointwise-cov}.

The chosen covariance function directly specifies $\Cov(\rho_*, \rho_*)$
The other terms, $\Cov(e_i, e_j)$ and $\Cov(\rho_*, e_j)$
are slightly more complex, and require integrating the covariance function
in one or both of its arguments.

\vspace{.5em}
\begin{claim}[Semi-integrated Covariance Function]
\label{claim:semi-integrated}
The covariance between a process value $\rho_i \triangleq \rho(x_i)$ and an
integrated value $e_j \triangleq \int_{x \in R_j} \rho(x)dx$ takes the form
\begin{align}
  \Cov(\rho_i, e_j)
    &= \int_{x \in R_j} k(x_i, x) dx 
     = k^{(\mathrm{semi})}(x_i, x_j) \, .  \label{eq:k-semi-def}
\end{align}
For consistency, we will write the integrated argument second.
\end{claim}
The proof is in Appendix~\ref{app:integrated-covariance}.

\vspace{.5em}
\begin{claim}[Doubly-integrated Covariance Function]
\label{claim:doubly-integrated}
The covariance between two integrated process values $e_i$ and $e_j$ takes
the form
\begin{align}
  \Cov(e_i, e_j)
    &= \int_{(x,x') \in R_i \times R_j} k(x, x')dx dx' 
    = k^{(\mathrm{double})}(x_i, x_j) \, .
\label{eq:k-doubly-def}
\end{align}
\end{claim}
The proof for the doubly-integrated kernel is similar to the
semi-integrated case.

To summarize, if we can compute these two types of covariance values
then we can plug the result into the distribution in
Equation~\ref{eq:integrated-joint}.  As for standard GP inference, we
can then manipulate the joint distribution using Gaussian conditioning
to compute the posterior distribution over the value of the unobserved
function $\rho(x_*)$ or functionals over test locations, for example
integrated values $e_*$. Algorithm~\ref{alg:integrated-gp-inference}
summarizes GP inference using integrated measurements.

\subsection{Choice of covariance function}

Assumptions about the dust map $\rho(\cdot)$ can be encoded in the
covariance function $k_{\btheta}(\cdot, \cdot)$. We focus on kernels
that are \emph{stationary} and \emph{isotropic}, which can be written
$k_{\btheta}(x, y) = \sigma^2 k\left(\frac{|x-y|}{\ell} \right)$
where $\btheta \triangleq (\sigma^2, \ell)$.  The parameter $\ell$ is the
\emph{length scale}, and the parameter $\sigma^2$ is the process marginal variance
\citep{genton2001classes}. The length scale controls how smooth
the function is; the marginal variance influences the function's amplitude.
We will compare three families of kernel functions: squared
exponential, \matern, and a kernel from Gneiting
\citep{gneiting2002compactly}.
The squared exponential kernel --- a commonly used covariance function for
Gaussian process regression ---  admits an analytically tractable form
for the semi-integrated kernel, but not the doubly integrated kernel.
(See \Cref{app:covariance-function-defs}.)
Other kernels, such as the \matern~\citep{matern1960spatial} or the kernel
presented in \citet{gneiting2002compactly} do not readily admit a
semi-integrated form, which complicates their use with integrated
observations.  Similarly, none of the mentioned kernels, including the
squared exponential, admit a doubly-integrated form over two line segments.
We develop a method to overcome this technical limitation in the following
section.  Furthermore, we use the analytic tractability of the
semi-integrated kernel to validate this method on the squared exponential
kernel.

Beyond the choice of covariance function family, the covariance
function parameters (e.g.~the length scale and process variance) ought
to be tuned.  Tuning GP covariance parameters can be complex --- the
process values $\brho$ can strongly depend on parameters $\sigma^2$
and $\ell$, making it difficult to explore their joint space.  We
discuss this phenomenon within the context of our scalable approximate
inference algorithm in Section~\ref{sec:scaling-inference}.

\section{Scaling Integrated GP Inference}
\label{sec:scaling-inference}

Gaussian process models scale poorly to large datasets
\citep{quinonero2005unifying, titsias2009variational}.  Computing the
posterior distribution of the latent function, the posterior predictive
distribution of new observations, or the marginal likelihood of observed
data (e.g.~for model comparison) all require computing the inverse or the
determinant of a $N \times N$ matrix --- a $O(N^3)$ operation.  As $N$
grows larger than a few thousand observations, this computation becomes
intractable. 
\added[id=AM]{To address this bottleneck, a common strategy 
is to approximate Gaussian process inference using $M << N$ 
\emph{inducing points}, or spatial locations in the input space that are used to approximate the full Gaussian process. 
Inducing point approximations can be interpreted in multiple ways. 
One interpretation is that the inducing points are used to construct a rank $M$ approximation to the $N \times N$ matrix that can be efficiently inverted (e.g., $O(N \cdot M^3)$) \citep{quinonero2005unifying}.
Another interpretation (which we adopt below), is that the inducing points are used to define a family of variational distributions for approximating the posterior over $\rho$, and optimizing the variational objective avoids direct manipulation of the $N \times N$ matrix \citep{titsias2009variational, hensman2013gaussian}.
}

Integrated observations create an additional
computational issue --- the semi-integrated or
doubly-integrated covariance value between any two values may be
unavailable in closed form (except in some special cases) and may
require high-precision numerical approximation.
{With the integrated observations considered here, we have an
additional computational issue.  Calculating the semi-integrated or
doubly-integrated covariance value between any two values may be
unavailable in closed form (except in some special cases) and may
require high-precision numerical approximation}.
The exact GP
inference algorithm in \Cref{alg:integrated-gp-inference} requires
computing a $N\times N$ matrix of doubly-integrated covariance values,
which will be computationally prohibitive when $N$ is a modest size.
Furthermore, when tuning covariance function parameters or comparing
covariance functions, the $N \times N$ matrix of numerically
integrated covariance function values will need to be recomputed many
times.

Thus, to use the model to analyze millions of integrated stellar
observations, we derive a scalable variational inference algorithm to
perform approximate posterior inference.  Variational inference
approximates a posterior distribution by optimizing a parameterized
family---the variational distribution---to be close to the exact
posterior~\citep{jordan1999introduction,wainwright2008graphical,blei2017variational}. For
the model here, our approach builds on the stochastic variational
Gaussian process (SVGP) framework \citep{hensman2013gaussian}.  It
uses stochastic optimization to approximate the
posterior~\citep{hoffman2013stochastic}, operating on small
mini-batches of observations.

Further, to accommodate the integrated observations, we construct a
Monte Carlo estimate of the semi-integrated kernel that generalizes to
covariance functions that do not admit a closed form semi-integrated
version.  Additionally, we use the rotational invariance of stationary
and isotropic covariance functions to construct a fast approximation
to the doubly integrated covariance kernel.

This section describes the SVGP framework and our approach to adapt the
framework to incorporate integrated observations.  \Cref{sec:svgp}
describes relevant details of the SVGP framework, including inducing
points, the variational family, and forming efficient mini-batch estimates
of the variational objective.
\Cref{sec:svgp-integrated-obs} details our approach to do
inference with integrated observations for a general class of
semi-integrated kernels.

\subsection{Stochastic Variational Gaussian Processes}
\label{sec:svgp}

SVGP casts inference in GP models as an optimization problem.
Crucially, the SVGP optimization objective is constructed such that it
can be written as a sum of $N$ terms, each depending on only one data
point.  Given this construction, the objective can be optimized using
stochastic gradients computed with mini-batches of data.  This makes
inference more computationally and memory efficient.

Here, we discuss the SVGP objective 
\replaced[id=AM]{with the standard observations --- noisy versions of $\rho(x)$.}{within the typical noisy pointwise observation framework}.  
We then adapt this approximate
inference framework to integrated observations in
Section~\ref{sec:svgp-integrated-obs}.  Our derivation of the
variational objective is slightly different from
\citet{hensman2013gaussian}---we first define a structured
approximating family and then directly derive the evidence lower bound
objective.

\added[id=AM]{For this section's presentation of inducing 
points, variational inference, and stochastic variational 
Gaussian processes, consider the typical GP model for $N$ observations with 
$\rho \sim GP\left(0, k_{\theta}(\cdot, \cdot)\right)$, and 
standard (i.e., non-integrated) observations 
$y_n \given \bx_n \sim \mathcal{N}(\rho(\bx_n), \sigma_n^2)$.
We write the observation vector $\by = (y_1, \dots, y_N)$ and the
corresponding latent process vector $\brho = (\rho(\bx_1), \dots,
\rho(\bx_N)$).
}

\parhead{Inducing points.} Inducing points are a common tool used to
scale Gaussian process inference \citep{quinonero2005unifying}.  An
inducing point is simply a location in the input space, $\bar{x}$,
which has a corresponding inducing point value $\rho(\bar{x})$. 
\added[id=AM]{Inducing points are typically distributed about the input 
space --- in our experiments we place them on a fixed 
3-dimensional grid throughout space.}

We augment the model above with $M$ inducing points and their
corresponding values
\begin{align}
  \label{eq:inducing-points}
  \bar{\bx}
  &\triangleq
    \bar{\bx}_1, \dots, \bar{\bx}_M
  & \text{inducing points}\\
  \bu
  &\triangleq \rho(\bar{\bx}_1), \dots, \rho(\bar{\bx}_M)
  & \text{inducing point values}
\end{align}
where $\bu$ is the $M$-length vector of values of $\rho$ evaluated at each
of the $M$ inducing points.  Note that the introduction of $\bar{\bx}$ and
$\bu$ has not altered the original model.

The model above specifies two latent variables, $\brho$ and $\bu$,
whose distributions we wish to infer given data $\by$.  That is, we
wish to characterize the posterior distribution
$p(\brho, \bu \given \by, \bx, \bar{\bx})$ in a computationally
efficient way.\footnote{Note that we are considering a model where the
  likelihood and prior are both Gaussian distributions --- a conjugate
  pair --- which implies that the posterior distribution will also be
  Gaussian.  Indeed it will be, but manipulating Gaussian will scale
  cubically in $N$.  The SVGP approach is constructed to scale
  linearly in $N$.} 
\added[id=AM]{SVGP uses inducing points and a specific variational approximation to avoid
the manipulation of the $N \times N$ observation covariance matrix.
}

\parhead{Variational inference.} Variational inference (VI) is an
optimization-based approach to approximate posterior inference
\citep{jordan1999introduction, wainwright2008graphical,
  blei2017variational}.  VI methods posit a \emph{variational family}
of distributions, $q \in Q$, and use optimization techniques to find
the optimal approximate distribution from the set $Q$.  Here, each
element of $Q$ is a distribution over the latent quantities ---
$q(\brho, \bu)$.

The variational inference framework defines an objective to be optimized
--- the most common VI objective is the \emph{evidence lower bound} (ELBO)
\begin{align}
\mathcal{L}(q) &= \mathbb{E}_{q(\bu, \brho)}\left[ 
    \ln p(\brho, \bu, \by \given \bx, \bar{\bx}) - \ln q(\brho, \bu)
  \right] 
  \leq \ln p(\by \given \bx, \bar{\bx}) \, .
  \label{eq:elbo}
\end{align}
Maximizing the ELBO minimizes the KL divergence between
$q(\brho, \bu)$ and the true posterior
$p(\brho, \bu \given \by, \bx, \bar{\bx})$.

\parhead{Stochastic variational Gaussian processes} The stochastic
variational GP  framework (SVGP) \citep{hensman2013gaussian} uses a
particular form for the variational family.  Given the set of inducing
points $\bar{\bx}$, the SVGP variational approximation is
\begin{align}
  q(\brho, \bu) &= p(\brho \given \bu)\,q_{\blambda}(\bu) \,,\quad
  q_{\blambda}(\bu) \triangleq \mathcal{N}(\bu \given \bm, \bS) \,,\quad
  \text{and} \,\, \blambda \triangleq (\bm, \bS),
  \label{eq:variational-approx}
\end{align}
where $p(\brho \given \bu)$ is specified by the Gaussian process prior
and recall that $\bu$ and $\brho$ are jointly multivariate normal.
The variational parameters $\blambda$ are fit by optimizing the ELBO
defined in Equation~\ref{eq:elbo}.

The SVGP approximation defined in Equation~\ref{eq:variational-approx} is
chosen because it has a very specific property --- when plugged into
Equation~\ref{eq:elbo}, the objective decomposes into the sum of $N$
separate terms.  This allows us to create unbiased \emph{mini-batched}
estimators of the objective (and its gradients), enabling efficient
inference.  We can see through a straightforward algebraic manipulation (we
suppress $\rho$ and $\bu$'s dependence on $\bx$ and $\bar{\bx}$ to remove
clutter)
\begin{align}
\mathcal{L}(\blambda)
  &= \mathbb{E}_{q(\brho, \bu)}\left[ \ln p(\brho, \bu, \by)- \ln q_{\blambda}(\rho, \bu) \right] \\
  &= \mathbb{E}_{q_{\blambda}(\bu)p(\brho \given \bu)}\left[
    \ln p(\by \given \brho) + \bcancel{\ln p(\brho \given \bu)} +
    \ln p(\bu) - \ln q_{\blambda}(\bu) - \bcancel{\ln p(\brho \given \bu)}
    \right]
  \label{eq:elbo-cancel} \\
  &= \underbrace{\mathbb{E}_{q_{\blambda}(\bu)p(\brho \given \bu) }\left[
    \ln p(\by \given \brho)
    \right] }_{\text{(i)}} - KL(q_{\blambda}(\bu)\,||\,p(\bu)) \,\, ,
  \label{eq:elbo-obj}
\end{align}
and we can write the term (i) as a sum over the $N$ observations
\begin{align}
        \text{(i)}
        &\triangleq \mathbb{E}_{q_{\blambda}(\bu)}\left[ \mathbb{E}_{p(\brho \given \bu)} \left[ \ln p(\by \given \brho) \right] \right] \\
        &= \mathbb{E}_{q_{\blambda}(\bu)}\left[ \mathbb{E}_{p(\brho \given \bu)} \left[ \sum_{n=1}^N \ln p(y_n \given \rho_n) \right] \right] \\
        &= \sum_{n=1}^N	\underbrace{\mathbb{E}_{q_{\blambda}(\bu)}\left[ \mathbb{E}_{p(\brho \given \bu)} \left[ \ln p(y_n \given \rho_n) \right] \right]}_{\triangleq \mathcal{L}_{n}} \,.
        \label{eq:additive-terms}
\end{align}
When the likelihood $p(y_n \given \rho_n)$ is Gaussian, the
expectation in each of the $\mathcal{L}_n$ terms can be computed
analytically.  Notice the cancellation in
Equation~\ref{eq:elbo-cancel} eliminates the term that involves all
$N$ data points and a $N\times N$ matrix inversion,
$\ln p(\brho \given \bu)$ \citep{titsias2009variational}.

To complete the derivation, we write a mini-batched estimator of the ELBO.
Given a set of randomly selected observations $B$,
\begin{align}
 \hat{\mathcal{L}}(\blambda)
    &= \frac{N}{|B|} \sum_{b \in B} \mathcal{L}_b -
        KL(q_{\blambda}(\bu)\,||\,p(\bu)) \,\, ,
\end{align}
which is an unbiased estimator of the full objective, and only touches $|B|
<< N$ observations.  We can use gradients of this estimator to optimize the
lower bound with respect to variational parameters $\blambda = (\bm, \bS)$.
\added[id=AM]{
With $M$ inducing points, a mini-batched estimator can be computed in $O(|B| \cdot M^3)$ time.
}

When predicting, note that we only condition on the
inducing point values $\bu$, relying only on the prior covariance
between $\bu$ and $\rho_*$.  The structure of this variational
approximation forces the inducing point values $\bu$ to represent all
of the information learned from the data.  While this approximation is
no longer ``nonparametric'' in the traditional sense, it has been
shown that inducing point methods can provide good approximations to
exact GP inference provided enough inducing points are used
\citep{burt2019rates}.  Empirically, we find that we can adequately
tune covariance parameters and predict on held-out data using a modest
number of inducing points (on the order of $M=6{,}000$).

\parhead{Natural gradients}
Given the objective in Equation~\ref{eq:elbo-obj}, we are now tasked
with finding the optimal parameters ${\blambda}$.  We use stochastic
optimization with mini-batches \citep{hoffman2013stochastic} along
with natural gradients \citep{amari1998natural}, which are effective
for variational inference of Gaussian
processes~\citep{hensman2013gaussian}.

For exponential family distributions (e.g.~the multivariate normal),
the natural gradient can be computed by taking the standard gradient
with respect to the mean parameters \citep{hoffman2013stochastic}.
Following \citet{hensman2013gaussian}, the natural gradient for
variational parameters $\blambda = \bm_{\blambda}, \bS_{\blambda}$ is
straightforward to express using the natural parameterization of the
multivariate normal $q_{\blambda}$.
See Appendix~\ref{sec:natural-gradients} for more details.

\parhead{Whitened parameterization.} Draws of
$\bu \sim p(\bu \given \bar{\bx}, \btheta)$ are highly dependent on
$\btheta$.  Consider a draw $\bu$ conditioned on a fixed length scale
parameter of $\ell$.  This same draw will be highly improbable under
the prior with a different length scale parameter, e.g.~$2\ell$.  The
variational approximation $\blambda$ targets the posterior
distribution $p(\bu \given \by, \btheta)$, which is highly dependent
on the structured Gaussian process prior.  This dependence frustrates
the joint inference of $\blambda$ and $\btheta$ --- small changes in
$\btheta$ can make the approximate distribution over $\bu$ suboptimal.
This dependence forces gradient methods to take small steps, which
leads to extremely slow joint inference.

An effective strategy for coping with this dependence is to do
posterior inference in an alternative parameterization of the same
model that decouples variables $\bu$ and $\btheta$.  For structured
latent Gaussian models (e.g.~Gaussian processes) this alternative
model uses the \emph{whitened} prior \citep{murray2010slice} or
\emph{non-centered} parameterization \citep{bernardo2003non}.  For the
dust map model, we whiten the prior by altering the target of
posterior inference --- instead of approximating the posterior over
inducing point values $p(\bu \given \by, \btheta)$, we target the
posterior over inducing point \emph{noise} values
$p(\bz \given \by, \btheta)$, where the relationship between $\bz$ and
$\bu$ is constructed to produce an equivalent model.  For notational
clarity, we define $\bK \triangleq \bK^{(\btheta)}_{\bu,\bu}$ for a
fixed value of $\btheta$ and $\bL$ such that $\bK = \bL \bL^\intercal$
is a general matrix square root.  The whitened model can be written
\begin{align}
    \bz &\sim \mathcal{N}(0, I_M) \,,\quad
    \bu \triangleq \bL \bz \sim \mathcal{N}(0, \bK) \, .
    \label{eq:whitened-model}
\end{align}
We now target the posterior distribution $p(\bz \given \by, \btheta)$ by
fitting an approximation $q_{\blambda}(\bz) = \mathcal{N}(\bz \given
\tilde{\bm}, \tilde{\bS})$.  The linear relationship between $\bz$ and
$\bu$ implies that $\bm = \bL \tilde{\bm}$ and
$\bS = \bL \tilde{\bS}\bL^\intercal$.
We describe additional details of this parameterization in
\Cref{sec:whitened-gradients}.

\subsection{Adapting to integrated observations}
\label{sec:svgp-integrated-obs}
The previous section detailed a scalable approximate inference algorithm
for Gaussian processes in the typical noisy pointwise observation setting.
Here, we adapt this framework to integrated observations.  The challenge is
now efficiently computing the appropriate semi- and doubly-integrated
covariance values to be used within each mini-batch update.

The natural gradient described in Equation~\ref{eq:natgrad} consists of two
types of covariance matrices --- $\bK_{\bu,\bu}$ and $\bK_{\bu, \brho}$.
The former describes the prior covariance between inducing point values and
remains exactly the same in the integrated observation setting.  The latter
describes the cross covariances between process locations and inducing
point values --- this cross covariance must be adjusted to reflect the
covariance between the \emph{integrated process} and the inducing point
values.  This is precisely what is defined by the semi-integrated
covariance function.  The entries of $\bK_{\brho,\bu}$ have to simply be
switched to the semi-integrated covariance defined in
Claim~\ref{claim:semi-integrated}
\begin{align}
    (\bK_{\brho,\bu})_{n,m} &= \int_{x \in R_n} k(x_m, x)dx \, .
    \label{eq:semi-integrated-gram}
\end{align}

In addition to the semi-integrated cross covariances, \ziggy
will need to calculate doubly-integrated covariance terms.  First,
note that the gradient defined in Equation~\ref{eq:natgrad} does not
require doubly integrated terms. This is convenient --- the inducing
point method does not require cross covariances between integrated
observations, and these cross covariances also happen to be the most
expensive to approximate with numerical methods. However, computing
the variational objective (Equation~\ref{eq:elbo-obj}) and making
predictions (Equation~\ref{eq:predictions}) do require the
doubly-integrated diagonal terms, $\bK_{n,n}$ for each observation.

Given the above dependence on integrated covariance terms, the full adaptation
of the SVGP method to integrated observations runs into
two technical hurdles.  First, the semi-integrated covariance is not
available in closed form beyond the simple squared exponential, which is
considered too smooth (and thus of limited capacity) in many settings.
This limitation will make it difficult to use
custom covariance kernels constructed with astronomical theory
\citep{leike2019charting}.  Further, it will make it difficult to compare
different models and choose the best predictor among them. Second,
even when the semi-integrated covariance kernel is available in closed
form, the doubly-integrated covariance is not. The SVGP (and any
inducing point) algorithm only relies on \emph{the diagonal} of the
covariance kernel, so we only need to approximate this diagonal, scalar-valued
function.  We show that this scalar-valued function can be easily
approximated using a small grid of interpolated values, where the ``true''
values are approximated using numerical integration.

We introduce two computationally efficient approximations to these
operations on the covariance function to overcome these technical issues:
(i) Monte Carlo estimates for semi-integrated covariance functions and (ii)
linear interpolation for the doubly-integrated diagonal covariance
function.

\parhead{Monte Carlo estimators for semi-integrated kernels}
The natural gradient updates in Equation~\ref{eq:natgrad} require computing
the covariance between observations and inducing point values, $\bK_{\brho,
\bu}$.  For integrated observations, the covariance between observations
and inducing points is given by the semi-integrated kernel in
Equation~\ref{eq:semi-integrated-gram}.  For kernels with a closed form
semi-integrated kernel (e.g. the squared exponential), we can simply
substitute $k^{(semi)}(\bar{x}_m, x_n)$ for $k(\bar{x}_m, x_n)$ in the loss
and gradients and proceed within the typical SVGP framework.

For kernels without a closed form semi-integrated kernel, we run into
a computational issue --- how do we compute the appropriate cross
covariance values?  One approach is to use numerical integration ---
the semi-integrated covariance is a one-dimensional integral that can
be approximated with quadrature techniques.  However, each batch
requires computing $M \times |B|$ semi-integrated covariance values
--- this can be computationally expensive, which would force us to use
small batches, and this leads to slower learning.

Stochastic natural gradient updates are already an estimate of the true
gradient --- how \emph{good} does this estimate have to be to effectively
find the optimal solution?  Numeric integration approximations to
$k^{(semi)}(\cdot, \cdot)$ are precise to nearly machine precision, but are
computationally expensive.  Analogously we can ask, how \emph{good} do the
semi-integrated covariance approximations have to be to effectively find
the optimal solution?

Pursuing this idea, we propose using a Monte Carlo approximation to the
semi-integrated covariance values --- we sample uniformly along the
integral path and average the original covariance values.  More formally,
we introduce a uniform random variable $\nu_{R_n}$ that takes values along
the ray $R_n$ from the origin to $x_n$.  We can now write the
semi-integrated covariance as
\begin{align*}
  k^{(semi)}(x_m, x_n) &= \int_{x \in R_n} k(x_m, x)dx 
    = |R_n| \int_{\nu \in R_n} \underbrace{\frac{1}{|R_n|}}_{= p(\nu)} k(x_m, x)dx 
    = |R_n| \, \mathbb{E}_{\nu}\left[ k(x_m, \nu) \right] \, .
\end{align*}
This form admits a simple Monte Carlo estimator
\begin{align}
  \nu^{(1)}, \dots, \nu^{(L)} &\sim Unif(R_n) \,,\quad 
    \hat{k}^{(semi)}(x_m, x_n) = \frac{|R_n|}{L}\sum_{\ell} k(x_m, \nu^{(\ell)}) \,\, .
    \label{eq:mc-semi}
\end{align}
It is straightforward to show that Equation~\ref{eq:mc-semi} is an unbiased
and consistent estimator for the true semi-integrated covariance value
evaluated at $x_m$ and $x_n$.  While plugging this directly into the
gradient formula results in a biased estimator for the natural gradient, we
find that with a modest number of samples we can closely match the
optimization performance of the true natural gradient estimator.
\added[id=AM]{We investigate the relationship between samples $L$ and the resulting optimization in Appendix \Cref{fig:optimization-comparison}.}
Appendix~\ref{sec:mc-semi-estimators} contains additional details and
experimental results validating this approach.

Although conceptually similar, note that this Monte Carlo estimate of the
semi-integrated covariance is not equivalent to discretizing along each
line of site as in \citep{kh2017inferring}.  The estimator can use a small
number of samples along the ray and still be useful within stochastic
gradient optimization.

\parhead{Integrated Observation Variance.}  The SVGP algorithm
requires an additional covariance calculation --- the marginal
variance of an integrated observation at point $x_n$,
$k^{\mathrm{(doubly)}}(x_n, x_n)$.  
\added[id=AM]{Note that we do not need to compute the non-diagonal
terms of the doubly integrated kernel, a consequence of 
the separability of Equation \ref{eq:additive-terms}.}
Similar to the semi-integrated kernel
function, the doubly integrated function can be approximated with
numerical methods --- a double quadrature routine for each
observation.  But, similar to the semi-integrated case, this double
quadrature call becomes a computational bottleneck within each batch.
Furthermore, computing the gradient of the ELBO with respect to the
covariance parameters with autodifferentiation tools would require
differentiating through a double quadrature call --- prohibitive in our setting.

We solve this problem by observing that a rotationally invariant (and
stationary) kernel admits a doubly-integrated version that is a function
only of the distances to each of the two integral endpoints.  That is,
$k^{\mathrm{(doubly)}}(x_i, x_j)$ can be written as a function $f(|x_i|, |x_j|)$ for
a fixed setting of covariance function parameters.  Further, we can write
the diagonal of the doubly-integrated kernel as a function of only the
distance to the single point, $k^{\mathrm{(doubly)}}(x_n, x_n) = f(|x_n|)$.

Given that this is a one-dimensional function shared by all $N$ of our
observations, we approximate the doubly-integrated kernel diagonal with
linear interpolation.  The interpolation scheme allows us to cheaply
compute a highly accurate approximation to the doubly-integrated diagonal.
It also enables us to easily backpropagate gradients to the kernel function
parameters, which enables efficient tuning during the variational inference
optimization routine. 
Appendix~\ref{sec:interpolated-doubly-kernels}
describes the interpolation scheme for $k^{\mathrm{(doubly)}}(x_n,x_n)$ in further
detail.

\subsection{Method summary and remarks}

To summarize, \ziggy scales Gaussian process inference with
integrated observations by fusing several strategies:
(i) the explicit representation of the inducing point posterior within
the SVGP framework;
(ii) a whitened parameterization to decouple $\blambda$ and $\btheta$
in the variational objective;
(iii) on-the-fly Monte Carlo estimates of the semi-integrated covariance
function;
(iv) and fast kernel interpolation of the doubly integrated diagonal
covariance function that leverages the stationary and isotropic properties
of the covariance functions considered.

Algorithm~\ref{alg:scalable-integrated-gp-inference} describes the main
loop for updating based on stochastic gradients of the 
variational parameters and the covariance function parameters $\theta$. 
\added[id=AM]{Maximizing the variational lower bound (\Cref{eq:elbo-obj})
with respect to prior parameters $\theta$ is a similar strategy 
to empirical Bayes \citep{efron2008microarrays, efron2019bayes},
a method used to reduce the average error of posterior estimates.} 

We find that the inducing point approach dovetails nicely with
integrated observations---most evaluations of the covariance function
are between inducing points and observations. This only requires
computing the (easier) semi-integrated covariance function.  Further
the doubly-integrated covariance values between any two distinct
observations never needs to be computed.  Rather, only the
\emph{diagonal} of this covariance function needs to be computed, but
this conveniently reduces to a one-dimensional function that can be
efficiently approximated. \ziggy circumvents computing the
(prohibitive) $N \times N$ doubly-integrated covariance.

We note that the SVGP framework for scaling Gaussian process is one
approach to build upon for our application --- other frameworks for scaling
kernel methods could have been employed.  One example is random (or
carefully selected) features for approximating the kernel
function~\citep{rahimi2008random,yang2012nystrom,le2013fastfood}.
While these approximate kernel methods could yield superior performance, 
one reason we reach for the SVGP framework is extensibility.  
The variational inference framework can incorporate more complex
probabilistic graphical model structure.  For this application, additional
types of observations can be incorporated into a more sophisticated likelihood
model to help identify the spatial dust distribution and incorporate additional
sources of uncertainty.

\begin{figure}[t!]
\centering
\begin{subfigure}[b]{.46\textwidth}
\centering
\includegraphics[width=\textwidth]{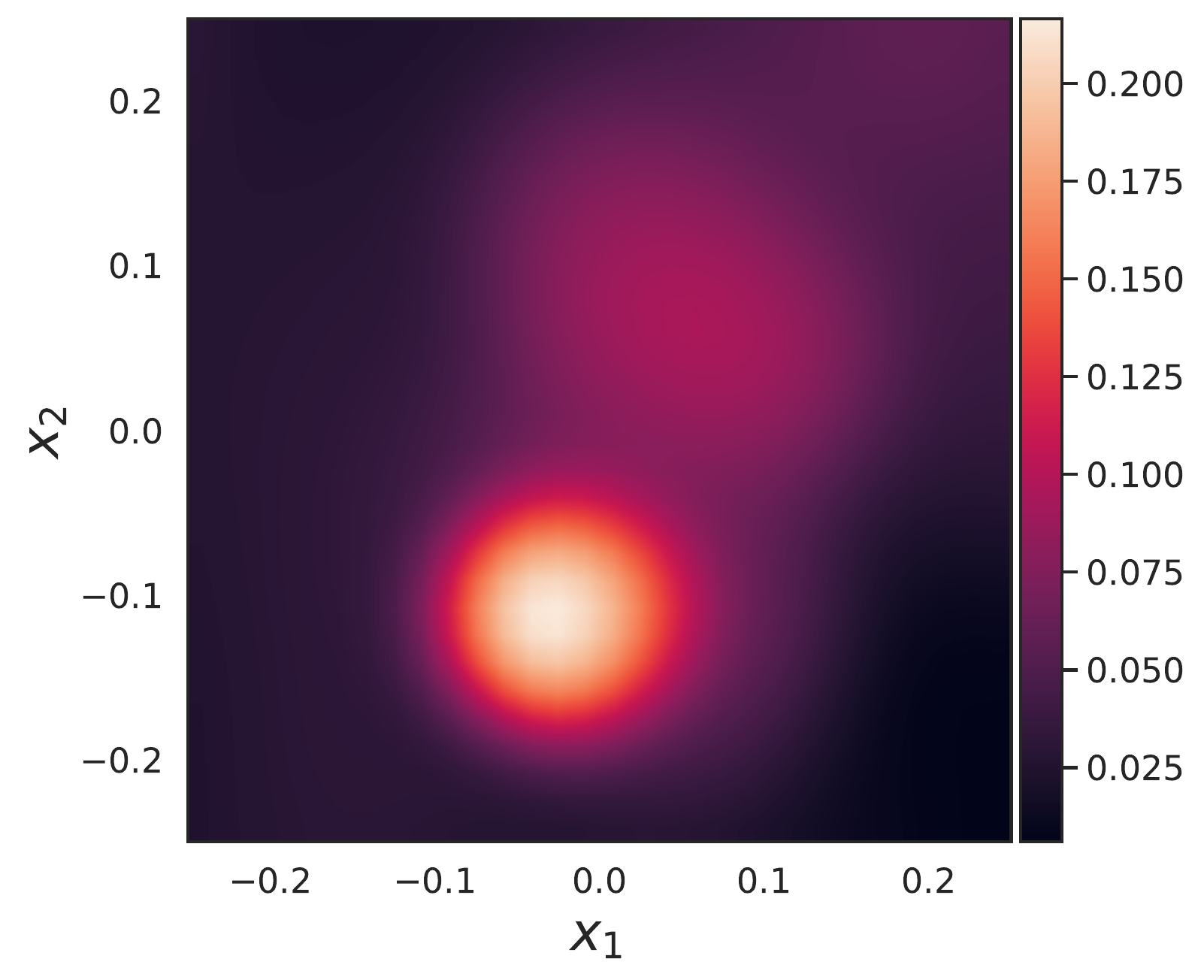}
\caption{True $\rho(x)$}
\label{fig:rho_true_domain}
\end{subfigure}
~
\begin{subfigure}[b]{.46\textwidth}
\centering
\includegraphics[width=\textwidth]{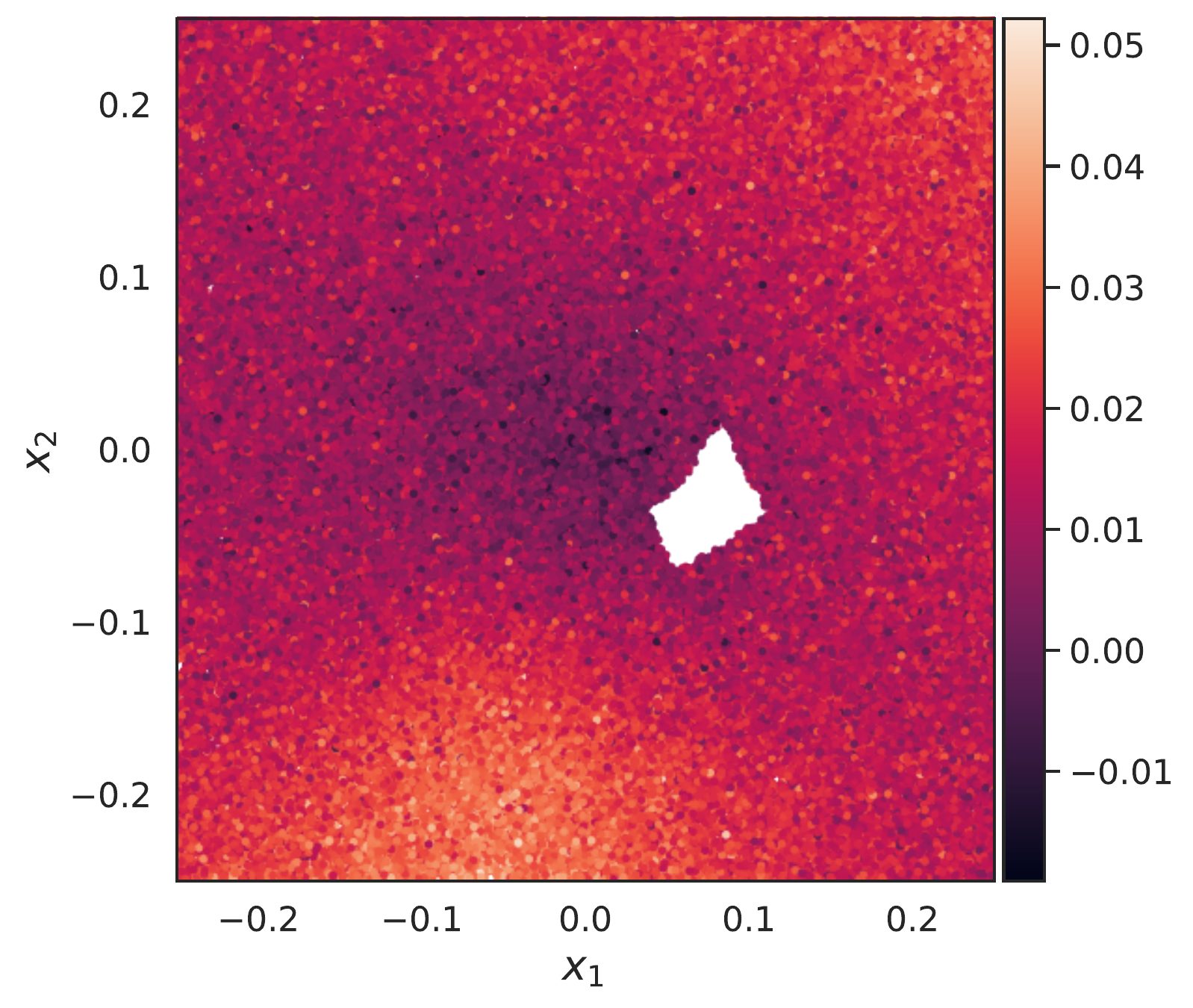}
\caption{Integrated observations $a_n$}
\label{fig:integrated-observations_domain}
\end{subfigure}

\begin{subfigure}[b]{.46\textwidth}
\centering
\includegraphics[width=\textwidth]{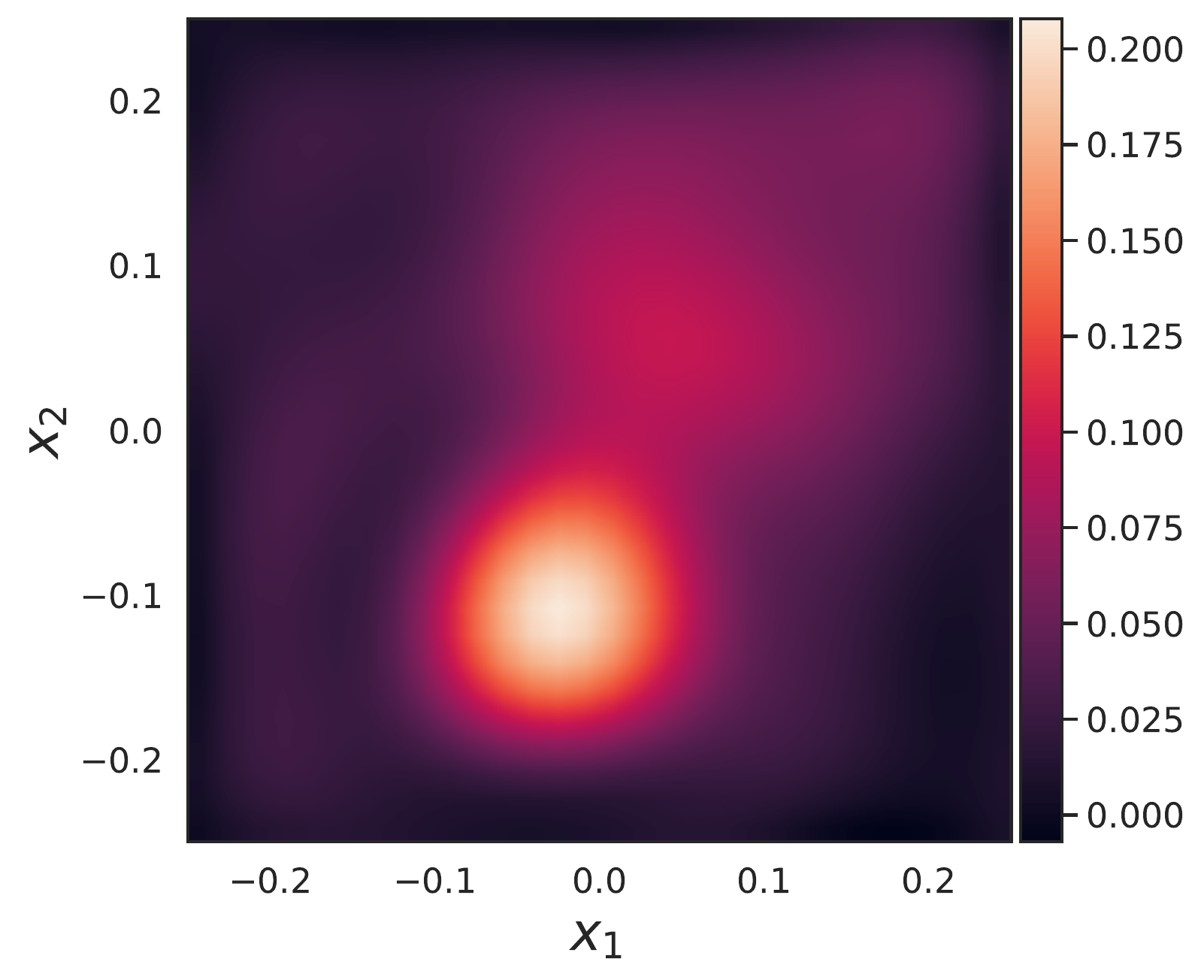}
\caption{Inferred $\mathbb{E}\left[\rho \given \mathcal{D}\right]$.}
\label{fig:rho_inferred_domain}
\end{subfigure}
~
\begin{subfigure}[b]{.46\textwidth}
\centering
\includegraphics[width=\textwidth]{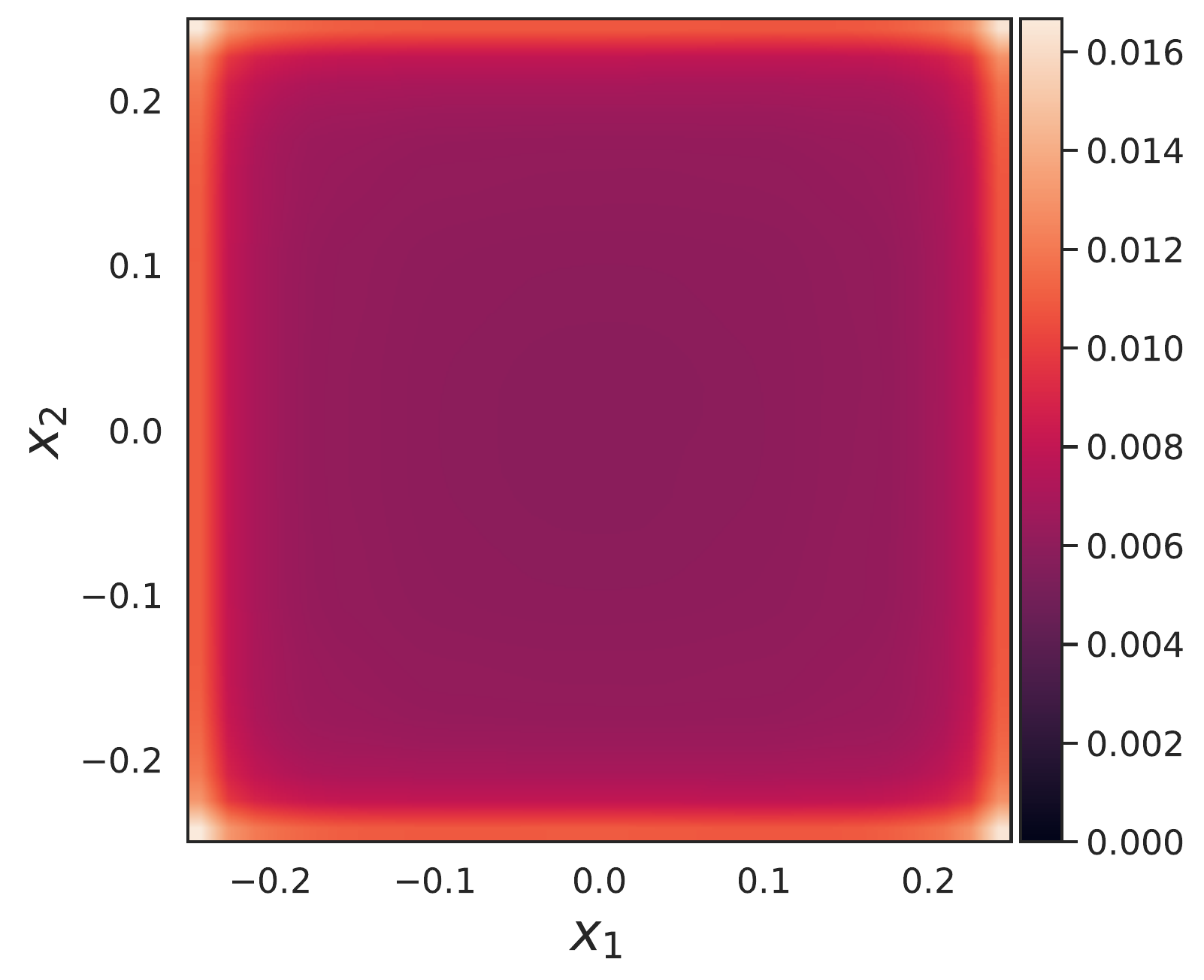}
\caption{Inferred posterior scale $\mathbb{V}\left[\rho(x) \given \mathcal{D}\right]^{1/2}$.}
\label{fig:rho_variance_domain}
\end{subfigure}
\caption{
Integrated observation GPs can reveal spatial structure from a severely limited
vantage point.  This figure summarizes domain simulated data posterior inference
using $N=1{,}000{,}000$ examples in a three-dimensional space.  Panel
\ref{fig:rho_true_domain} depicts the true (unobserved) $\rho(x)$ that generates
the data.  Panel \ref{fig:integrated-observations_domain} depicts the observed
data, noisy integrated measurements of $\rho$ from the origin to stellar
positions in the domain simulation.  \ref{fig:rho_inferred_domain} depicts the
inferred $\rho$ given these observations using \ziggy 
with a squared exponential kernel.  \ref{fig:rho_variance_domain}
depicts the marginal posterior standard deviation at each location in
$\mathcal{X}$-space. All Panels are a slice at $z=0$}
\label{fig:domain-model-fit}
\end{figure}

\section{Validation with a Domain Simulation}
\label{sec:domain-simulation}

We developed \ziggy with the goal of generating a three-dimensional dust
map of the Milky Way
using 2 billion stars in the Gaia dataset \citep{brown2018gaia}. As a step
towards that goal---and to test \ziggy in an applied environment---we 
infer the dust distribution of a mechanistic and physically motivated domain
simulation. The domain dataset is a synthetic Gaia survey
of a high resolution Milky Way-like galaxy simulation. 

The domain dust density $\rho(x)$ was generated in a Milky Way-like galaxy
simulation, one of many in the Latte suite of simulations of Milky-Way-mass
galaxies \citep{aw2016reconciling, ph2018fire}. It was run as part of the
Feedback In Realistic Environments (FIRE) simulation project, which
self-consistently models extinction and star formation in a cosmological context
at high resolution. 

In astronomy, most observations prefer a model in which our Universe is predominantly 
made of cold dark matter \citep{2018planck}. On small scales, however, there are significant challenges 
to this model \citep{1999klypin}, which this suite of simulations helps resolve \citep{aw2016reconciling}. This Latte simulation 
overcomes the computational challenge of including gas, dust, and stars 
with high enough resolution to resolve these tensions. It includes more complex physics of 
“normal” matter, as apposed to just dark matter, including a high resolution disk of gas and dust. 
Because of the results from these simulations, which resolved these tensions, astronomers are more confident in a universal model of dark matter.

The simulation has enough resolution to resolve the
formation of structures in the dense gas that forms dust and stars, creating a
realistic environment for inference. The positions of the extinction
observations, or the positions of stars within $\rho(\cdot)$, was generated by
the Ananke framework. The Ananke framework generated a realistic, mock star
catalog from the FIRE simulation, and kept intact important observational
relationships between gas, extinction, and stellar populations
\citep{sanderson2018ananke}.  This catalogue was specifically designed to
resemble a Data Release 2 Gaia astrometric survey \citep{brown2018gaia}, with a
similar resolution of stellar density as Gaia.
We integrated $\rho(x)$ along lines of sight to 1 million stars in the Ananke
dataset within a $0.5 \mathrm{kpc} \times 0.5 \mathrm{kpc} \times 0.1
\mathrm{kpc}$ region of the synthetic sun. %
So, we generate a set of N noisy integrated observations 
\begin{align}
  x_n &\sim \text{Ananke}(\mathcal{X}) & \text{ stellar locations } \\
  e_n &= \int_{x \in R_n} \rho(x) dx & \text{domain extinctions } \\
  a_n &\sim \mathcal{N}(e_n, \sigma_n^2) & \text{ noisy integrated observation }
\end{align}
where the domain $\mathcal{X} = [(-0.25, 0.25), (-0.25, 0.25), (-0.05, 0.05)]$, 
the median extinction value is $0.01$, and noise variance is chosen to be $\sigma_n = 0.005$.
So our median signal-to-noise value is $3$. 
The true extinctions, $e_n$, are computed to high precision using numerical
quadrature, integrating from the origin to the Ananke stellar locations $x_n$.
The integrated observations are depicted in Figure
\ref{fig:integrated-observations_domain}. We estimate this domain specific dust
density $\rho(\cdot)$ from noisy integrated observations using \ziggy.

In Section \ref{sec:domain-data-size}, we compare posterior estimates of
$\rho(\cdot)$ as a function of training data set size. This highlights the
importance of scaling inference to more observations in order to obtain a more
accurate statistical estimator. In Section \ref{sec:domain-kernel-comparison},
we compare posterior estimates of $\rho(\cdot)$ as a function of kernel choice,
with a fixed data set size. 

\begin{figure}[t!]
\centering
\begin{subfigure}[b]{.45\textwidth}
\centering
\includegraphics[width=\textwidth]{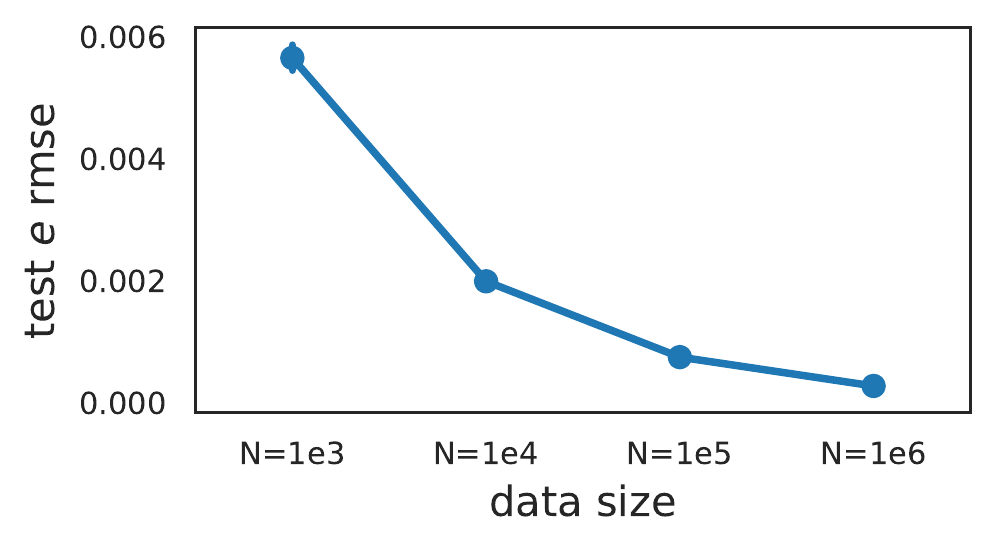}
\caption{Test $e$ RMSE}
\end{subfigure}
~
\begin{subfigure}[b]{.45\textwidth}
\centering
\includegraphics[width=\textwidth]{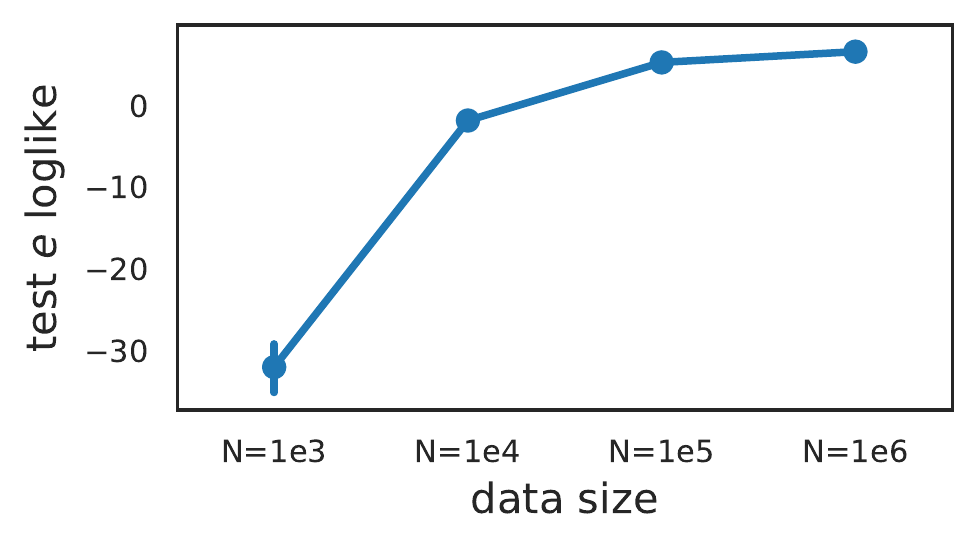}
\caption{Test $e$ log likelihood}
\end{subfigure}

\caption{Scaling to massive $N$
improves estimator performance. We compare predictive RMSE (left) and log
likelihood (right) as a function of data set size $N$ on a held out sample of
test stars.}
\label{fig:domain-data-size}
\end{figure}

\begin{figure}[t!]
\centering
  \begin{subfigure}[b]{.43\textwidth}
  \includegraphics[width=\textwidth]{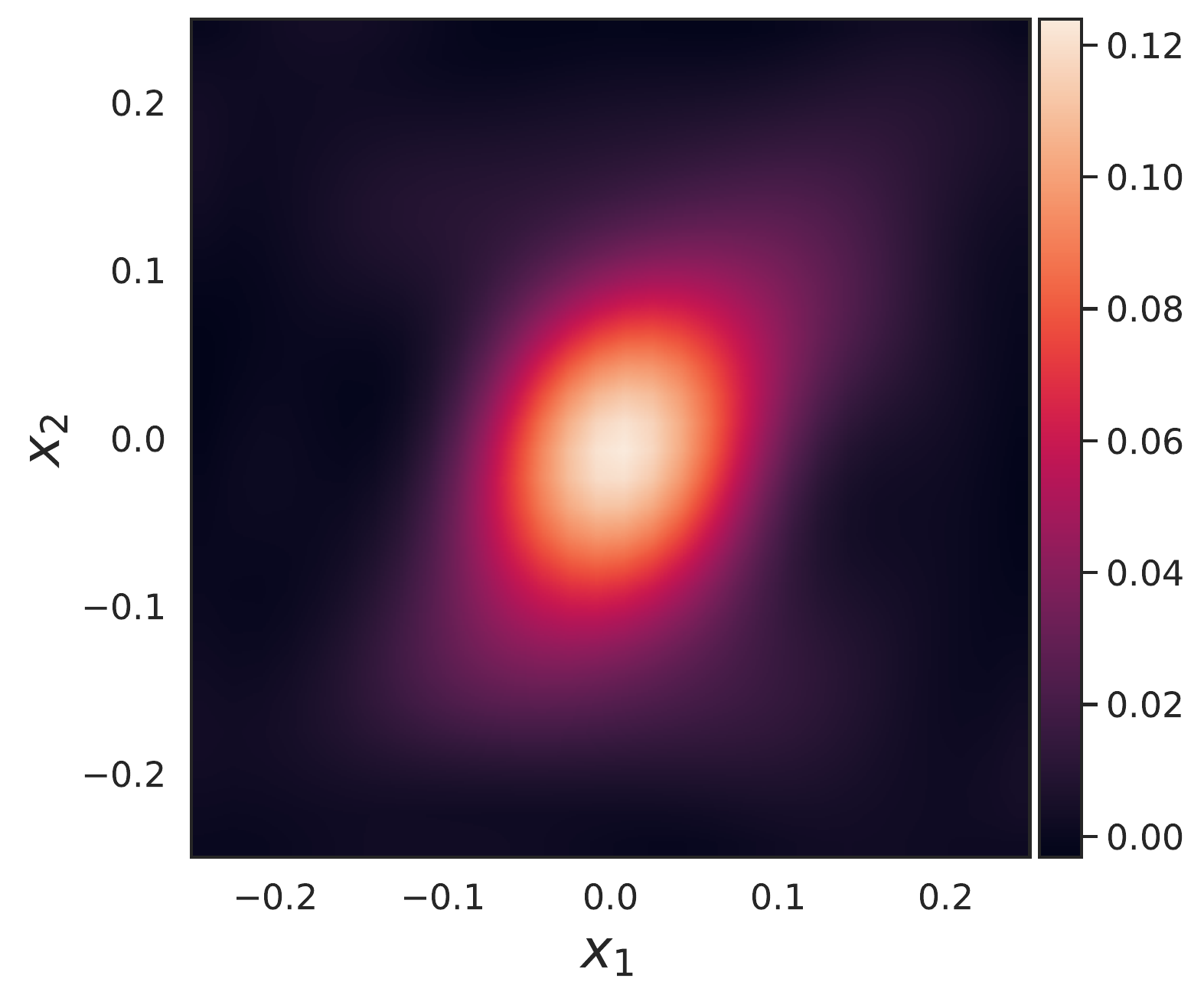}
  \caption{$\mathbb{E}[\rho \given \mathcal{D}]$, $N=1{,}000$}
  \end{subfigure}
  ~
  \begin{subfigure}[b]{.43\textwidth}
  \includegraphics[width=\textwidth]{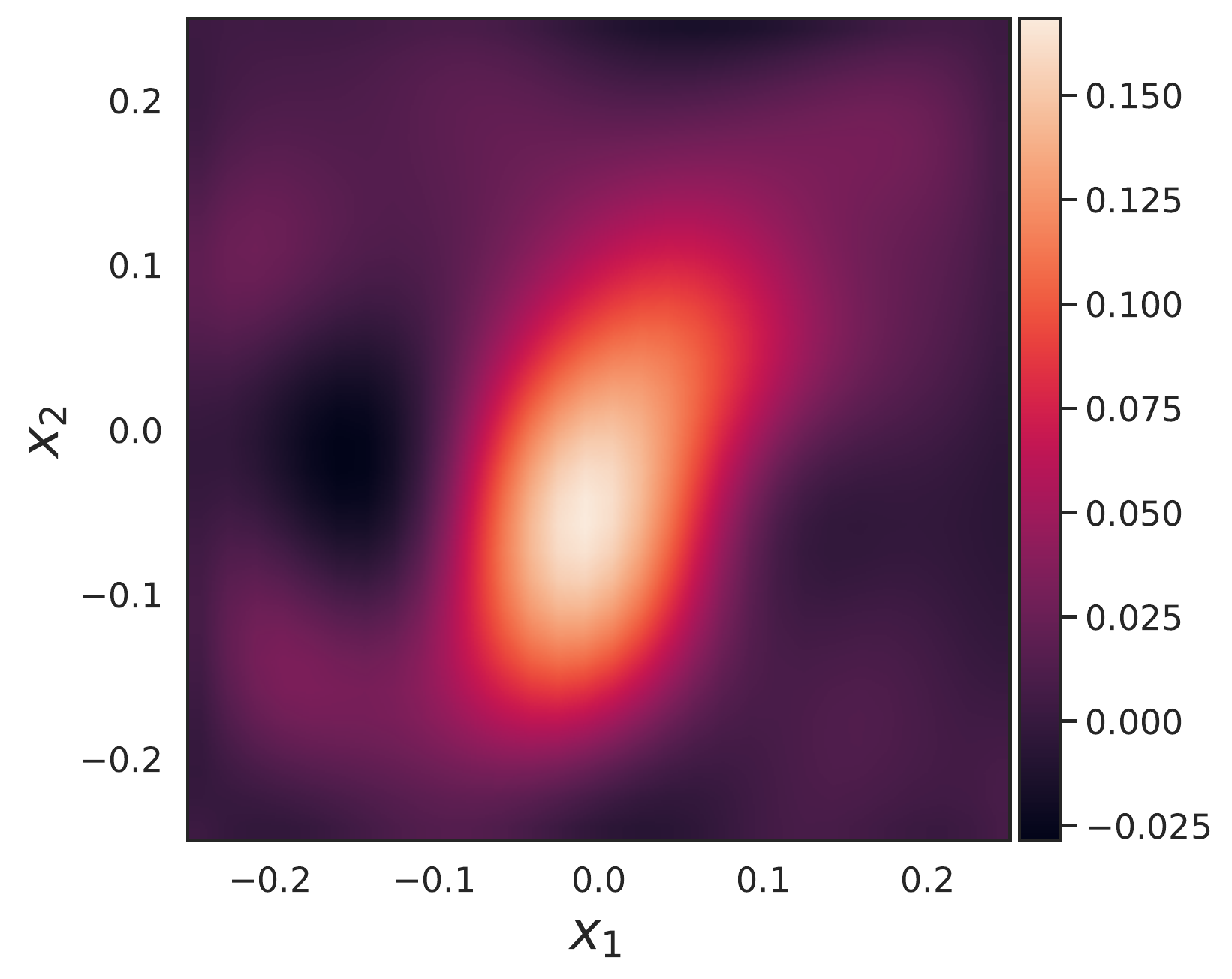}
  \caption{$\mathbb{E}[\rho \given \mathcal{D}]$, $N=10{,}000$}
  \end{subfigure}
  ~
  \begin{subfigure}[b]{.43\textwidth}
  \includegraphics[width=\textwidth]{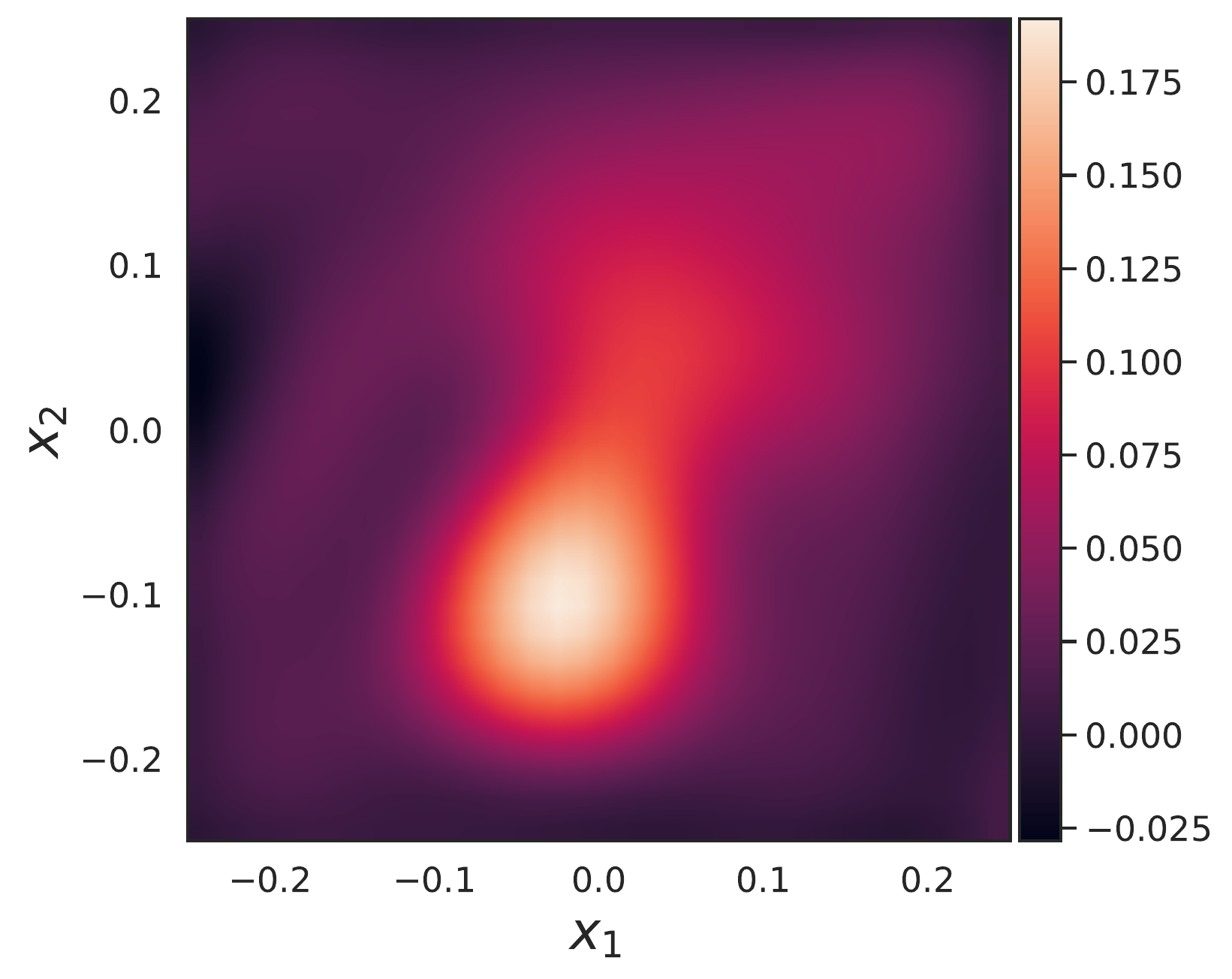}
  \caption{$\mathbb{E}[\rho \given \mathcal{D}]$, $N=100{,}000$}
  \end{subfigure}
  ~
  \begin{subfigure}[b]{.43\textwidth}
  \includegraphics[width=\textwidth]{figs/domain-data-size/SqExp-N-1000000/posterior-fmu-z0.pdf}
  \caption{$\mathbb{E}[\rho \given \mathcal{D}]$, $N=1{,}000{,}000$}
  \end{subfigure}
\caption{Slices of the posterior mean function at $z=0$ for various dataset sizes. More data leads to more accurate predictions of $\rho(\cdot)$.}
\label{fig:domain-data-size-posterior-comparison}
\end{figure}

\subsection{The quality of the estimate}
\label{sec:domain-data-size}
We study the quality of the estimate of $\rho(x)$ as a function of
data set size $N$.  We fit the dust model to the Ananke dataset with
data set sizes $N \in \{10^3, 10^4, 10^5, 10^6\}$ for
${10^5, 10^4, 10^3, 100}$ epochs, respectively.  We trained each model
using mini-batches of size $|B|=2{,}000$, an initial step size of 0.01,
reducing it every epoch by multiplying by a decay factor of
0.99. %
We measure model quality by computing root mean squared error (RMSE)
and log likelihood (LL) on a set of $N_{\mathrm{test}} = 2{,}000$ held
out extinction values. Each model uses $M=16 \times 16 \times 4$
inducing points, evenly spaced in a grid in the input space.

Figure~\ref{fig:domain-model-fit} summarizes the model's fit using one
million observations, the most we tried in this example.
Figure~\ref{fig:rho_inferred_domain} displays the posterior mean for $\rho(x)$
given the one million observations depicted in
\ref{fig:integrated-observations_domain}.  Figure~\ref{fig:rho_variance_domain} shows the
posterior uncertainty (one standard deviation) about the estimated mean.
We see that the million-observation model recovers the true latent function
particularly well.

Figure~\ref{fig:domain-data-size} quantifies the improvement in the RMSE and log
likelihood on a held out sample of test stars as a function of dataset size $N$. 
As the dataset size increases, both statistics drastically improve. 
Figure~\ref{fig:domain-data-size-posterior-comparison} visualizes a
direct comparison of the posterior estimate of the latent
function as we increase data set size $N$.  We can clearly see that as we
incorporate more data into the nonparametric model, the form of the true
underlying function emerges.  To accurately visualize and interpret large
scale features of the latent dust distribution, we will want to incorporate
as many stellar observations as possible. This is particularly true in the low
signal to noise regime of extinction estimates.

\subsection{Comparing kernels}
\label{sec:domain-kernel-comparison}
Similar to the synthetic dataset, we compare multiple models by fitting the
variational approximation and tuning the covariance function parameters for 5
different kernel choices. 

In this domain setup, we ran
each model for 100 epochs, saving the the model with the best ELBO value.
We used mini-batches of size $|B|=2{,}000$, and started the step size at
$0.01$, reducing it every epoch by multiplying by a decay factor of $0.99$.
For kernels that do not have a closed form semi-integrated version, we used
$L=50$ uniform grid Monte Carlo samples to estimate the integrated
covariance.  

We validate the inferences against a set of $2{,}000$ held-out test extinctions $e_*$.
We quantify the quality of the inferences with mean
squared error, log likelihood, and the $\chi^2$-statistic on the test data.
We also visually validate \ziggy with a $QQ$-plot and coverage
comparison, also on test data.  We also visualize the behavior of \ziggy
as a function of distance to the origin ---
specifically, we visualize how accurate and well-calibrated test inferences
are as synthetic observations are made farther away.

In Figure~\ref{fig:domain-error-comparison} we compare the mean squared
error (MSE), the log likelihood, and the $\chi^2$ statistic on the held out
test data.  We compare five different kernels -- (i) Gneiting $\alpha=1$ from \cite{sr2018mwsa}
(ii-iv) \matern, $\nu \in \{1/2, 3/2, 5/2 \}$, and (v) the squared
exponential kernel.
The results match our expectations. The true $\rho(\cdot)$ is quite
smooth, and we see that the smoother kernels tend to have lower predictive
error and higher log likelihood. All kernels have similarly calibrated $\chi^2$ statistic, except for 
\matern, $\nu=3/2$, which under predicts it's posterior variances. This can also be seen in Figure~\ref{fig:domain-qq}.
The test statistic comparison has chosen a good match in
the squared exponential kernel.

To test the calibration of posterior uncertainties, we inspect
statistics of test-sample $z$-scores for $e_*$.  We test that the
statistics of predictive $z$-scores resemble a standard normal
distribution using a QQ-plot.  For a more detailed description fo the
QQ-plot, please see Section \ref{sec:synthetic-kernel-comparison}. To
compare different models, we visualize QQ-plots for the different
covariance functions in Figure~\ref{fig:domain-qq}.  As an additional
summary of calibration, we compare the fraction of predicted examples
covered by $1/2, 1, 2$, and 3 posterior standard deviations,
summarized in Table~\ref{tab:domain-coverage}.  We see that there is
room for improvement in the posterior variances. 
The $z$-scores and empirical coverage show posterior variances that are
well calibrated for the squared exponential and \matern, $\nu=5/2$ kernels, but are too large for the other kernels. 
The models for the less smooth kernels are compute bound, requiring thousands of inducing points 
to tile the domain. We ran the models for 100 epochs, which maxed out the memory and compute time.  
\added[id=AM]{Miscalibration could be the result of misspecification of the model itself --- e.g., mismatch between the covariance function $k$ and the true underlying $\rho$, or the inflexibility of the variational approximation.  To diagnose this shortcoming, methods that allow us to scale the expressivity of the variational approximation --- e.g., larger $M$ --- would be necessary.}

We also depict model fit statistics as a function of distance from the
origin in Figure~\ref{fig:domain-summary-by-distance}.
Figure~\ref{fig:domain-e-value-by-distance} visualizes the test extinction
values.
\ziggy is able to accurately reconstruct the dust map. %
Figure~\ref{fig:domain-e-posterior-variance-by-distance} shows the posterior
standard deviation by distance, which shows a noise reduction of about
$10\times$ for most stars. %
In Figure~\ref{fig:domain-e-zscore-by-distance}
we see that the predictions have posterior variances that are slightly too high, 
causing $z$-scores that lie mostly within $2\sigma$ independent of distance, with only the most nearby stars extending to $3\sigma$. %
Decent calibration, including at far distances, is encouraging, as the
goal is to form good estimates of this density far away from the
observer location.

\begin{figure}[t!]
\centering
\begin{subfigure}[b]{.32\textwidth}
\centering
\includegraphics[width=\textwidth]{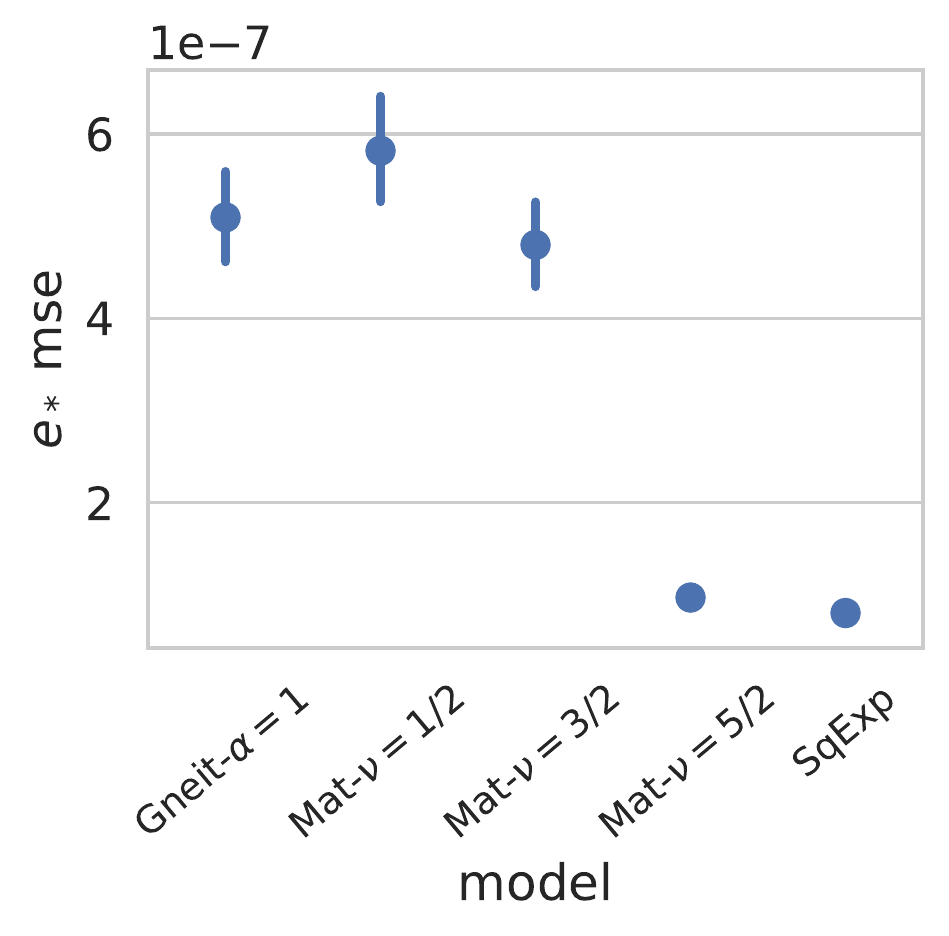}
\caption{Test $e$ MSE}
\end{subfigure}
~
\begin{subfigure}[b]{.32\textwidth}
\centering
\includegraphics[width=\textwidth]{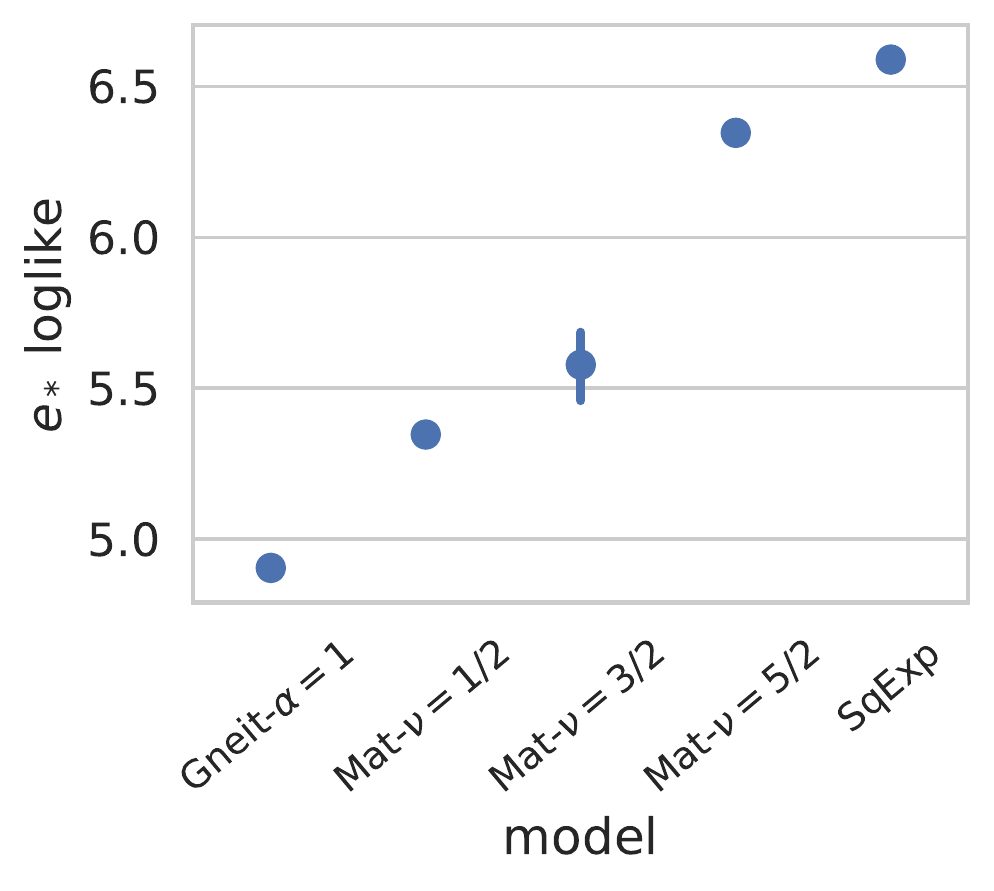}
\caption{Test $e$ log likelihood}
\end{subfigure}
~
\begin{subfigure}[b]{.32\textwidth}
\centering
\includegraphics[width=\textwidth]{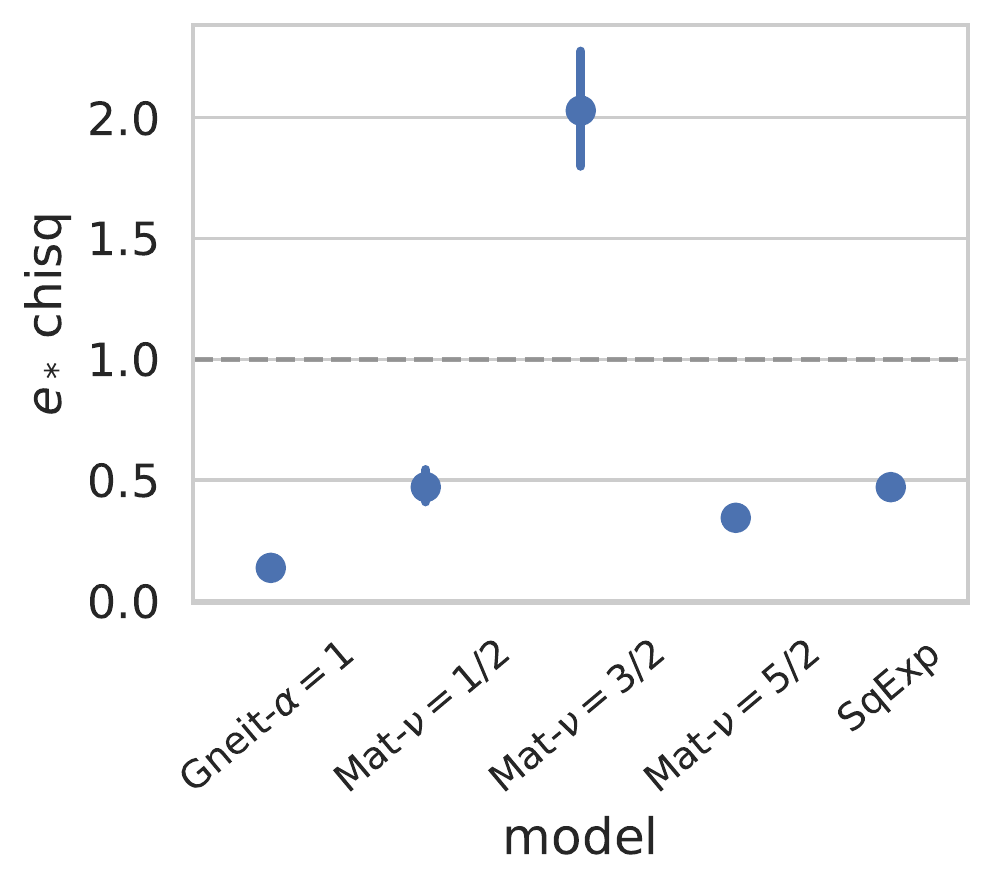}
\caption{Test $e$ $\chi^2$}
\end{subfigure}
~
\caption{Predictive summaries across different models and inference schemes. The smoother models have better predictive summary statistics which is consistent with the underlying density field being very smooth.} 
\label{fig:domain-error-comparison}
\end{figure}

\begin{figure}[t!]
\centering
  \begin{minipage}[c]{0.35\textwidth}
    \includegraphics[width=\textwidth]{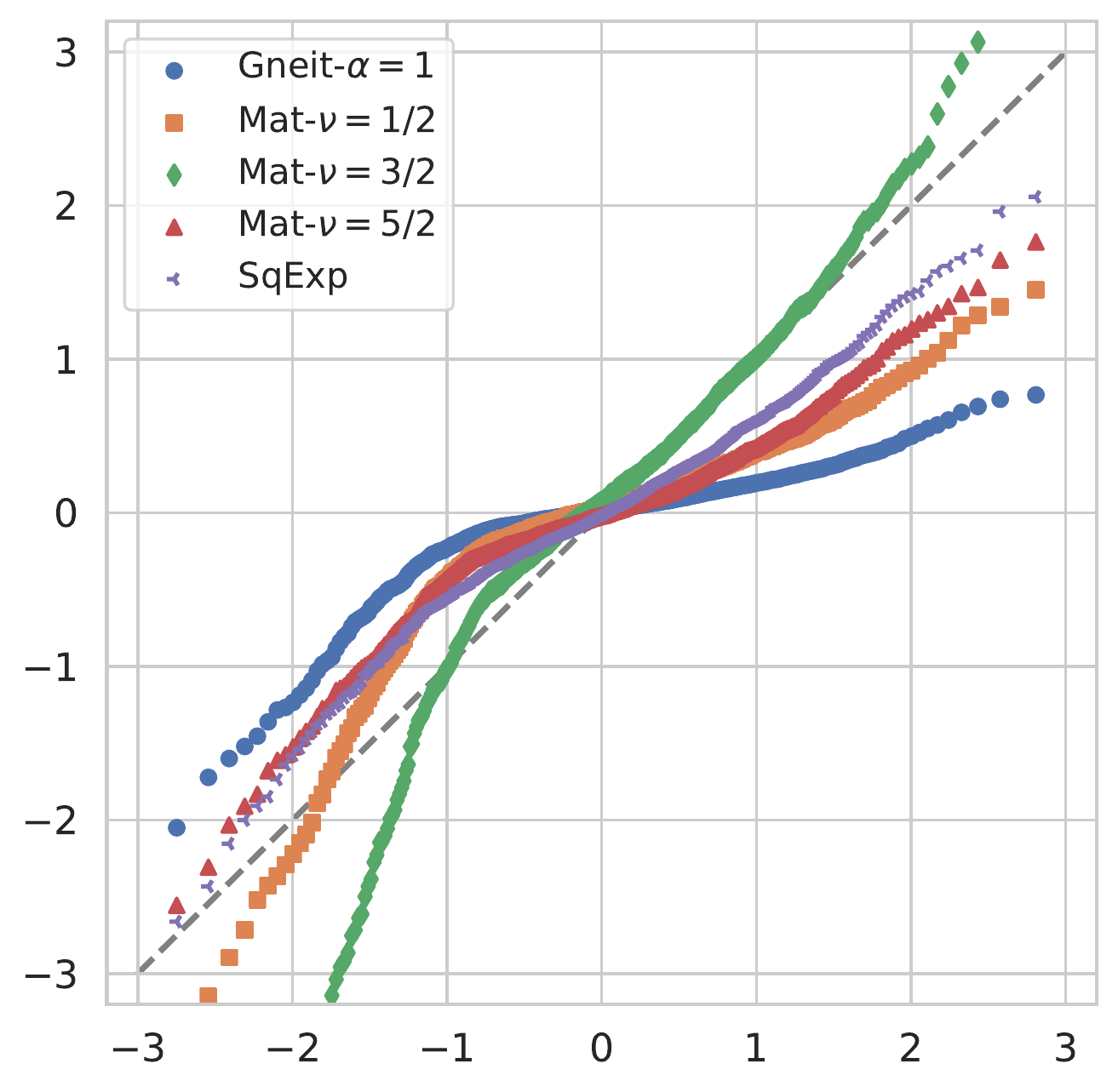}
  \end{minipage}\hfill
  \begin{minipage}[c]{0.6\textwidth}
    \caption{QQ-plot for predicted distributions on 2{,}000 held out
      integrated-$\rho$ test points $e_*$.  Theoretical normal quantiles are on
      the horizontal axis, and predictive $z$-scored quantiles are on the
      vertical axis. Posterior for extinctions $e$ are moderately well
      calibrated, with the squared exponential kernel having the best calibration. 
    } \label{fig:domain-qq}
  \end{minipage}
\end{figure}

\begin{table}[]
    \centering
    \scalebox{.8}{
        \begin{tabular}{lrrrrrr}
\toprule
{} &  Gneit-$\alpha=1$ &  Mat-$\nu=1/2$ &  Mat-$\nu=3/2$ &  Mat-$\nu=5/2$ &  SqExp &  $\mathcal{N}(0, 1)$ \\
$\sigma$ &                   &                &                &                &        &                      \\
\midrule
0.5      &             0.894 &          0.760 &          0.451 &          0.730 &  0.625 &                0.383 \\
1.0      &             0.965 &          0.900 &          0.677 &          0.900 &  0.869 &                0.683 \\
2.0      &             0.997 &          0.969 &          0.878 &          0.991 &  0.986 &                0.955 \\
3.0      &             1.000 &          0.994 &          0.947 &          1.000 &  0.999 &                0.997 \\
\bottomrule
\end{tabular}

    }
    \caption{Coverage fractions at various levels of z-score standard
    deviation $\sigma$ for test extinctions $e_*$. Compared with theoretical
    coverage fractions, we see that our inferences generate predictions that are well calibrated for the smoothest kernels}
    \label{tab:domain-coverage}
\end{table}

\begin{figure}[t!]
\centering
\begin{subfigure}[b]{.32\textwidth}
\centering
\includegraphics[width=\textwidth]{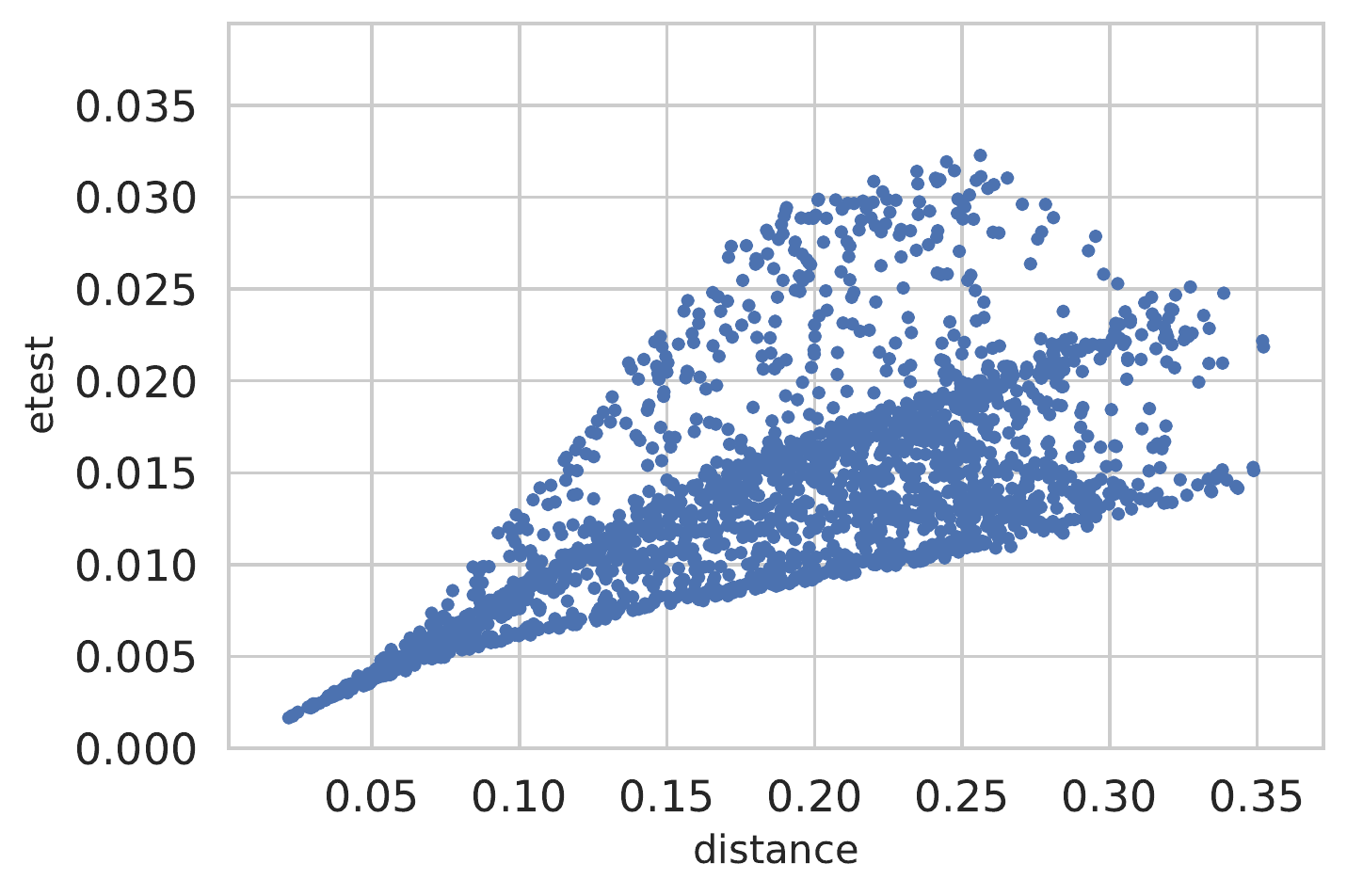}
\caption{Test $e$ values by distance}
\label{fig:domain-e-value-by-distance}
\end{subfigure}
~
\begin{subfigure}[b]{.32\textwidth}
\centering
\includegraphics[width=\textwidth]{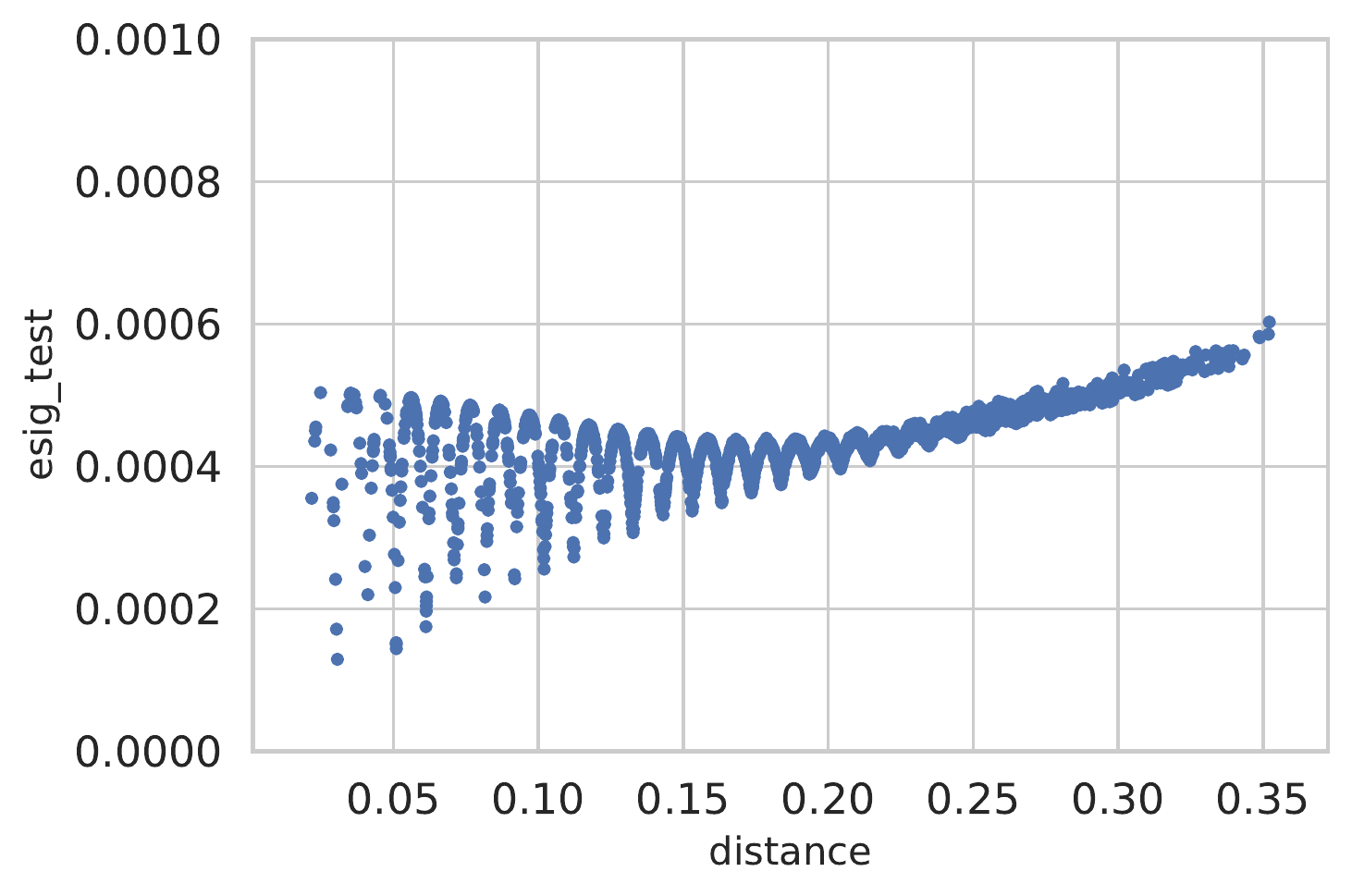}
\caption{Test $e$ posterior standard deviation}
\label{fig:domain-e-posterior-variance-by-distance}
\end{subfigure}
~
\begin{subfigure}[b]{.32\textwidth}
\centering
\includegraphics[width=\textwidth]{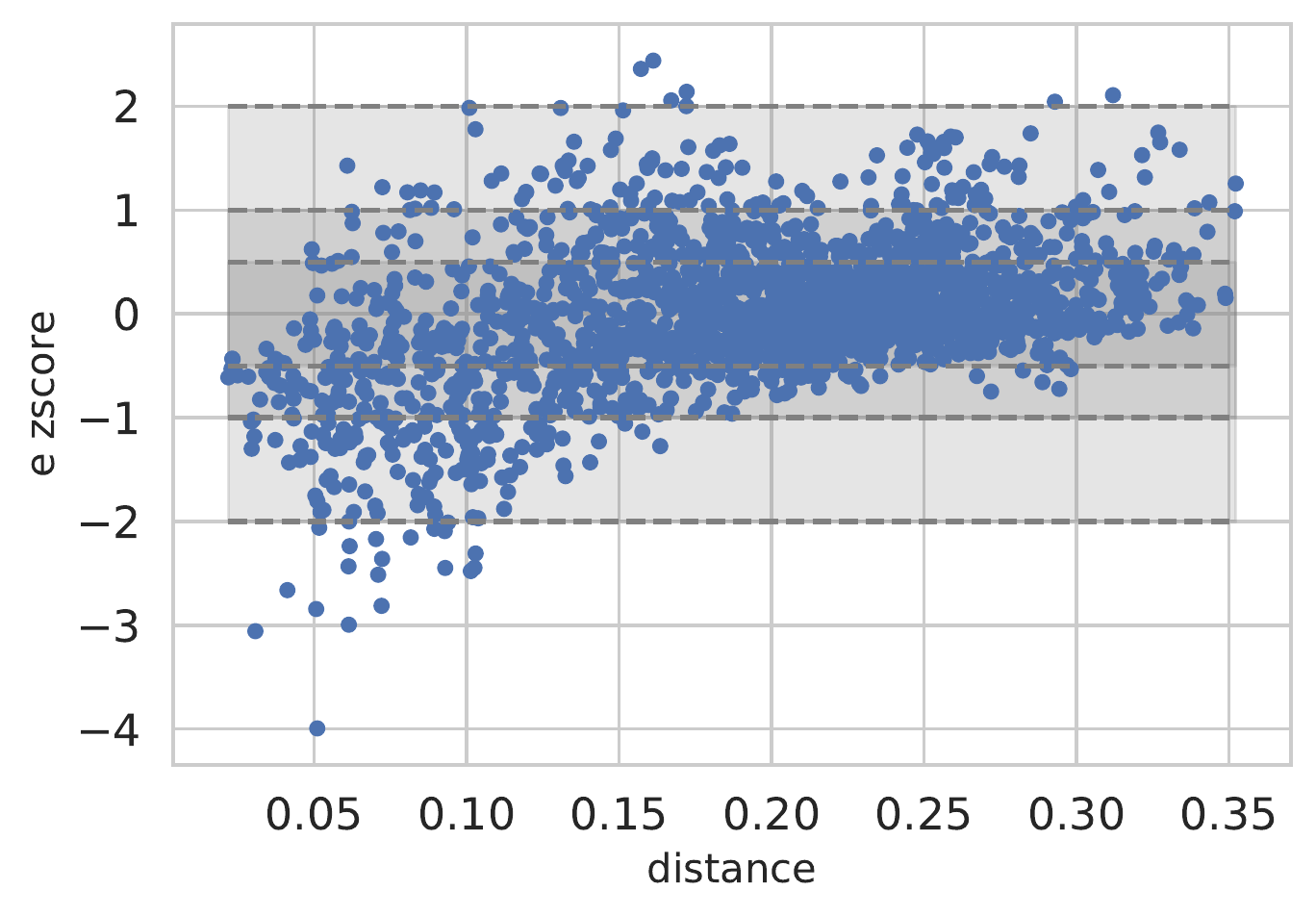}
\caption{Test $e$ z scores by distance}
\label{fig:domain-e-zscore-by-distance}
\end{subfigure}
~
\caption{Posterior summaries as a function of distance to the origin (observation point).  The noise variance for the extinction estimates is $\sigma_n = 0.005$, so Figure a demonstrates our low signal to noise regime, Figure b demonstrates our shrinkage by a factor of 10, and Figure c demonstrates our decent calibration in this very low signal to noise regime. 
}
\label{fig:domain-summary-by-distance}
\end{figure}

\section{Discussion}
\label{sec:conclusion}

Mapping the three-dimensional distribution of dust in the Milky Way is
fundamental to astronomy.  Star dust obscures our observations of star
light, and so an accurate dust map would help us more clearly observe
the universe.  Such a dust map would be a new lens into
the dynamics of the Milky Way, providing a trace of the process by
which it formed.

In this paper, we developed \ziggy, a method to estimate the dust map
from millions of astronomical observations.  \ziggy incorporates noisy
information into a single, coherent spatial model.  We developed
technical innovations to accommodate integrated observations into the
Gaussian process framework, and to scale such a model to millions of
observations without spatial discretization.  We validated these
algorithmic innovations with numerical studies, and measured the
performance of \ziggy in both a synthetic setting and a realistic,
high-fidelity dataset used in astronomical study.  We showed that
\ziggy is able to incorporate millions of observations and recover
accurate and well-calibrated estimates of both stellar extinctions and
pointwise values of the three-dimensional dust map.

There are several avenues for future research and improvement.  \ziggy
assumes that distances are known, when in reality they are noisy.  In
practical applications, we will only include stars with precise
distance measurements, lowering the spatial density of observations,
which may degrade estimates.  We view \ziggy as a step toward a
probabilistic model of photometric measurements that jointly infers
brightnesses, colors, distance, and spatial dust.  Jointly modeling
brightnesses, distances, and a spatial dust map using a Bayesian
hierarchical model will allow us to leverage the coherence of Bayesian
probabilistic modeling to form tighter estimates of all modeled
quantities.

\added[id=AM]{Additionally, a limitation of our treatment 
is the Gaussian error assumption.
Non-Gaussian observation errors could degrade extinction 
estimates --- e.g., forcing $\rho$ to be more complex
(i.e., with a smaller length scale) than it truly is.
A flexible hierarchical Bayesian model could
incorporate a non-Gaussian likelihood that could be more robust to 
heavier tailed observation error. }

A full map that spans the entirety of the Milky Way will require
increasing the capacity of the approximation.  To reconstruct both
global and local features using an inducing point method, we need to
ensure that small length scales (relative to the input domain) can be
resolved. When inducing points are spatially distant from one another,
the inducing point approximation will not have the capacity to
represent sharp changes in $\rho(x)$ at small scales.  This limitation
can be overcome by including more, and closer, inducing points.
However, introducing more inducing points butts up against the
$O(M^3)$ computational limitations.  In this work, we assume we are
able to populate the space with enough inducing points; but the
scaling concern is an important avenue of future work.  Incorporating
more inducing points will be crucial to resolving both global features
and fine local features within the Milky Way.

The ultimate goal of this line of work is to produce an accurate
catalog of the properties of stars, such as probabilistic brightnesses
and distances, and the focus of scale is motivated by the size of
modern photometric catalogs.  PAN-STARRS 1 catalog includes 2.4
million detected stars \citep{flewelling2016panstarrs}, the fifteenth
data release of the Sloan Digital Sky Survey includes photometry for
over 260 million detected stars \citep{aguado2019fifteenth}, and the
second GAIA data release has 1.3 billion stars with parallax
measurements that can be used to estimate the interstellar dust
distribution \citep{brown2018gaia}.  As a final avenue of research, we
will scale \ziggy to incorporate billions of observations in a much
larger spatial domain.

\bibliographystyle{apa}
\bibliography{refs.bib}

\clearpage
\appendix
\begin{center}
  \LARGE \textbf{Mapping Interstellar Dust with Gaussian
    Processes: Supplementary Material}
\end{center}

\section{Covariance Function Definitions}
\label{app:covariance-function-defs}

We compare three families of covariance functions in this work: (i) squared exponential, (ii) \matern, (ii) and a kernel from Gneiting \citep{gneiting2002compactly}.  Each are parameterized with an amplitude scale $\sigma^2$ and length scale parameter $\ell$.

The squared exponential kernel is defined
\begin{align}
    k(x, y) &= \sigma^2 \exp\left(-\frac{1}{2\ell^2}|x-y|_2^2 \right)
    \label{eq:sqexp-def}
\end{align}

The Matern class of kernels 
\begin{align}
    k(x, y) &= \sigma^2
      \frac{2^{1-\nu}}{\Gamma(\nu)} 
      \left(\sqrt{2\nu} \frac{d}{\ell} \right)^{\nu}
      K_{\nu}\left(\sqrt{2\nu}\frac{d}{\ell} \right)
    \label{eq:matern-def}
\end{align}

And the Gneiting kernel is defined
\begin{align}
    k(\tau) &= \sigma^2
    (1 + \tau^\alpha)^{-3} \left( (1 - \tau) \text{cos}(\tau \cdot \pi) + (1/\pi) \text{sin}(\tau \cdot \pi) \right) 
    \label{eq:gneit-def}
\end{align}
for $\tau = |x - y| / \ell$, $\tau < 1$. 

We graphically depict these covariance functions in Figure~\ref{fig:kernel-comparison}. 

\section{Integrated Covariance Functions}
\label{app:integrated-covariance}

Below is a simple derivation of the semi-integrated covariance function. 

\begin{claim}[Semi-integrated Covariance Function]
\label{claim:semi-integrated}
The covariance between a process value $\rho_i \triangleq \rho(x_i)$
and an integrated value $e_j \triangleq \int_{x \in R_j} \rho(x)dx$ takes
the form 
\begin{align}
  \text{Cov}(\rho_i, e_j) &= \int_{x \in R_j} k(x_i, x) dx \\
  &= k^{(semi)}(x_i, x_j) \, .
  \label{eq:k-semi-def}
\end{align}
For consistency, we will write the integrated argument second.
\end{claim}
\begin{proof}
Without loss of generality consider a mean zero process $\rho$. 
For positive and integrable covariance functions $k(\cdot,\cdot)$ (e.g. the 
Mat\'ern and squared exponential) the double integral will be finite.
By Fubini's theorem we can reverse the order of integration, yielding
\begin{align*}
\text{Cov}(\rho_i, e_j) 
&= \mathbb{E}\left[\rho_i \int_{x \in R_j} \rho(x)dx \right] \\
&= \mathbb{E}\left[ \int_{x\in R_j} \rho_i \rho(x)dx \right] \\
&= \int_{x \in R_j} \mathbb{E}\left[ \rho_i \rho(x) \right] dx \\
&= \int_{x \in R_j} \text{Cov}(\rho_i, \rho(x)) dx 
\, = \int_{x \in R_j} k(x_i, x) dx \,.
\end{align*}
\end{proof}

\begin{figure}[t]
    \centering
    \includegraphics[width=.5\textwidth]{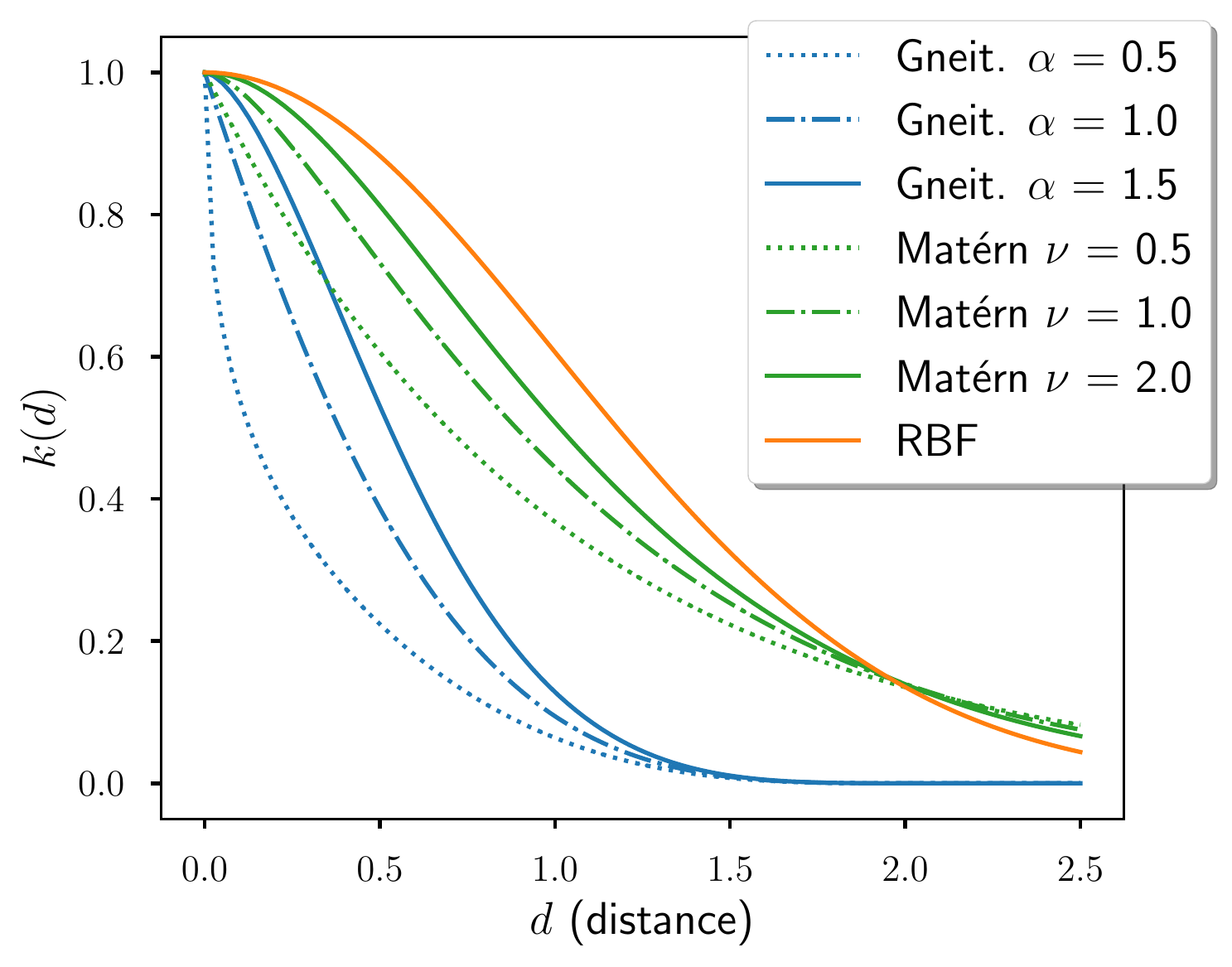}
    \caption{Correlation function for squared exponential, \matern, and Gneiting kernels.  Here the squared exponential and \matern~kernels have $\sigma^2=1$ and $\ell=1$; the Gneiting kernel has $\ell=2$.  Note that the Gneiting kernel goes to zero above $\ell$, while the other two kernels.  Further, the steepness of the function at 0 determines how smooth the GP sample paths will be --- the squared exponential kernel is the smoothest, while the Gneiting $\alpha=.5$ is the most jagged.}
    \label{fig:kernel-comparison}
\end{figure}

\subsection{Semi-integrated Squared exponential kernel}
The \emph{squared exponential} covariance function is commonly used
for Gaussian process regression \citep{rasmussen2006gaussian}.  For
$D$-dimensional inputs $x_i, x_j$, it is
\begin{align}
k_{sqe}(x_i, x_j)
  &= \sigma^2 \exp\left( -\frac{1}{2} (x_i - x_j)^\intercal
     L^{-1} (x_i - x_j) \right), \label{eq:sqe-kernel}
\end{align}
where $\btheta = (L, \sigma^2)$ are the $D \times D$-size length scale
matrix $L$ and the process marginal variance $\sigma^2$.  For an
isotropic Gaussian process, the length scale matrix $L$ is
$\ell \cdot I_D$ for a scalar length scale $\ell$.  The length scale
controls the smoothness of the function; smaller length scales create
more wiggly functions and longer length scales create more
slowly-changing functions.  This kernel puts support on functions
with infinitely differentiable sample paths.

The squared exponential kernel admits an analytically tractable form
for the semi-integrated kernel
\begin{align}
  k_{\mathrm{sqe}}^{(\mathrm{semi})}(x_i, x_j)
    &= \int_{[0,1]} k_{\mathrm{sqe}}(x_i, \alpha\cdot x_j) ||x_j||_2 d\alpha \\
    &= \frac{1}{2} ||x_j||_2 \sigma^2 \sqrt{\frac{2\pi}{a}}
      \exp\left(\frac{b^2}{2a} - \frac{c}{2}\right)
      \left[
        \mathrm{erf} \left(\frac{1-b/a}{\sqrt{2/a}} \right) -
        \mathrm{erf} \left( \frac{-b/a}{\sqrt{2/a}} \right)
      \right]
    \label{eq:sqe-semi}
\end{align}
where $\text{erf}(\cdot)$ is the error function, and $a$, $b$, and $c$ are defined
\begin{align*}
  a = x_j^\intercal L^{-1} x_j \,,\quad\quad
  b = x_i^\intercal L^{-1} x_j \,,\quad\quad
  c = x_i^\intercal L^{-1} x_i
\end{align*}
and the origin, or observer, is at 0.

The doubly-integrated squared exponential covariance function is not
as straightforward to express.  Deriving the doubly-integrated squared
exponential covariance function requires integrating the expression in
\Cref{eq:sqe-semi}.  Integrating the error function multiplied by a
quadratic requires either numerical approximation or a more
sophisticated approach \citep{fayed2014evaluation}.  We note that,
while it is straightforward to analytically derive the integral of the
squared exponential kernel over the \emph{unbounded region}
$[-\infty, \infty ]^2$ (e.g.~to compute a Gaussian normalizing
constant), it is not straightforward to integrate over two-dimensional
compact sets, e.g.~$[0, x_i]\times[0,x_j]$ --- this is exactly what is
required to compute the doubly-integrated kernel.

\begin{algorithm}[t]
  \BlankLine
  \KwData{
    data set $\mathcal{D} \triangleq \{ x_n, \sigma_n, a_n
    \}_{n=1}^N$,
    test location $x_*$,
    covariance functions: standard $k(\cdot, \cdot)$,
      semi-integrated $k^{(\textrm{semi})}(\cdot, \cdot)$,
      and doubly-integrated    $k^{(\textrm{double})}(\cdot, \cdot)$
  }
  \BlankLine
  \KwResult{
    $\mu_*$, $\sigma_*^2$ for posterior
    $p(\rho(x_*) | \mathcal{D}) = \mathcal{N}(\mu_*, \sigma_*^2)$
  }
  \BlankLine
$(\bK_{\be})_{i,j} \leftarrow k^{(\textrm{double})}( x_i, x_j )$
  for $i,j = 1, \dots, N$ \tcp*{doubly-integrated Gram matrix}
$(\bK_{*, \be})_{*, j} \leftarrow k^{(\textrm{semi})}(x_*, x_j)$
  for $j = 1, \dots, N$ \tcp*{semi-integrated cross covariance}
$K_{*,*} \leftarrow k(x_*, x_*)$ 
  \tcp*{prior variance from standard covariance function}
$\Sigma \leftarrow \text{diag}(\sigma_1^2, \dots, \sigma_N^2)$ \\
$\mu_* \leftarrow \bK_{*, \be}(\bK_{\be} + \Sigma)^{-1} \ba$ \\
$\sigma_*^2 \leftarrow K_{*} - \bK_{*,\be} (\bK_{\be} + \Sigma)^{-1}
    \bK_{*,\be}^\intercal$ \\
\Return{$\mu_*,\sigma_*^2$}
\caption{Exact Gaussian process inference with noisy integrated
observations.}
\label{alg:integrated-gp-inference}
\end{algorithm}

\section{Stochastic Natural Gradients for SVGPs}
\label{sec:natural-gradients}
Following \citet{hensman2013gaussian} and \citet{hoffman2013stochastic},
the natural gradient for
variational parameters $\blambda = \bm_{\blambda}, \bS_{\blambda}$ is
straightforward to express using the natural parameterization of the
multivariate normal $q_{\blambda}$
\begin{align}
        \etab_1 = \bS_{\blambda}^{-1} \bm_{\blambda} \,\,\, , \quad
        \etab_2 = -\frac{1}{2} \bS_{\lambda}^{-1} \,\,\, .
\end{align}
The natural gradient of $\mathcal{L}$ with respect to $\etab_1$ and $\etab_2$ can be written
\begin{align}
        \tilde{\nabla}_{\etab_2}\mathcal{L} &= -\frac{1}{2}\Lambda + \etab_2 \\
        \tilde{\nabla}_{\etab_1}\mathcal{L} &= -\frac{1}{2} \bK_{\bu,\bu}^{-1} \tilde{\bK}_{\bu,\brho} \by  + \etab_1
        \label{eq:natgrad}
\end{align}
where we define
\begin{align}
  \Lambda &= \bK_{\bu,\bu}^{-1} \tilde{\bK}_{\bu,\brho} \tilde{\bK}_{\brho,\bu} \bK_{\bu,\bu}^{-1} + \bK_{\bu,\bu}^{-1} \\
  &= \sum_{n} \bK_{\bu,\bu}^{-1} \left( \tilde{\bK}_{\bu,n} \tilde{\bK}_{n,\bu} \right) \bK_{\bu,\bu}^{-1} + \bK_{\bu,\bu}^{-1} \,,
  \label{eq:cov-grad}
\end{align}
and
\begin{align}
  \tilde{\bK}_{\brho,\bu} &\triangleq \Sigma^{-1/2} \bK_{\brho,\bu} \\
  (\bK_{\brho,\bu})_{n,m} &= k(x_n, \bar{x}_m) \\
  \Sigma &\triangleq \text{diag}(\sigma_1^2, \dots, \sigma_N^2) \, .
\end{align}
Note that Equation~\ref{eq:cov-grad} shows that the term $\Lambda$ can be
decomposed into a sum of $N$ data-specific terms.  Further, the gradient
for $\etab_1$ includes the term $\tilde{\bK}_{\bu,\brho}\by = \sum_{n}
\tilde{\bK}_{\bu,n} y_n $ that also decomposes into a sum over the $N$
data-specific terms.  These two decompositions enable us to compute
unbiased mini-batched natural gradients and optimize the ELBO using
stochastic gradient ascent.
An additional benefit of using the natural gradient is that a stochastic
gradient update to $\etab_2$ will stay in the positive semi-definite cone.
This allows us to directly parameterize $\etab_2$ and not its Cholesky
factor or symmetric square root, simplifying our implementation.

\section{Whitened Parameterization Gradients}
\label{sec:whitened-gradients}
The relationship between the whitened model over $\bz$ and the standard
model over $\bu$ implies the following relationship between the
Gaussian variational parameters
\begin{align}
    \bm &= \bL \tilde{\bm} \\
    \bS &= \bL \tilde{\bS} \bL^\intercal.
    \label{eq:white-to-standard}
\end{align}

Similarly, the natural parameters have the relationship
\begin{align}
\etab_1 &= \bS^{-1}\bm
    = (\bL^\intercal)^{-1} \tilde{\bS}^{-1} \bL^{-1} \bL \tilde{\bm} = (\bL^\intercal)^{-1} \tilde{\etab}_1 \\
\etab_2 &= -\frac{1}{2}\bS^{-1}
     = -\frac{1}{2} (\bL^\intercal)^{-1} \tilde{\bS}^{-1} \bL^{-1}
     = (\bL^\intercal)^{-1} \tilde{\etab}_2 \bL^{-1}.
\end{align}

The natural gradient is invariant to parameterization
\citep{amari1982differential}, which implies that the whitened natural
gradient transformation can be derived by multiplying the standard update
by $\bL^\intercal$
\begin{align}
    \etab_1^{(t+1)} &= \etab_1^{(t)} + (\bmu - \etab_1^{(t)}) & \text{ natural gradient updates }\\
    \etab_2^{(t+1)} &= \etab_2^{(t)} + (-\frac{1}{2}\Lambda - \etab_1^{(t)}) \\
    \implies
    \tilde{\etab}_1^{(t+1)} &= \tilde{\etab}_1^{(t)} + (\bL^\intercal \bmu - \tilde{\etab}_1^{(t)}) & \text{ whitened natural gradient updates }\\
    \tilde{\etab}_2^{(t+1)} &= \tilde{\etab}_2^{(t)} + (-\frac{1}{2}
                              \bL^\intercal \Lambda \bL -
                              \tilde{\etab}_2^{(t)}).
\end{align}
This yields a natural interpretation --- the whitened sufficient
statistics must undo the inverse covariance structure from the prior
before updating the variational parameters.\footnote{Note that the
  natural parameters are in ``information form'' --- pre- and
  post-multiplying by the covariance Cholesky will whiten a precision
  matrix.}

When it comes time to predict, we will simply transform the whitened
parameters into standard parameters using
Equation~\ref{eq:white-to-standard}, and proceed with the usual
multivariate Gaussian conditioning.  In
Section~\ref{sec:whitening-comparison} we show that optimizing in the
whitened space can yield better tuned covariance function parameters.  This
empirical evidence supports the practical conclusion that jointly tuning
the variational approximation over $\bz$ and $\btheta$ is an easier
optimization problem than jointly tuning parameters over $\bu$ and
$\btheta$ due to the stronger dependence between $\bu$ and $\btheta$.

\section{Scaling Integrated Observation GPs: Algorithm Details}

\subsection{Monte Carlo estimators for semi-integrated kernels}
\label{sec:mc-semi-estimators}

For tractable semi-integrated covariance functions (e.g.~the squared
exponential kernel) we can easily substitute a closed form computation for
$\bK_{\brho,\bu}$ in Equation~\ref{eq:natgrad}.  However, when the
semi-integrated covariance term is intractable, directly using natural
gradient updates becomes complicated.  One approach is to use numeric
integration --- the semi-integrated covariance is a one-dimensional
integral that can be approximated with quadrature techniques.  However,
when we have $N$ and $M$ very large, solving $N \times M$ numerical
integration problems to populate the cross covariance matrix can be too
expensive to be practical.  

Stochastic natural gradient updates employ computationally cheap estimates
of the true gradient to solve the original optimization problem.  The
operative question is, how \emph{good} do these estimates have to be to
effectively find the optimal solution? 

Similarly, numeric integration approximations to $k^{(semi)}(\cdot, \cdot)$
are precise to nearly machine precision, but are computationally expensive.
Analogously we can ask, how \emph{good} do the semi-integrated covariance
approximations have to be to effectively find the optimal solution?

Pursuing this idea, we propose using a Monte Carlo approximation to the
semi-integrated covariance values --- we sample uniformly along the
integral path and average the original covariance values.  More formally,
we introduce a uniform random variable $\nu_{R_n}$ that takes values along
the ray $R_n$ from the origin to $x_n$. %
We can now write the semi-integrated covariance as  
\begin{align}
    k^{(semi)}(x_m, x_n) &= \int_{x \in R_n} k(x_m, x)dx \\
    &= |R_n| \int_{\nu \in R_n} \underbrace{\frac{1}{|R_n|}}_{= p(\nu)} k(x_m, x)dx \\
    &= |R_n| \, \mathbb{E}_{\nu}\left[ k(x_m, \nu) \right] \, . 
\end{align}
This form admits a simple Monte Carlo estimator
\begin{align}
    \nu^{(1)}, \dots, \nu^{(L)} &\sim Unif(R_n) \\
    \hat{k}^{(semi)}(x_m, x_n) &= \frac{|R_n|}{L}\sum_{\ell} k(x_m, \nu^{(\ell)}) \,\, .
    \label{eq:mc-semi-appendix}
\end{align}
It is straightforward to show that Equation~\ref{eq:mc-semi-appendix} is an
unbiased and consistent estimator for the true semi-integrated covariance
value evaluated at $x_m$ and $x_n$.

\parhead{Uniformly distributed grids}
The unbiased Monte Carlo estimator $\hat{k}^{(semi)}$ will have some
variance about the mean.  The quality of an unbiased estimator is often
measured by how low its variance is --- the more precise the better.
Independent and identically distributed samples $\nu^{(\ell)}$ will yield
the correct expectation, but may be lower variance than a more clever
sampling scheme.  For our application, we uniformly sample along the ray
that traces from the origin $O$ to the observation location $x_n$.
Intuitively, drawing two uniform random variables that are very close to
one another will carry redundant information (for smooth functions
$k(\cdot, \cdot)$).  Spreading our samples out to cover more of the length
of the ray may yield more precise estimates, provided that we can construct
a sample that yields the correct expectation.

One inexpensive way to generate a set of identically distributed (but not
independent samples) is to lay down a random grid of points.  Consider the
sampling procedure for $L$ samples along the unit interval
\begin{align}
    \nu^{(1)} &\sim Unif(0, 1/L) \\
    \nu^{(\ell)} &= \nu^{(\ell-1)} + \tfrac{1}{L} \quad \text{ for } \ell = 2, \dots, L \, .
\end{align}
This sampling scheme will generate a set of $L$ correlated samples that
constructs a randomly placed grid within the unit interval.  If we randomly
permute the variables $\nu_1, \dots, \nu_L$ (or, equivalently, randomly
select one of them) then the marginal distribution of the random variable
will be $Unif(0, 1)$ --- this ensures that the expectation of the estimator
in Equation~\ref{eq:mc-semi} will remain unbiased.  However, because the
samples are anti-correlated, the variance of the estimator can be
substantially reduced. 

\parhead{Biased gradient estimates}
The Monte Carlo estimate of $\hat{k}^{(semi)}(x_m, x_n)$ is unbiased, but
non-linear functions of an unbiased estimator, in general, may jettison
this property.  In our application, the goal is to compute cheap, unbiased
estimates of the gradient of the variational parameters $\bm, \bS$ (or
their information form parameterization $\etab_1, \etab_2$). 

For a single term $n$, the gradient for $\eta_1$ is a linear function of
the semi-integrated covariance vector $\bK_{n,\bu} = \left[ k^{(semi)}(x_n,
x_1), \dots, k^{(semi)}(x_n, x_M) \right]$.\footnote{Note that the order of
the subscripts indicates the vector dimensions, e.g.~$\bK_{n,\bu} \in
\mathbb{R}^{1 \times M}$ and $\bK_{n,\bu}^\intercal = \bK_{\bu,n} \in
\mathbb{R}^{M \times 1}$.} Due to this linearity, plugging in an unbiased
estimator, $\hat{\bK}_{n,\brho}$, for $\bK_{n,\bu}$ will form an unbiased
gradient estimate for $\tilde{\nabla}_{\eta_1} \mathcal{L}$.  However,
plugging in $\hat{\bK}_{n,\bu}$ into the gradient term for $\eta_2$ will
form a non-linear function of an unbiased estimator, resulting in a biased
estimate of $\tilde{\nabla}_{\eta_2}$. 

To see this, consider the vector of unbiased estimates of the
semi-integrated covariance 
\begin{align}
    \hat{\bK}_{n,\brho} &= \left[ \hat{k}^{(semi)}(x_1, x_n), \dots, \hat{k}^{(semi)}(x_m, x_n) \right] \, .
\end{align}

The stochastic component of the gradient for $\eta_2$ can be written 
\begin{align}
    \mathbb{E}\left[ \bK_{\bu,\bu}^{-1} \bK_{\bu,n} \bK_{n,\bu} \bK_{\bu,\bu}^{-1} \right]
    &= \bK_{\bu,\bu}^{-1} \mathbb{E}\left[\bK_{\bu,n} \bK_{n,\bu} \right] \bK_{\bu,\bu}^{-1}
\end{align}
where $\bK_{\bu,\bu}^{-1}$ can be compute exactly, leaving the expectation
of the outer product of the semi-integrated covariance terms
$\mathbb{E}\left[ \bK_{\bu,n} \bK_{n,\bu} \right]$ Given an unbiased
estimator for $\bK_{n,\bu}$ the naive plugin estimate for $\bK_{n,\bu}
\bK_{\bu,n}$ is
\begin{align}
    \hat{O}^{(plugin)} = \hat{\bK}_{\bu,n}\hat{\bK}_{n,\bu} 
    \label{eq:naive-plugin-estimator}
\end{align}
which is \emph{biased} in general
\begin{align}
\mathbb{E}\left[ \hat{\bK}_{\bu,n}\hat{\bK}_{n,\bu} \right ] 
    &\neq \mathbb{E}\left[ \hat{\bK}_{\bu,n} \right] \mathbb{E}\left[\hat{\bK}_{n,\bu}\right] = \bK_{n,\bu} \bK_{n,\bu}^\intercal \, .
\end{align}
Fortunately, we can construct an unbiased estimator for this outer product.
Note that the outer product of an expectation can be expressed as 
\begin{align}
\mathbb{E}\left[\bx \right] \mathbb{E}\left[ \bx \right]^\intercal
  &= \mathbb{E}\left[ \bx \bx^\intercal \right] - \text{Cov}(\bx)
\end{align}
Given an $L$-size sample $\bx_1, \dots, \bx_L$ we can compute unbiased
estimators for $\mu = \mathbb{E}[\bx]$, $M \triangleq \mathbb{E}\left[ \bx
\bx^\intercal \right]$, and $\Sigma \triangleq \text{Cov}(\bx)$
\begin{align}
    \hat{\mu} &= \frac{1}{L} \sum_{\ell} \bx_\ell \\
    \hat{\Sigma} &= \frac{1}{L-1} \sum_{\ell} (\bx_{\ell} - \hat{\mu})(\bx_{\ell} - \hat{\mu})^\intercal \\
    \hat{M} &= \frac{1}{L} \sum_{\ell} \bx_\ell \bx_\ell^{\intercal} 
\end{align}
With these estimators, $\hat{M} - \hat{\Sigma}$ forms an unbiased estimate
for $\mathbb{E}[\bx]\mathbb{E}[\bx]^\intercal$
\begin{align}
    \mathbb{E}\left[ \hat{M} - \hat{\Sigma} \right] &= \mathbb{E}[\bx]\mathbb{E}[\bx]^\intercal \, .
\end{align}
We can use this fact to construct a de-biased estimator for $O$
\begin{align}
    \hat{O}^{(unbiased)} &= \hat{M}\left( \{\hat{\bK}^{(\ell)}_{n,\bu}\}_{\ell=1}^L \right) - \hat{\Sigma}\left( \{\hat{\bK}^{(\ell)}_{n,\bu}\}_{\ell=1}^L \right) \, .
\end{align}
Though $\hat{M} - \hat{\Sigma}$ forms an unbiased and symmetric estimator
of the outer product, this estimator may not be positive semi-definite.  If
this estimator is not positive semi-definite, then a gradient step along
this direction may take the optimization routine outside of the PSD cone.
To ensure safe gradient updates, we will have to meticulously track the
magnitude of the lowest eigenvalues of this matrix. 

Though biased, the naive plugin estimator in
Equation~\ref{eq:naive-plugin-estimator} will always be non-negative
definite, resulting in gradient updates that are guaranteed to remain in
the PSD cone.  Proceeding with this estimator, we must ask, how much bias
is too much bias?  We explore this question empirically in
Section~\ref{sec:synthetic-experiments}, comparing the optimization traces
of algorithms using various sample sizes $L$ to the exact
$k^{(semi)}(\cdot, \cdot)$ values in the squared exponential kernel case.
We see that, though biased, this gradient estimator can closely approximate
the performance of the exact semi-integrated kernel given a modest number
of samples ($\approx 20$).

\begin{algorithm}[t]
\KwData{
$\mathcal{D} \triangleq \{ a_n, x_n, \sigma^2_n \}_{n=1}^N$ (data),
$k^{(\btheta)}(\cdot, \cdot)$ (covariance function);
$\left(\blambda^{(0)}, \btheta^{(0)} \right)$ (initial values);
$k(\cdot, \cdot)$ (covariance function);
$(\Delta_{\blambda}, \Delta_{\btheta})$ (step-sizes);
$B$ (batch size);
$L$ (Monte Carlo semi-integrated estimator sample size);
}
\KwResult{$\blambda, \btheta$ variational and covariance function parameters }
$\hat{k}^{(dd)}(\cdot) \leftarrow \text{Interpolate-Doubly-Diag}(k^{(\btheta_0)}(\cdot, \cdot))$ \tcp*{construct diagonal interpolator}
\For{$t\gets1$ \KwTo $t_{max}$}{
    $\mathcal{D}_{t} \leftarrow \text{Batch}(\mathcal{D}, B)$ \tcp*{mini-batch of data}
    $\blambda^{(t)} \leftarrow \blambda^{(t-1)} + \Delta_{\blambda} \cdot \tilde{\nabla}^{(MC)}_{\blambda}\mathcal{L}(\blambda^{(t-1)}, \btheta^{(t-1)}, \mathcal{D}_t, L)$ \tcp*{MC natural grad update}
    $\btheta^{(t)} \leftarrow \btheta^{(t-1)} + \Delta_{\btheta} \cdot {\nabla}^{(MC)}_{\btheta}\mathcal{L}(\blambda^{(t-1)}, \btheta^{(t-1)}, \mathcal{D}_t, L, \hat{k}^{(dd)}(\cdot))$ \tcp*{MC kernel grad update}
    $\text{sgd-update}(\blambda^{(t)}, \tilde{\nabla}{\blambda}^{(MC)})$ \tcp*{gradient-based update}
    }
return $\blambda^{(t_{max})}, \btheta^{(t_{max})}$ \tcp*{ optimal parameters}
\caption{Scalable integrated observation GP inference with mini-batches and Monte Carlo semi-integrated covariance estimates.  The natural gradient estimator uses $L$ Monte Carlo samples to approximate the semi-integrated covariance and $B$ data observations to approximate the complete data objective. }
\label{alg:scalable-integrated-gp-inference}
\end{algorithm}

\parhead{Monte Carlo SVGP comparison}
\begin{figure}[t!]
\centering
\includegraphics[width=.75\textwidth]{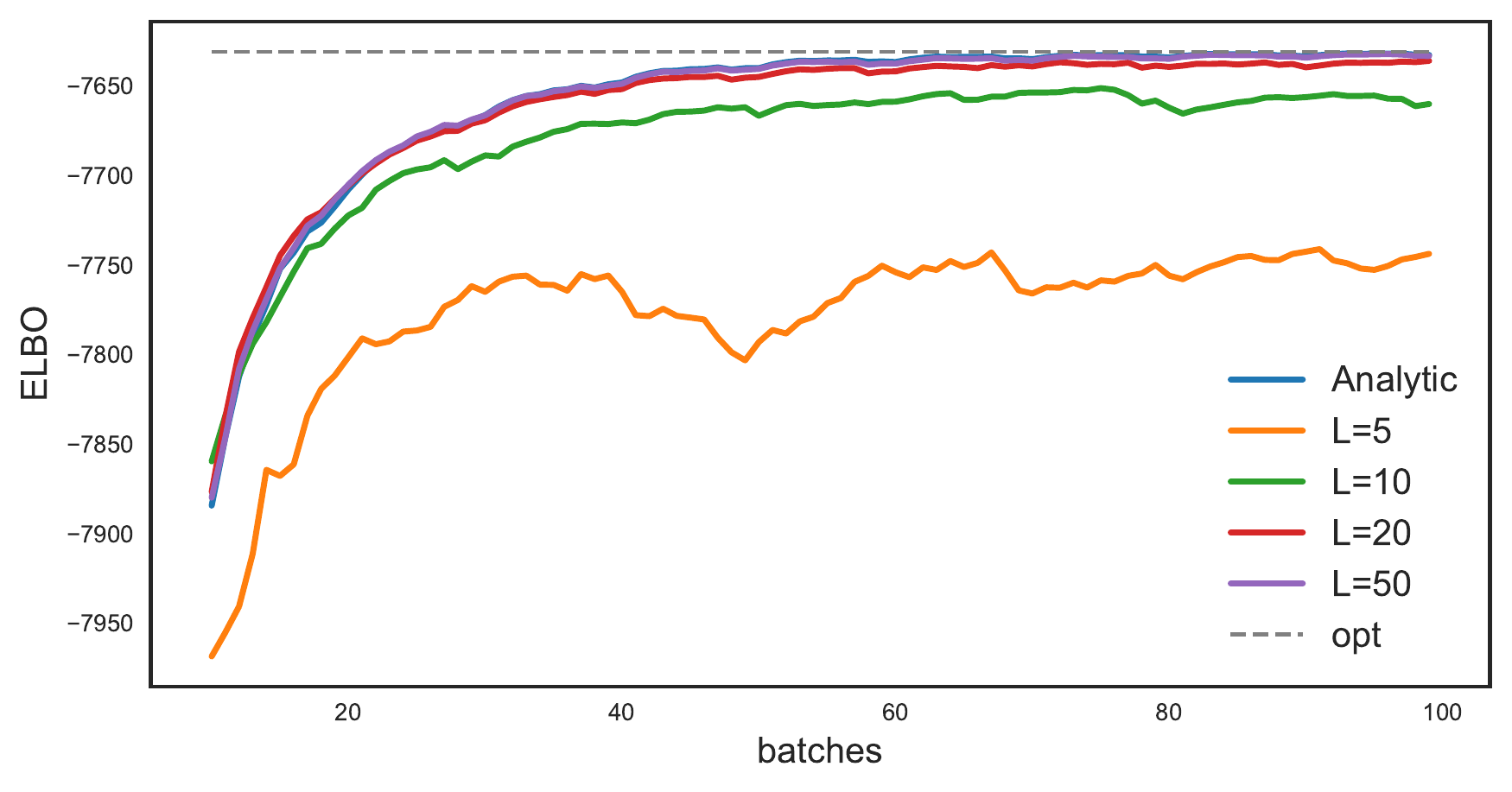}
\caption{Monte Carlo estimates of semi-integrated covariance functions enable efficient ELBO inference competitive with the analytic semi-integrated covariance. Above is a comparison of ELBO traces by mini-batch ($x$-axis) for the squared exponential kernel.  The analytic semi-integrated kernel was used in blue (hidden behind the $L=50$ purple line. The optimal value is depicted with the dotted gray line. }
\label{fig:optimization-comparison}
\end{figure}

The squared exponential kernel admits an analytic form for the
semi-integrated covariance function.  This gives us a way to directly
compare our Monte Carlo semi-integrated approximation within SVGP to the
same algorithm using exact semi-integrated covariance values.  In this
experiment, we compare the progress of the exact variational objective
throughout optimization using different gradient estimators --- the
analytic version, and the biased Monte Carlo gradients using sample size
$L=5, 7, 10, 20,$ and $50$.

We fit variational parameters using stochastic natural gradient descent
with mini-batches of size $1{,}000$ and step size $0.1$.  We fix the
covariance function parameters in this experiment to report a more direct
comparison between exact and Monte Carlo gradients --- in subsequent
experiments we apply gradient updates to these two parameters as well.
Inducing points are placed on a fixed grid of size $20x20$ ($M=400$),
tiling the 2-d space.

At each batch, we compute the complete-data ELBO objective (using the exact
squared exponential semi-integrated covariance function across methods).
Additionally, we compute the optimal value of the ELBO (by solving for the
optimal values of $\bm_{\blambda}$ and $\bS_{\blambda}$ in closed form ---
an operation that is too costly when $N$ is on the order of millions of
observations) which we depict alongside optimization traces. 

Figure~\ref{fig:optimization-comparison} compares the progress of
optimizers using different gradient estimators.  The ``analytic'' version
uses the closed form semi-integrated covariance, while the Monte Carlo runs
use sample sizes of varying values.  We see that with as few as $L=20$
samples the resulting optimization routine very nearly finds the optimal
value (depicted in grey), and almost exactly matches the path of the
``analytic'' routine.  This confirms that these Monte Carlo semi-integrated
covariance estimators are suitable plug-in estimators for fitting models
with covariance functions that do not admit an analytic form for the
semi-integrated covariance, (e.g.~the Gneiting and \matern~kernels). 

\subsection{Fast doubly-integrated kernel diagonal interpolation}
\label{sec:interpolated-doubly-kernels}
\parhead{Approximating doubly-integrated diagonal terms}
To make predictions within the SVGP framework, we need to compute the
marginal variance of each observation according to the Gaussian process
prior.  Concretely, this amounts to computing the $\bK_{*,*}$ term in
Equation~\ref{eq:inducing-diagonal-term}.  In the integrated observation
setting, this requires computing the variance of an integrated observation
\begin{align}
(\tilde \bK_{N})_{n,n} 
  = Cov(a_n, a_n)
  = Var(e_n) = Var\left( \int_{[0,1]} \rho(\alpha \cdot x_n)d\alpha \right)\,.
\end{align}
This marginal variance corresponds to the diagonal of the doubly-integrated
kernel and may not be analytically tractable.
For the squared exponential kernel, this would require integrating the
semi-integrated kernel expressed in Equation~\ref{eq:sqe-semi} over a
compact set, which does not yield a closed form \citep{fayed2014evaluation}.

Fortunately, for a one-dimensional compact set, this integral can be
numerically computed relatively efficiently.
Unfortunately, if we are scaling to millions of observations, we will
have to perform millions of numerical integrations to simply evaluate
the likelihood of the data under one setting of the covariance
function parameters.  Further, we will have to recompute all of these
numerical integrations at each iteration when fitting covariance function
parameters (or computing their posterior).

The inducing point variance correction method requires only computing
the marginal variance of each integrated observation under the prior.
For stationary kernels this is a function only of the distance of the
observation to the origin (or more generally, the length of the integrated ray).
Similar to the semi-integrated kernel, we can see that this is true by 
expanding the stationary kernel
\begin{align}
	\int_{0}^1 k(|x-\alpha x|) d\alpha 
	&= \int_{0}^1 k( (x^\intercal x  + \alpha^2 x^\intercal x - 2 \alpha x^\intercal x)^{1/2}) d\alpha \\
	&= \int_{0}^1 k( (|x|^2 + \alpha^2 |x|^2 - 2\alpha |x|^2)^{1/2} ) d \alpha \\
	&= f(|x|) \, .
\end{align}

Due to this stationarity, we can view each numerical computation as
evaluating the function $k_\theta(\bd_n) = \tilde{K}_nn$ where $\bd_n$
is the distance of observation $\bx_n$ from the origin. 
This perspective suggests a numerically cheap approximation --- simply
estimate this function along a fixed grid of distances (that spans the
range of observed distances) and interpolate the value for each observation.

This one-dimensional function can be estimated independent of the length
scale parameter for isotropic and stationary covariance functions.  The
effect of the length scale parameter is to effectively redefine distance,
or the magnitude of $|x|$.  We can account for this change by re-scaling
the $|x|$ distance and multiplying the result by $\ell$.  Specifically for
the diagonal, this becomes simply a function of a one-dimensional distance,
which can easily be interpolated with a small number of knots.

\parhead{Making predictions.}  To estimate the process value $\rho_*$
at new location $x_*$, we use the posterior predictive distribution
$p(\rho_* \given \mathcal{D})$.  Within the SVGP framework, we
approximate this distribution using the posterior $q_{\blambda}(\bu)$
\begin{align}
p(\rho_* \given \mathcal{D})
  &= \int p(\rho_* \given \bu) p(\bu \given \mathcal{D}) d\bu \\
  &\approx \int p(\rho_* \given \bu) q_{\blambda}(\bu) d\bu \,\, .
\end{align}
Under the approximation, $\bu$ and $\rho_*$ are jointly normal, making this
predictive distribution available in closed form
\begin{align}
    p(\rho_* \given \mathcal{D}) &= \mathcal{N}(\mu_*, \sigma^2_*) \\
    \mu_* &= \bK_{*,\bu} \bK_{\bu,\bu}^{-1} \bm \\
    \sigma_*^2 &= \tilde{\bK}_{*,*} +
        \bK_{*,\bu} \bK_{\bu,\bu}^{-1} \bS \bK_{\bu,\bu}^{-1} \bK_{\bu,*}
    \label{eq:predictions}
\end{align}
where
\begin{align}
  (\bK_{*,\bu})_m &= \text{Cov}(\rho_*, \bu_m) \\
  \bK_{*,*} &= \text{Var}(\rho_*) \\
  \tilde{\bK}_{*,*} &\triangleq
    \bK_{*,*} - \bK_{*,\bu} \bK_{\bu,\bu}^{-1} \bK_{\bu,*}
    \label{eq:inducing-diagonal-term}\,.
\end{align}
Prediction for integrated process values,
${e_*=\int_{x \in R_*} \rho(x)dx}$, can be computed similarly by
substituting semi-integrated covariance values into the vector
$\bK_{*,\bu}$ and the doubly integrated covariance value $\bK_{*,*}$.

\begin{figure}[t!]
\centering
\begin{subfigure}[b]{.43\textwidth}
\centering
\includegraphics[width=\textwidth]{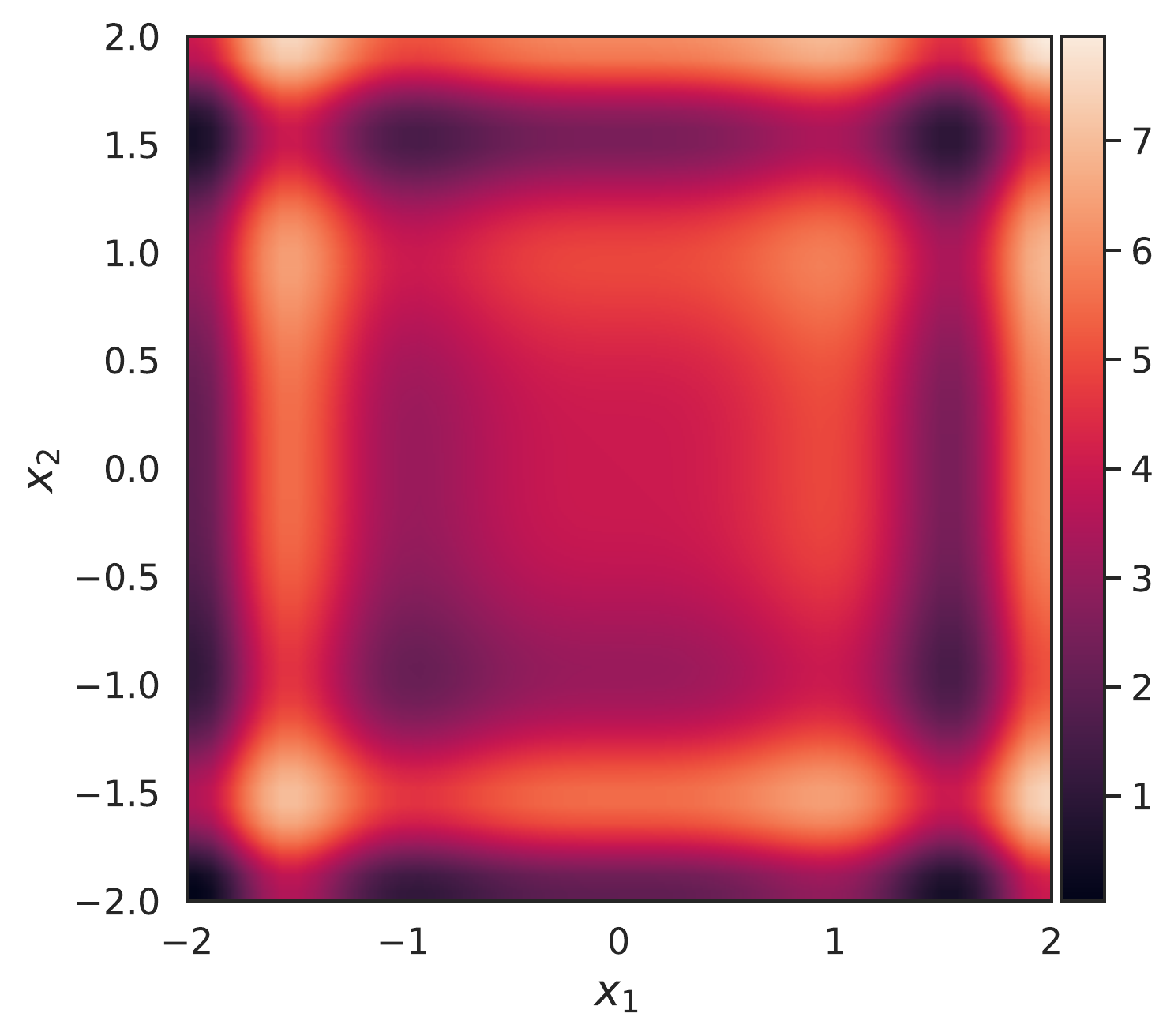}
\caption{True $\rho(x)$}
\label{fig:rho_true}
\end{subfigure}
~
\begin{subfigure}[b]{.43\textwidth}
\centering
\includegraphics[width=\textwidth]{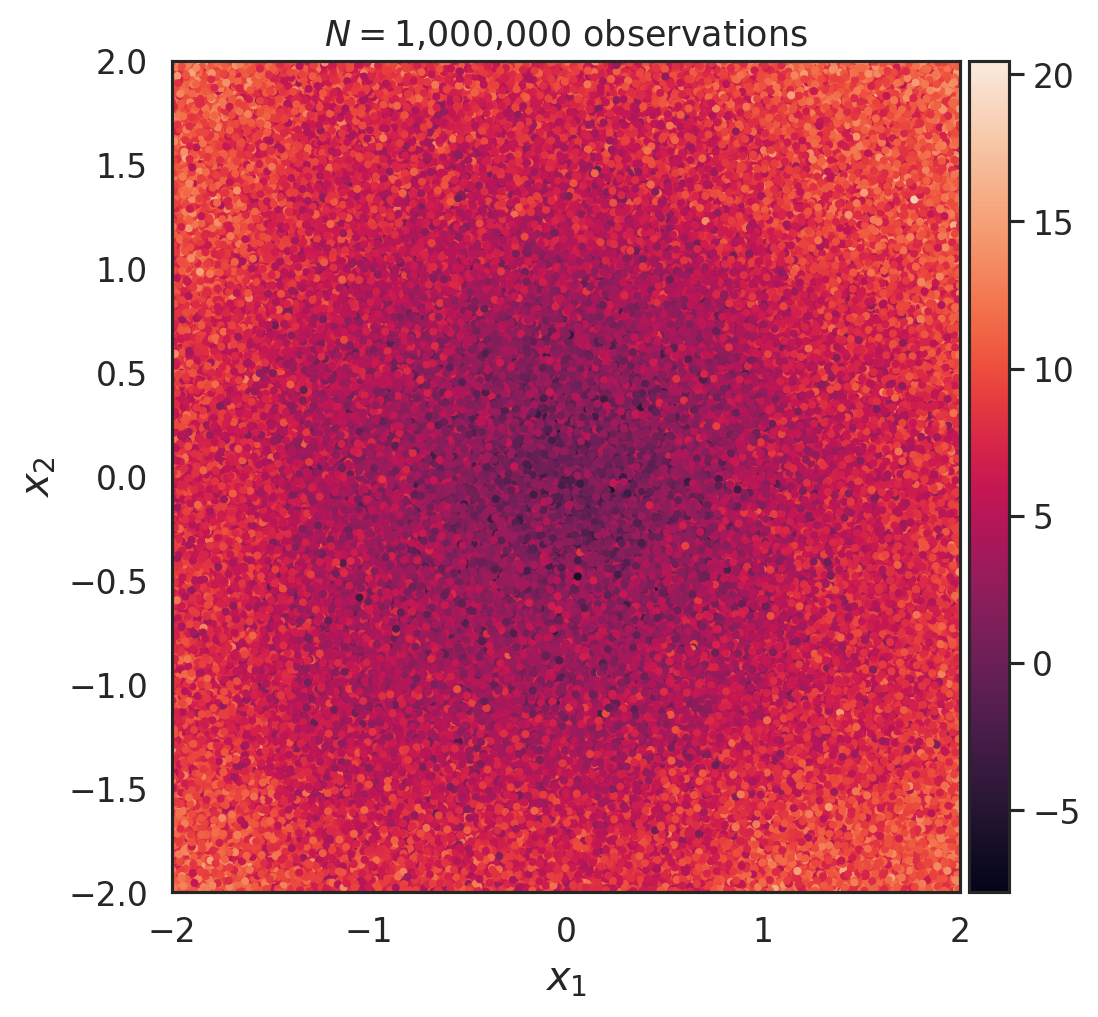}
\caption{Integrated observations $a_n$}
\label{fig:integrated-observations}
\end{subfigure}

\begin{subfigure}[b]{.43\textwidth}
\centering
\includegraphics[width=\textwidth]{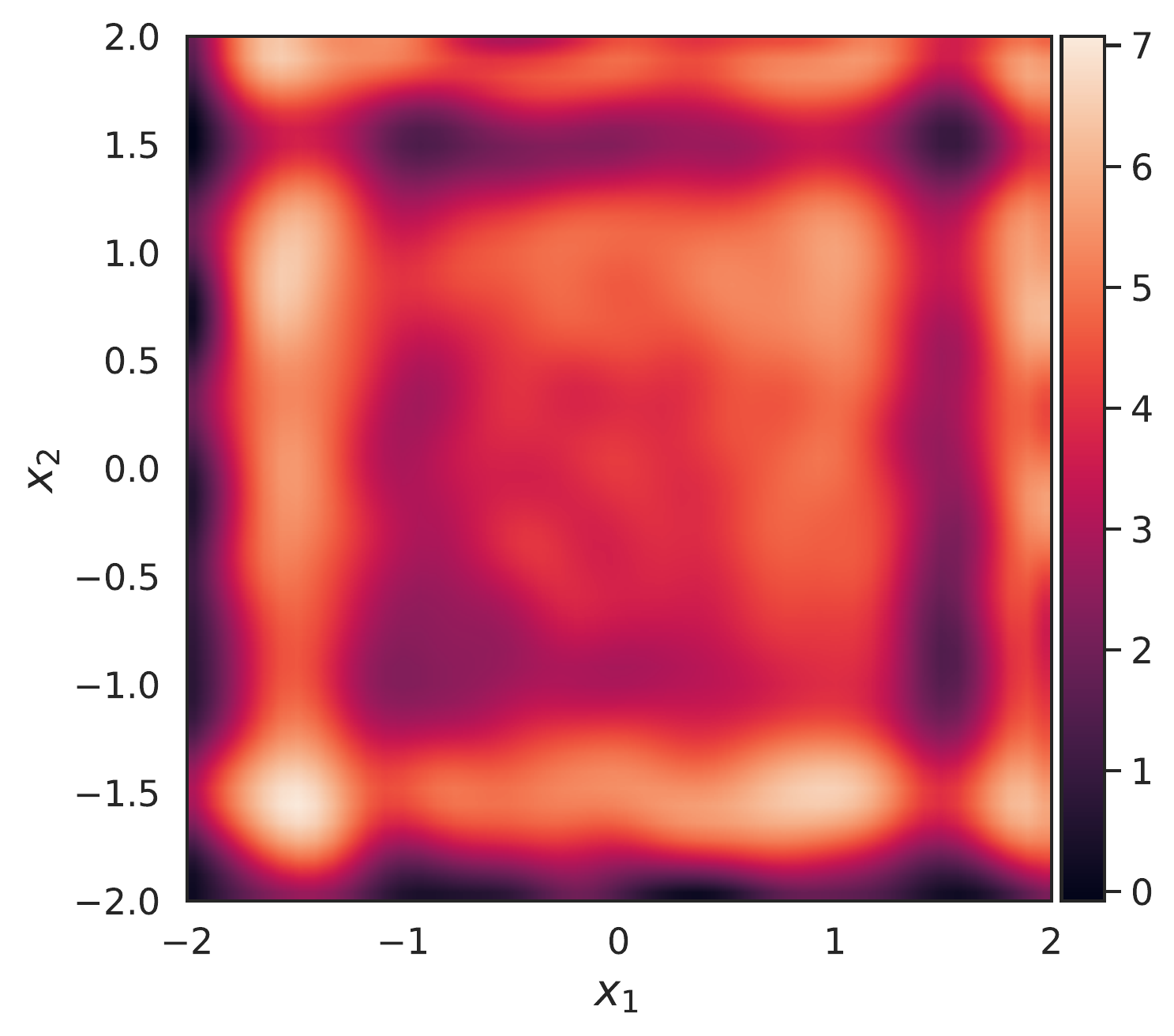}
\caption{Inferred $\mathbb{E}\left[\rho \given \mathcal{D}\right]$.}
\label{fig:rho_inferred}
\end{subfigure}
~
\begin{subfigure}[b]{.43\textwidth}
\centering
\includegraphics[width=\textwidth]{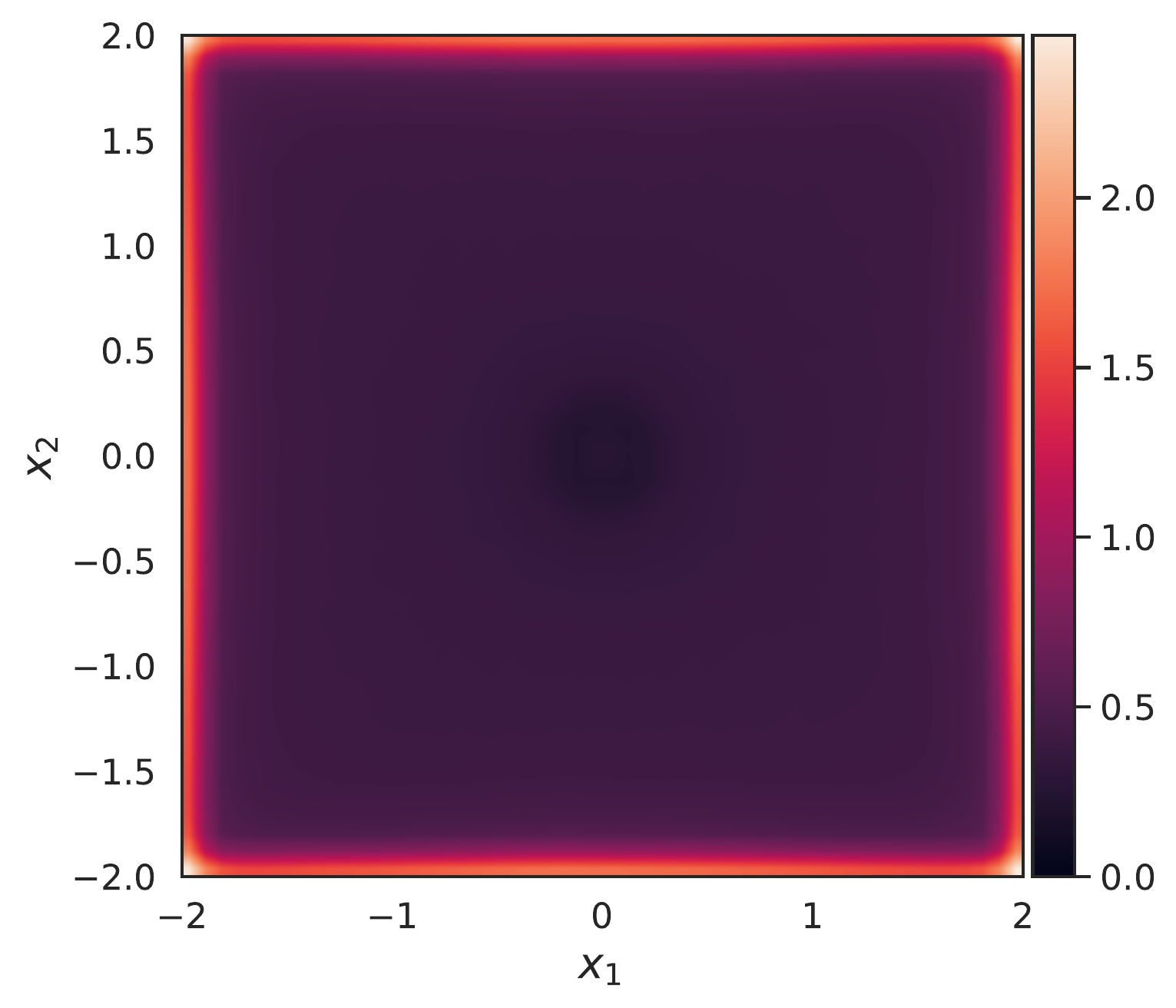}
\caption{Inferred posterior scale $\mathbb{V}\left[\rho(x) \given \mathcal{D}\right]^{1/2}$.}
\label{fig:rho_variance}
\end{subfigure}
\caption{Integrated observation GPs can reveal spatial structure from a
  severely limited vantage point.  This figure summarizes synthetic data
  posterior inference using $N=1{,}000{,}000$ examples in a two-dimensional space.
  Panel \ref{fig:rho_true} depicts the true (unobserved) $\rho(x)$ that
  generates the data.  Panel \ref{fig:integrated-observations} depicts the
  observed data, noisy integrated measurements of $\rho$ from the origin to
  uniformly random points in the square.  \ref{fig:rho_inferred} depicts the
  inferred $\rho$ given these observations using \ziggy
  with a squared exponential kernel.  \ref{fig:rho_variance}
  depicts the marginal posterior standard deviation at each location in
  $\mathcal{X}$-space.
}
\label{fig:synthetic-model-fit}
\end{figure}

\section{Validation on Synthetic Data}
\label{sec:synthetic-experiments}

In this section we study the performance of \ziggy on synthetic data
under varying conditions.  To empirically evaluate \ziggy we measure
the quality of its statistical inferences on a two-dimensional
synthetic example.  To generate synthetic data, we first define the
unobserved ${\rho : \mathbb{R}^2 \mapsto \mathbb{R}}$ and generate
noise with known variance.  The form we chose for $\rho$ is
\begin{align}
  \rho(x) &= 4 + \sum_{d=1}^d x_d \cdot \sin(2 x_d^2 ) \, , 
  \label{eq:synthetic-rho}
\end{align}
which is visualized in Figure~\ref{fig:rho_true}.  This test function has
clear regions of high density far from the observation point (e.g.~around
$x = (-1.5, -1.5)$) and it is positive everywhere.  We generate a dataset
of $N$ noisy integrated observations 
\begin{align}
  x_n &\sim \text{Unif}(\mathcal{X}) & \text{ simulated star locations } \\
  e_n &= \int_{x \in R_n} \rho(x) dx & \text{ simulated extinctions } \\
  a_n &\sim \mathcal{N}(e_n, \sigma_n^2) & \text{ noisy integrated observation }
\end{align}
where the domain is $\mathcal{X} = [-2, 2]^2$, and the noise variance is
chosen to be $\sigma_n^2 = 4$ to roughly approximate the noise level of
extinction measurements.  The true extinctions, $e_n$, are computed
to high precision using numerical quadrature, integrating from the origin
to the random location $x_n$. The integrated observations are depicted in
Figure~\ref{fig:integrated-observations}.  We estimate this synthetic
$\rho(\cdot)$ from noisy integrated observations using the approach
developed in the previous section.  In this experiment, we fix a grid of
$20 \times 20$ inducing points, evenly spaced from $[-2, 2]$ in each of the
two dimensions.

In Section~\ref{sec:synthetic-data-size}, the first set of simulations
compares posterior estimates of $\rho(\cdot)$ as a function of training
data set size.  Posterior estimates continue to become more
accurate and well-calibrated even as the number of integrated observation
$N$ grows from one hundred thousand to one million.  This accuracy
highlights the importance of scaling inference to more observations in
order to obtain a more accurate statistical estimator.

In Section~\ref{sec:synthetic-kernel-comparison}, we compare posterior
estimate quality with a fixed dataset over different kernels, each encoding
different assumptions about the underlying function. This experiment tests
the Monte Carlo approach that can incorporate kernels that do not admit an
analytic form for the semi-integrated kernel.  Finally,
\Cref{sec:app-synthetic-experiments} describes additional synthetic
experiments that examine the effect of two algorithm hyper parameter
choices on variational objective optimization, the whitened
parameterization and the number of Monte Carlo semi-integrated covariance
samples.

A \texttt{python} and \texttt{PyTorch} implementation is
available.\footnote{Available at [REDACTED TO PRESERVE AUTHOR
  BLINDNESS].} Code to reproduce experiments are included in the
repository.

\subsection{The quality of the estimate}
\label{sec:synthetic-data-size}

We study the quality of the estimate of $\rho(x)$ as a function of
data set size $N$.  We fit the dust model to synthetic observations
with data set sizes $N \in \{10^3, 10^4, 10^5, 10^6\}$.  We measure
model quality by computing root mean squared error (RMSE) and log
likelihood (LL) on a set of $N_{\mathrm{test}} = 2{,}000$ held out
extinction values
\begin{align}
  \mathrm{RMSE} &= 
    \left( \frac{1}{N_{\mathrm{test}}}\sum_{n=1}^{N_{\mathrm{test}}}
      \left(\mathbb{E}\left[e_* \given \mathcal{D}\right] - e_*\right)^2
    \right)^{1/2}
  \label{eq:test-rmse-def} \\
  \mathrm{LL} &= \frac{1}{N_{\mathrm{test}}} \sum_{n=1}^{N_{\mathrm{test}}}
      \ln p(e_n \given \mathcal{D}) \,\, .
  \label{eq:test-ll-def}
\end{align}

Each model uses $M=20 \times 20$ inducing points, evenly spaced in a
grid in the input space.  We trained each model until convergence,
using batches of size $1{,}000$, a step size of $1e-3$ and a kernel
parameter step size of $1e-6$.  (Additional optimization details can
be found in the implementation.)

Figure~\ref{fig:synthetic-model-fit} summarizes the model's fit using
one million observations, the most we tried in this example.
Figure~\ref{fig:rho_inferred} displays the posterior mean for
$\rho(x)$ given the one million observations depicted in
\ref{fig:integrated-observations}.  Figure~\ref{fig:rho_variance}
shows the posterior uncertainty (one standard deviation) about the
estimated mean.  The million-observation model recovers the true
latent function well.

Figure~\ref{fig:synthetic-data-size} summarizes predictive quality of \ziggy
as a function of data set size on held out test data using RMSE and
LL.  This illustrates the benefit of including more training data
--- root mean squared error (RMSE) decreases (and log likelihood increases)
significantly as data set size grows.  For example, the RMSE for
$N=100{,}000$ is $2.6$ times bigger than the RMSE for $N=1{,}000{,}000$,
(.082 vs.~.031) and the average log likelihood is significantly better for
the million-observation model.

Figure~\ref{fig:synthetic-data-size-posterior-comparison} makes a more
direct graphical comparison of the posterior estimate of the latent
function as we increase data set size $N$.  As we incorporate more data
into the nonparametric model, the form of the true underlying function
emerges.  To accurately visualize and interpret large scale features of the
latent dust distribution, we will want to incorporate as many stellar
observations as possible.

\begin{figure}[t!]
\centering
\begin{subfigure}[b]{.45\textwidth}
\centering
\includegraphics[width=\textwidth]{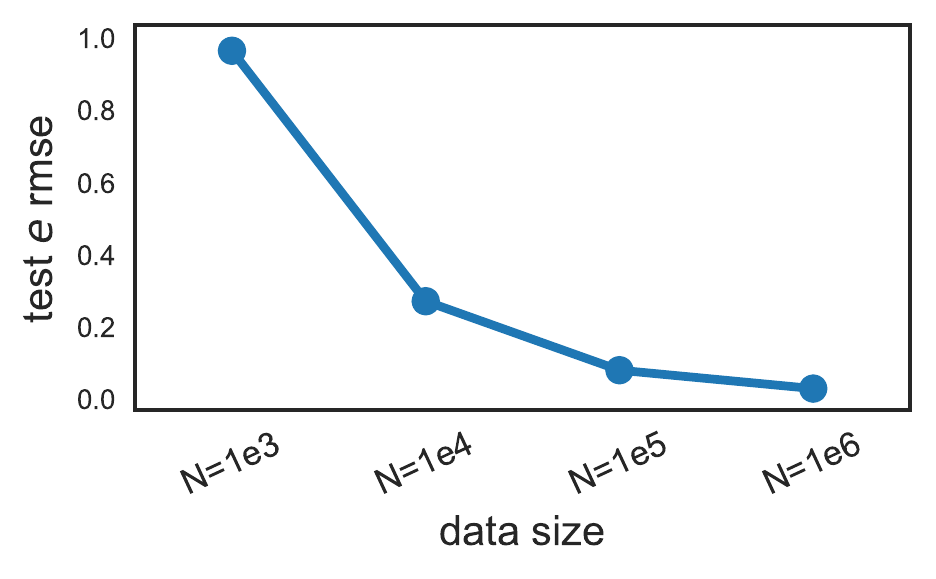}
\caption{Test $e$ RMSE}
\end{subfigure}
~
\begin{subfigure}[b]{.45\textwidth}
\centering
\includegraphics[width=\textwidth]{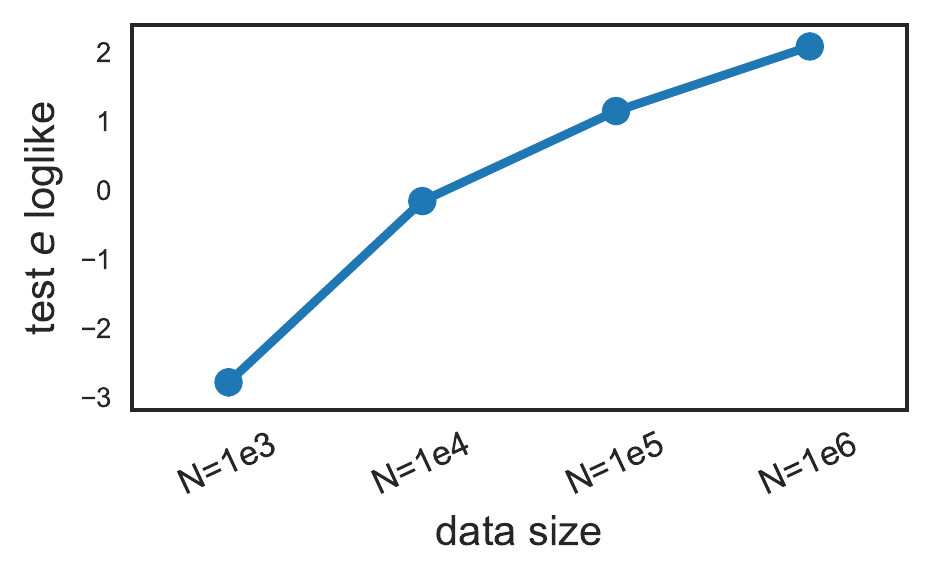}
\caption{Test $e$ log likelihood}
\end{subfigure}

\caption{Scaling to massive $N$ improves estimator performance. Above, we
compare predictive RMSE (left) and log likelihood (right) as a function of
data set size $N$ on a held out sample of test stars.}
\label{fig:synthetic-data-size}
\end{figure}

\begin{figure}[t!]
\centering
  \begin{subfigure}[b]{.4\textwidth}
  \includegraphics[width=\textwidth]{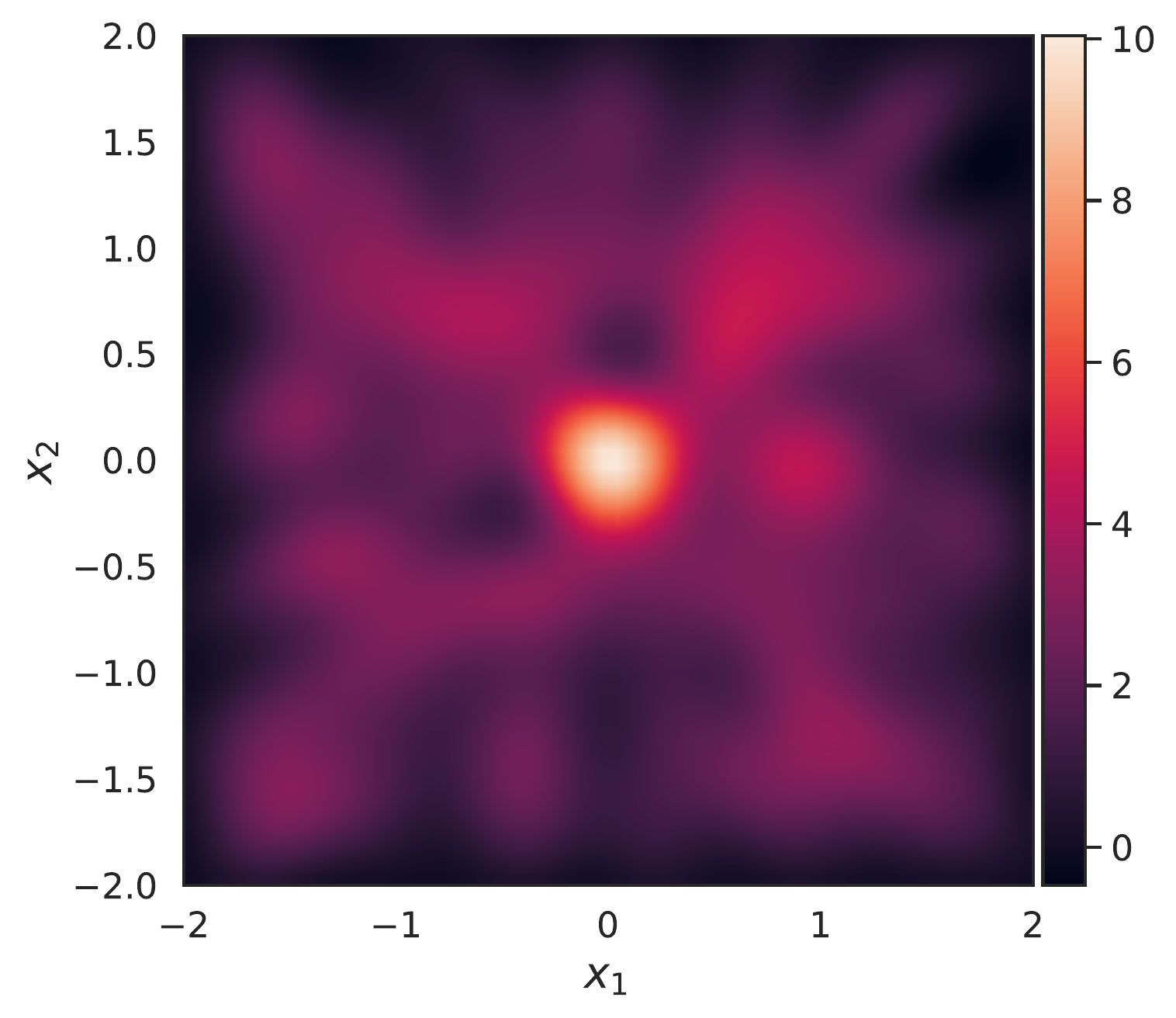}
  \caption{$\mathbb{E}[\rho \given \mathcal{D}]$, $N=1{,}000$}
  \end{subfigure}
  ~
  \begin{subfigure}[b]{.4\textwidth}
  \includegraphics[width=\textwidth]{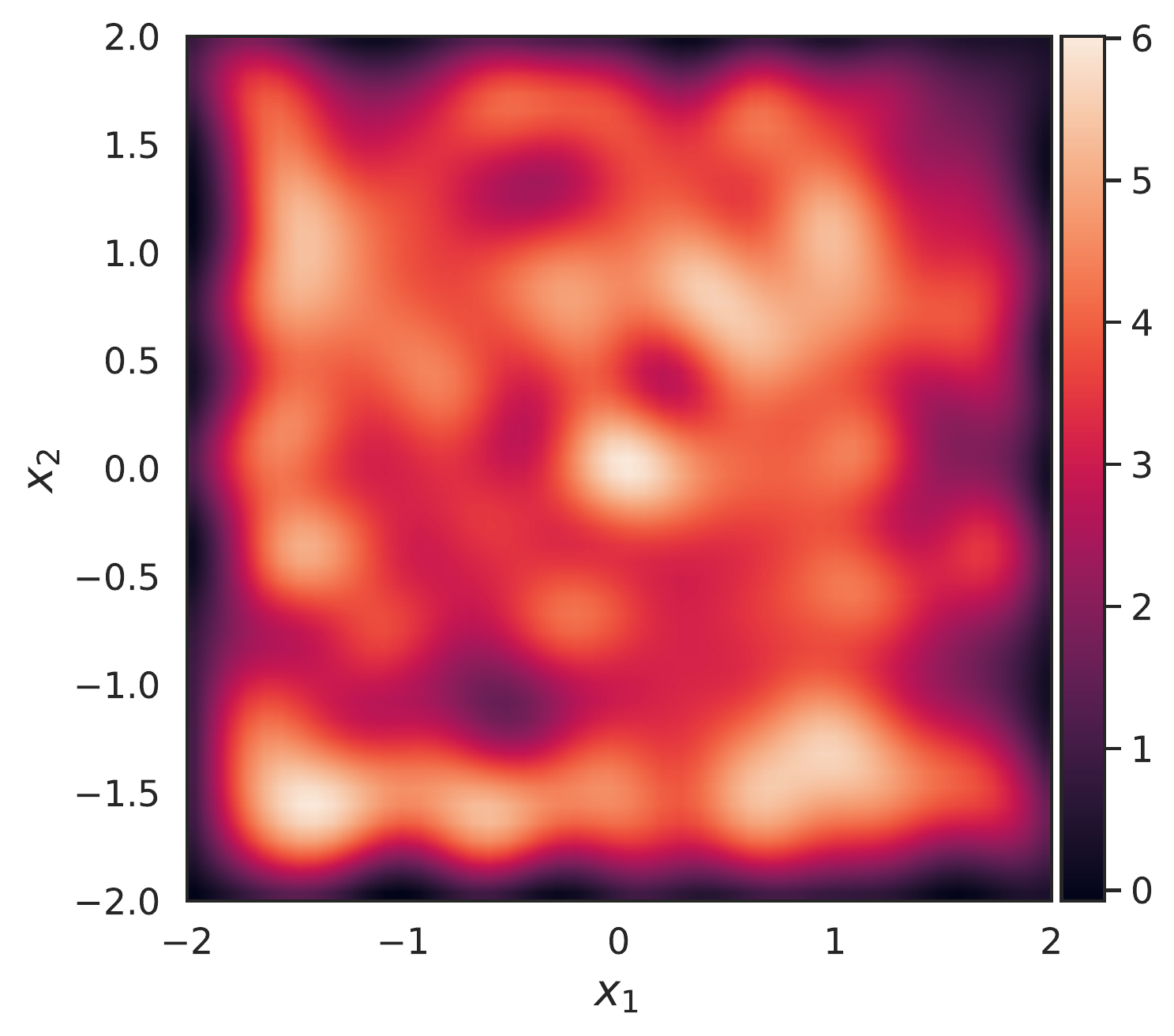}
  \caption{$\mathbb{E}[\rho \given \mathcal{D}]$, $N=10{,}000$}
  \end{subfigure}
  ~
  \begin{subfigure}[b]{.4\textwidth}
  \includegraphics[width=\textwidth]{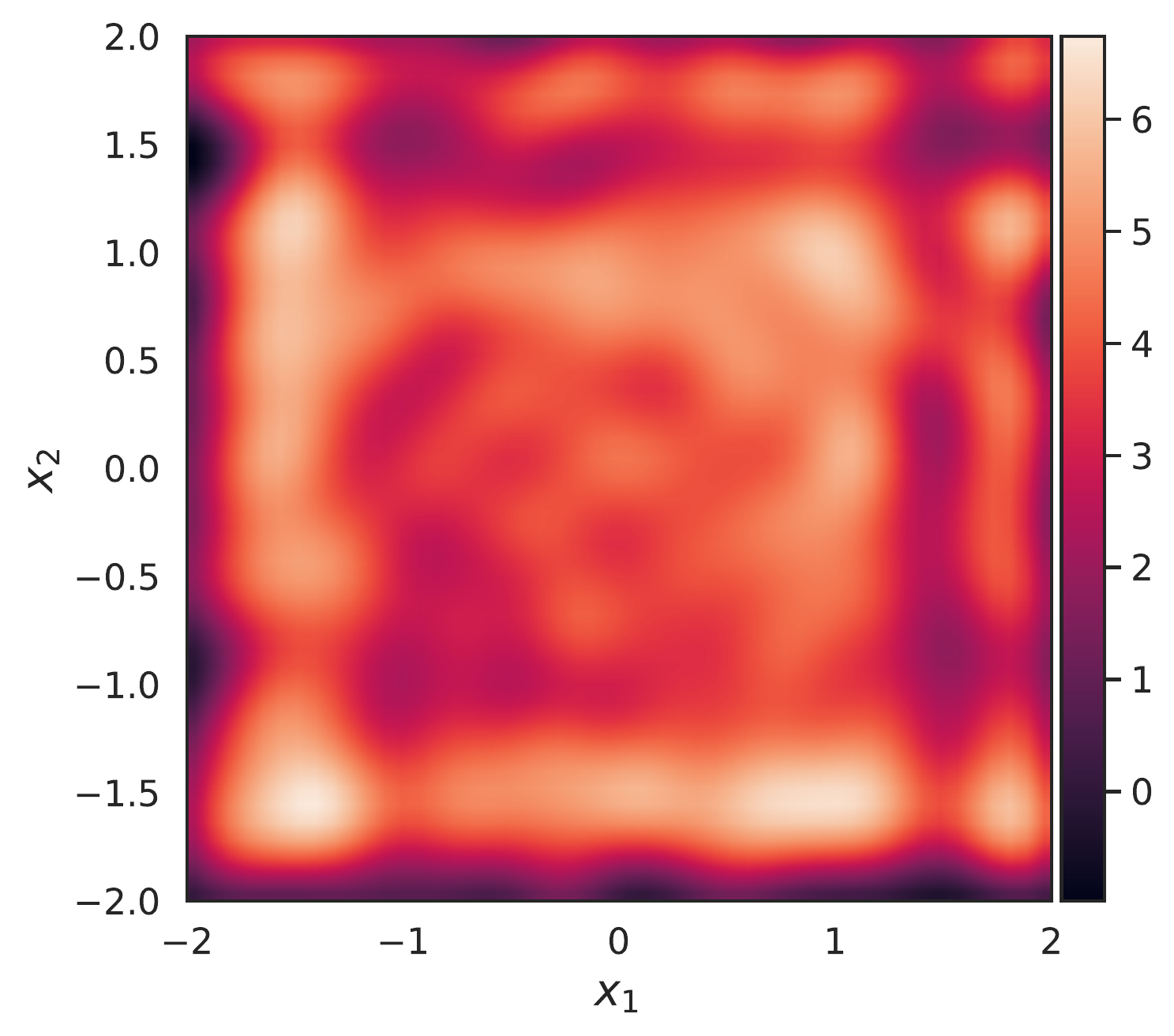}
  \caption{$\mathbb{E}[\rho \given \mathcal{D}]$, $N=100{,}000$}
  \end{subfigure}
  ~
  \begin{subfigure}[b]{.4\textwidth}
  \includegraphics[width=\textwidth]{figs/synthetic-data-size/SqExp-N-1000000/posterior-fmu.pdf}
  \caption{$\mathbb{E}[\rho \given \mathcal{D}]$, $N=1{,}000{,}000$}
  \end{subfigure}
\caption{More data leads to more accurate function posterior predictions.}
\label{fig:synthetic-data-size-posterior-comparison}
\end{figure}

\subsection{Comparing kernels}
\label{sec:synthetic-kernel-comparison}

\begin{figure}[t!]
\centering
\begin{subfigure}[b]{.31\textwidth}
\centering
\includegraphics[width=\textwidth]{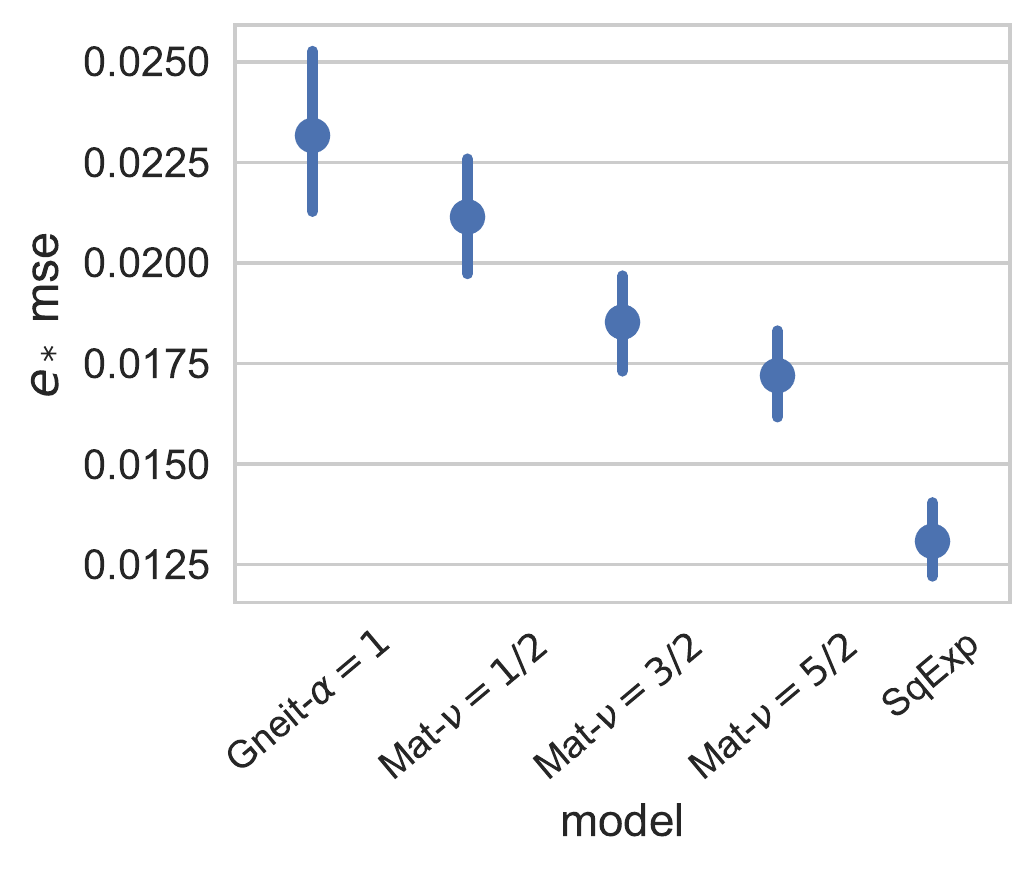}
\caption{Test $e$ MSE}
\end{subfigure}
~
\begin{subfigure}[b]{.31\textwidth}
\centering
\includegraphics[width=\textwidth]{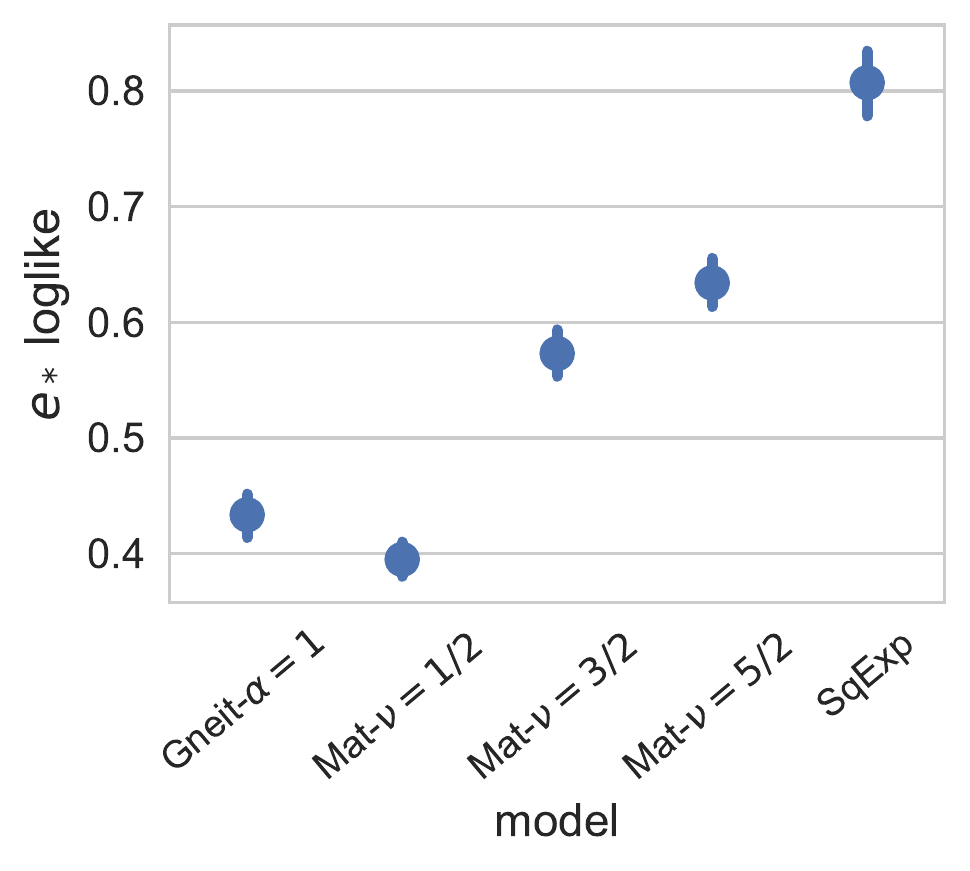}
\caption{Test $e$ log likelihood}
\end{subfigure}
~
\begin{subfigure}[b]{.31\textwidth}
\centering
\includegraphics[width=\textwidth]{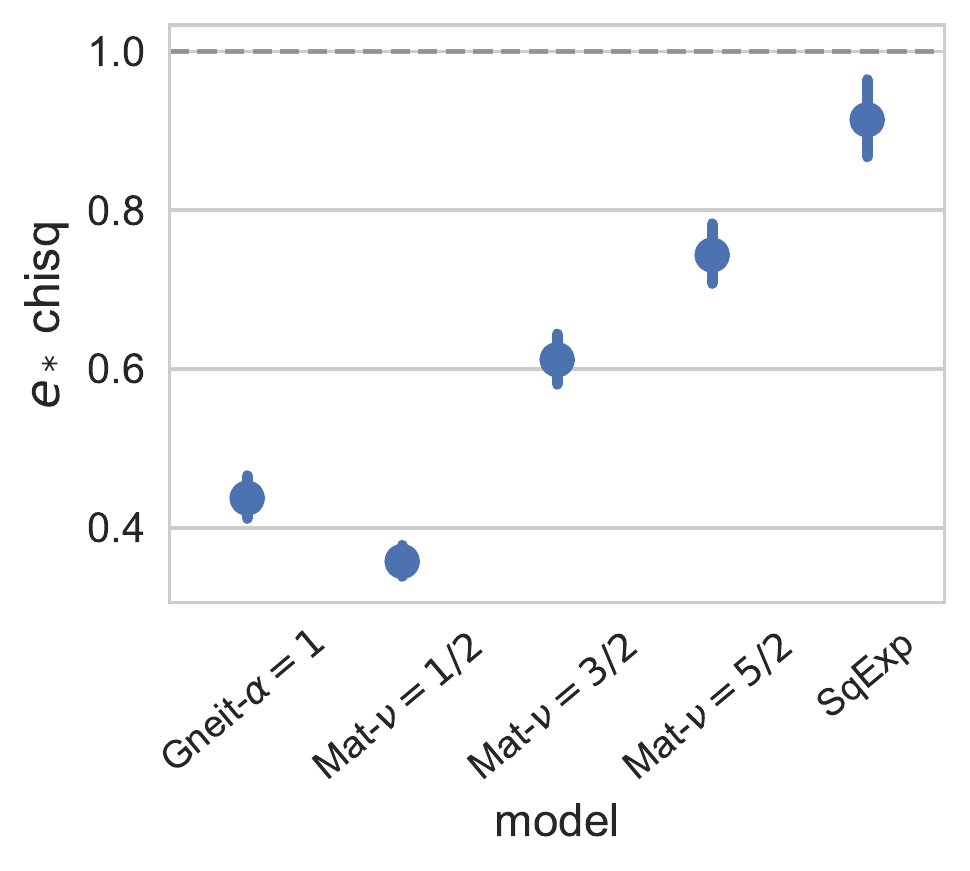}
\caption{Test $e$ $\chi^2$}
\end{subfigure}

\begin{subfigure}[b]{.31\textwidth}
\centering
\includegraphics[width=\textwidth]{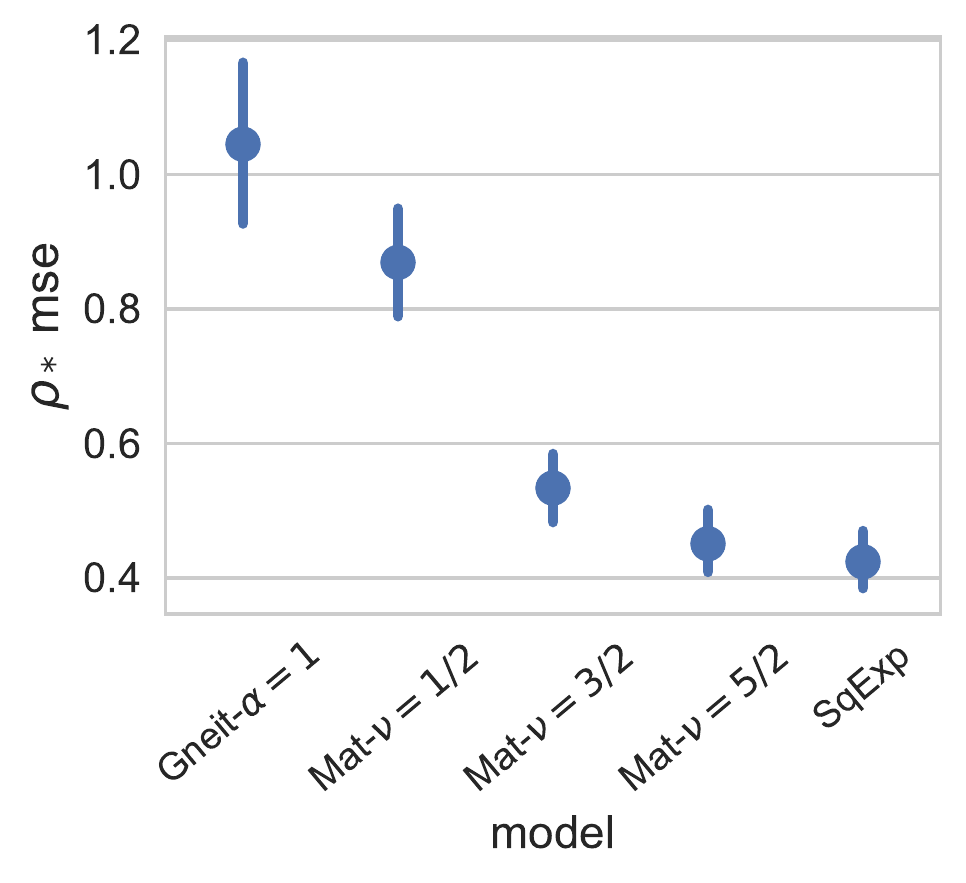}
\caption{Test $\rho(x)$ MSE}
\end{subfigure}
~
\begin{subfigure}[b]{.31\textwidth}
\centering
\includegraphics[width=\textwidth]{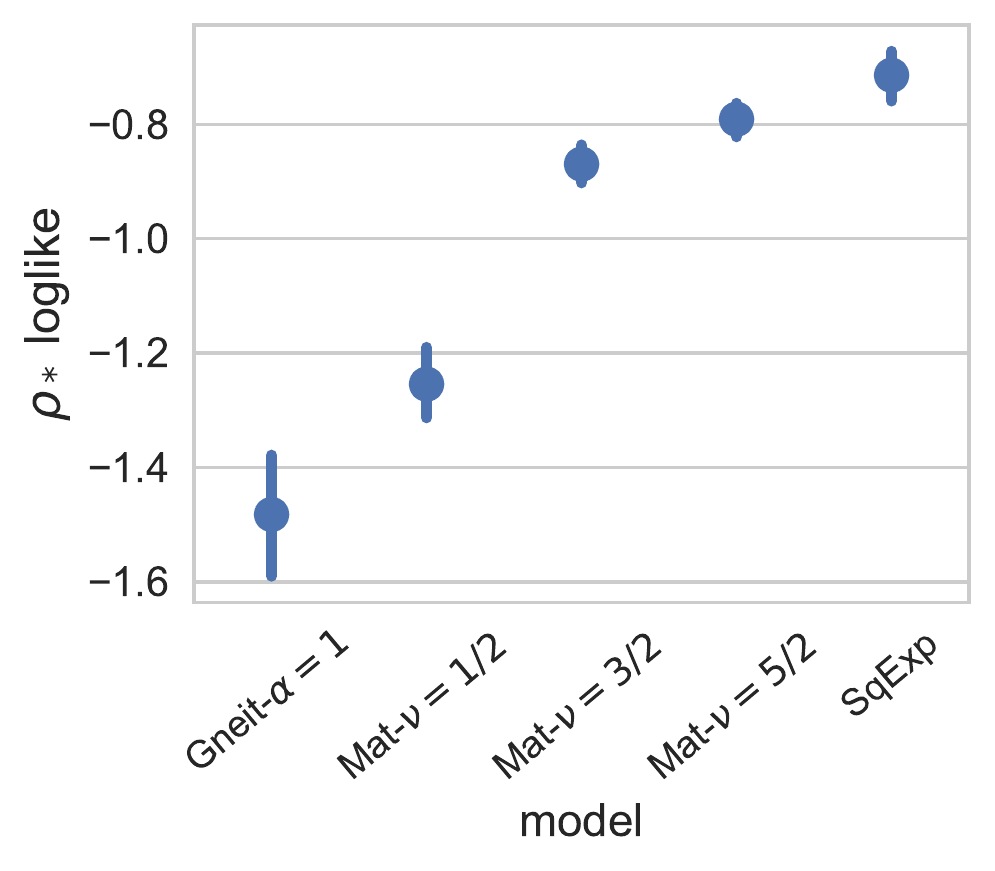}
\caption{Test $\rho(x)$ log likelihood}
\end{subfigure}
~
\begin{subfigure}[b]{.31\textwidth}
\centering
\includegraphics[width=\textwidth]{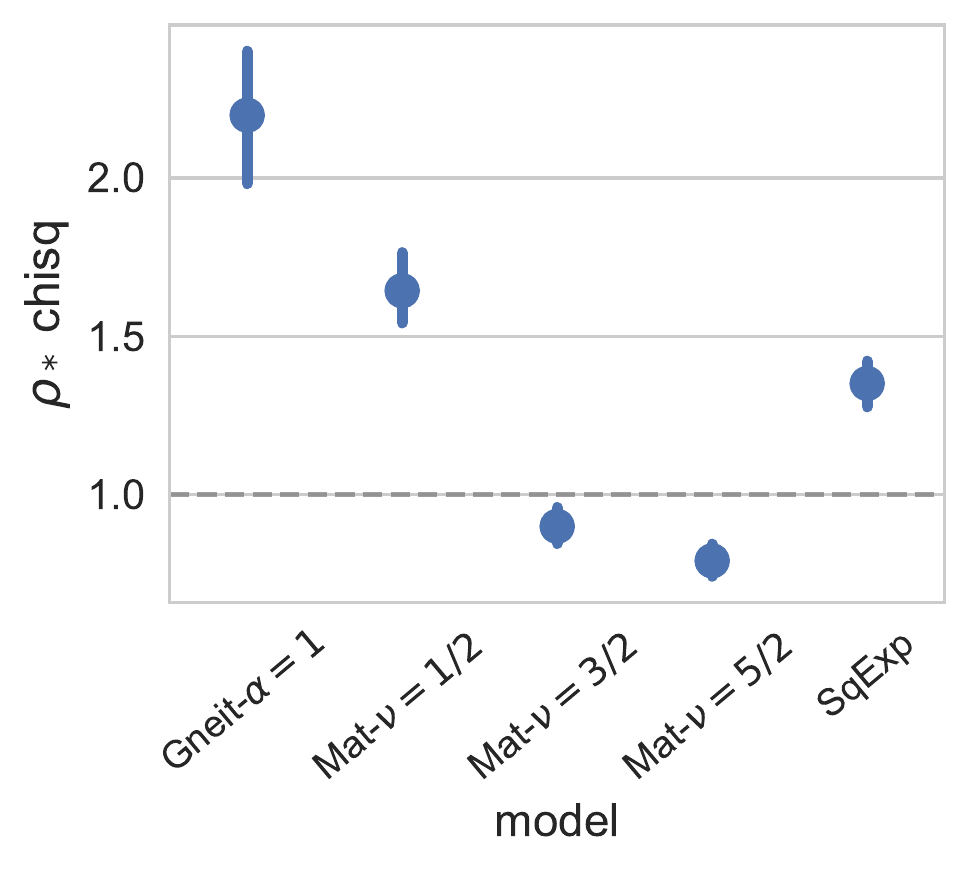}
\caption{Test $\rho(x)$ $\chi^2$}
\end{subfigure}
\caption{Predictive summaries across different models and inference schemes.}
\label{fig:synthetic-error-comparison}
\end{figure}

We compared multiple models by fitting the variational approximation and
tuning the covariance function parameters for five different kernel choices.
In this synthetic setup, we ran
each model for 50 epochs, saving the the model with the best ELBO value.
We used mini-batches of size $|B|=1{,}000$, and started the step size at
$.1$, reducing it every epoch by multiplying by a decay factor of .95.

For kernels that do not have a closed form semi-integrated version, we used
$L=30$ uniform grid Monte Carlo samples to estimate the integrated
covariance.  We chose this on the conservative end --- in
Appendix~\ref{sec:mc-semi-estimators} shows that on the order of $L=10$ to
$20$ samples is often just as effective as the analytic kernel in the
squared exponential case.

We validate the inferences against a set of $2{,}000$ held-out test
observations, which include ground truth extinctions $e_*$ and pointwise
values, $\rho_*$.  We quantify the quality of the inferences with mean
squared error, log likelihood, and the $\chi^2$-statistic on the test data.
We also visually validate \ziggy with a $QQ$-plot and coverage
comparison, also on test data.  We also visualize the behavior \ziggy
as a function of distance to the origin ---
specifically, we visualize how accurate and well-calibrated test inferences
are as synthetic observations are made farther away.

In Figure~\ref{fig:synthetic-error-comparison} we compare the mean squared
error (MSE), the log likelihood, and the $\chi^2$ statistic on the held out
test data.  We compare five different kernels -- (i) Gneiting $\alpha=1$,
(ii-iv) \matern, $\nu \in \{1/2, 3/2, 5/2 \}$, and (v) the squared
exponential kernel.
The results match our expectations. The true $\rho(\cdot)$ is
smooth, and the smoother kernels tend to have lower predictive
error, higher log likelihood, and better calibrated $\chi^2$ statistic
(closer to one). The test statistic comparison has chosen a good match in
the squared exponential kernel.

We also note that the MSE for the process evaluations $\rho_*$ is much
higher than the test MSE for extinctions $e_*$.  This is also expected ---
intuitively, if the $\rho(\cdot)$ predictions are well-calibrated, then
over a larger distance the pointwise errors can be averaged out, resulting
in a better predictor for extinctions.

A crucial property of a good posterior approximation is well-calibrated
uncertainty.  To evaluate the quality of the uncertainty measurements for
unobserved extinctions and pointwise process values, we inspect statistics
of test-sample $z$-scores
\begin{align}
  z_* &= \frac{e_* - \mathbb{E}[e_* \given x_*, \mathcal{D}]}
              {\sqrt{\mathbb{V}\left[ e_* \given x_*, \mathcal{D} \right]}}
\end{align}
for both $e_*$ and $\rho_*$.
If the predictions are well-calibrated, the statistics of predictive
$z$-scores should resemble a standard normal distribution.  Concretely, the
order statistics of a collection of $N_{test}$ $z$-scores should closely
match the theoretical quantiles of a standard normal distribution --- a
relationship that can be graphically inspected with a QQ-plot.  To compare
different models, we visualize QQ-plots for the different covariance
functions in Figure~\ref{fig:synthetic-qq}.  As an additional summary of
calibration, we compare the fraction of predicted examples covered by $1/2,
1, 2$, and 3 posterior standard deviations, summarized in
Table~\ref{tab:synthetic-coverage}.  We see that the squared exponential
kernel $z$-scores and empirical coverage are not far from what is expected
by a theoretical standard normal distribution, indicating that the
estimates are well-calibrated.

We also depict model fit statistics as a function of distance from the
origin in Figure~\ref{fig:summary-by-distance}.  In
Figures~\ref{fig:e-zscore-by-distance} and \ref{fig:f-zscore-by-distance}
show that pointwise predictions are again well-calibrated as a
function of distance, though we observe a small amount of concentration about
$0$ for the extinction predictions.  This indicates that the predictive
model was too conservative, or not confident enough in the posterior
expectation at nearby values.  This is corroborated by
Figure~\ref{fig:e-posterior-variance-by-distance}, which shows the
posterior variance for extinctions is high at nearby (and distant) values,
however MSE are high mostly for distant values.  Good calibration, even at
far distances, is encouraging, as the goal is to form good estimates of
this density far away from the observer location.

\begin{figure}[t!]
\centering
\begin{subfigure}[b]{.35\textwidth}
\centering
\includegraphics[width=\textwidth]{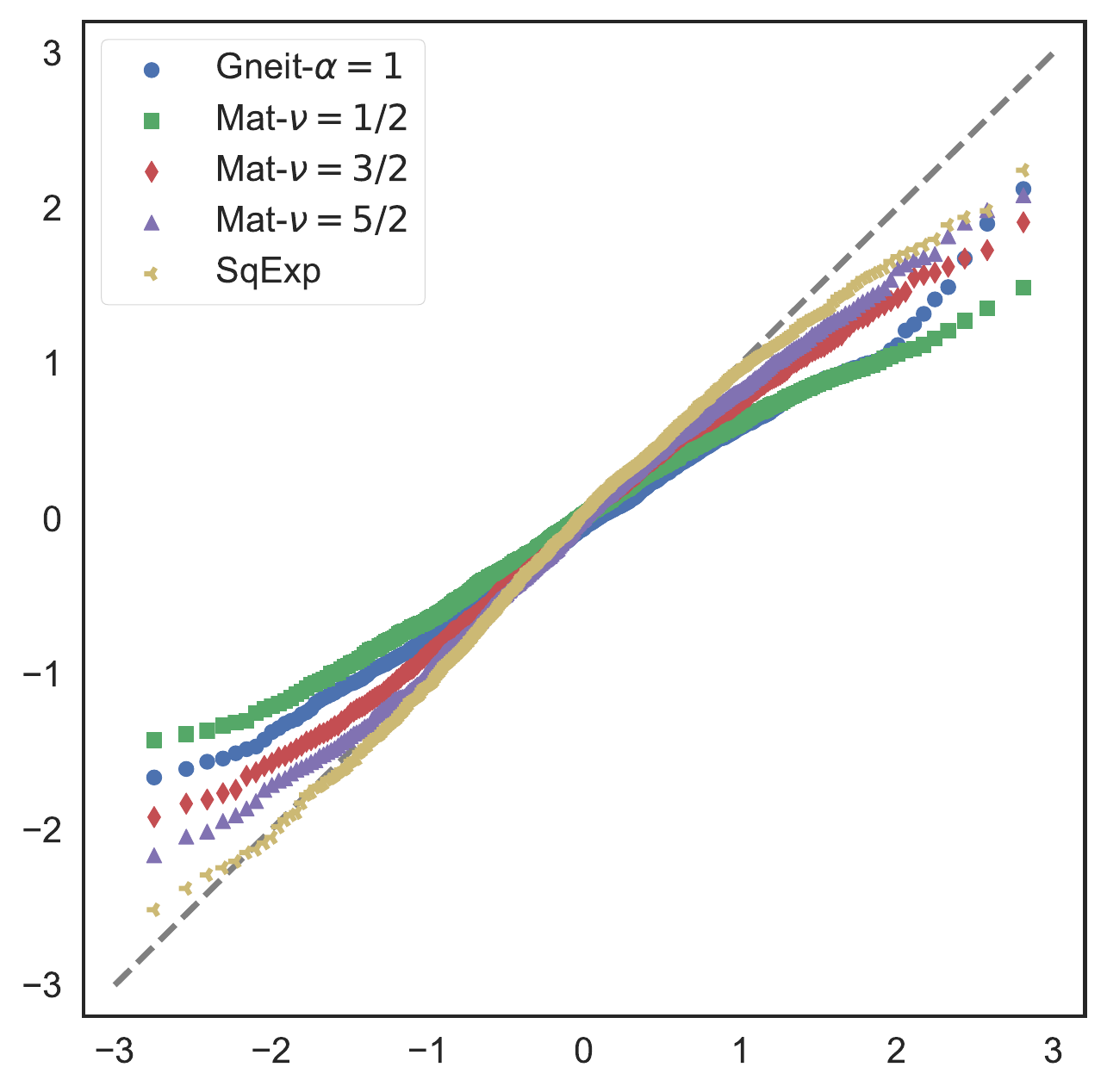}
\caption{Test extinctions $e$}
\end{subfigure}
~
\begin{subfigure}[b]{.35\textwidth}
\centering
\includegraphics[width=\textwidth]{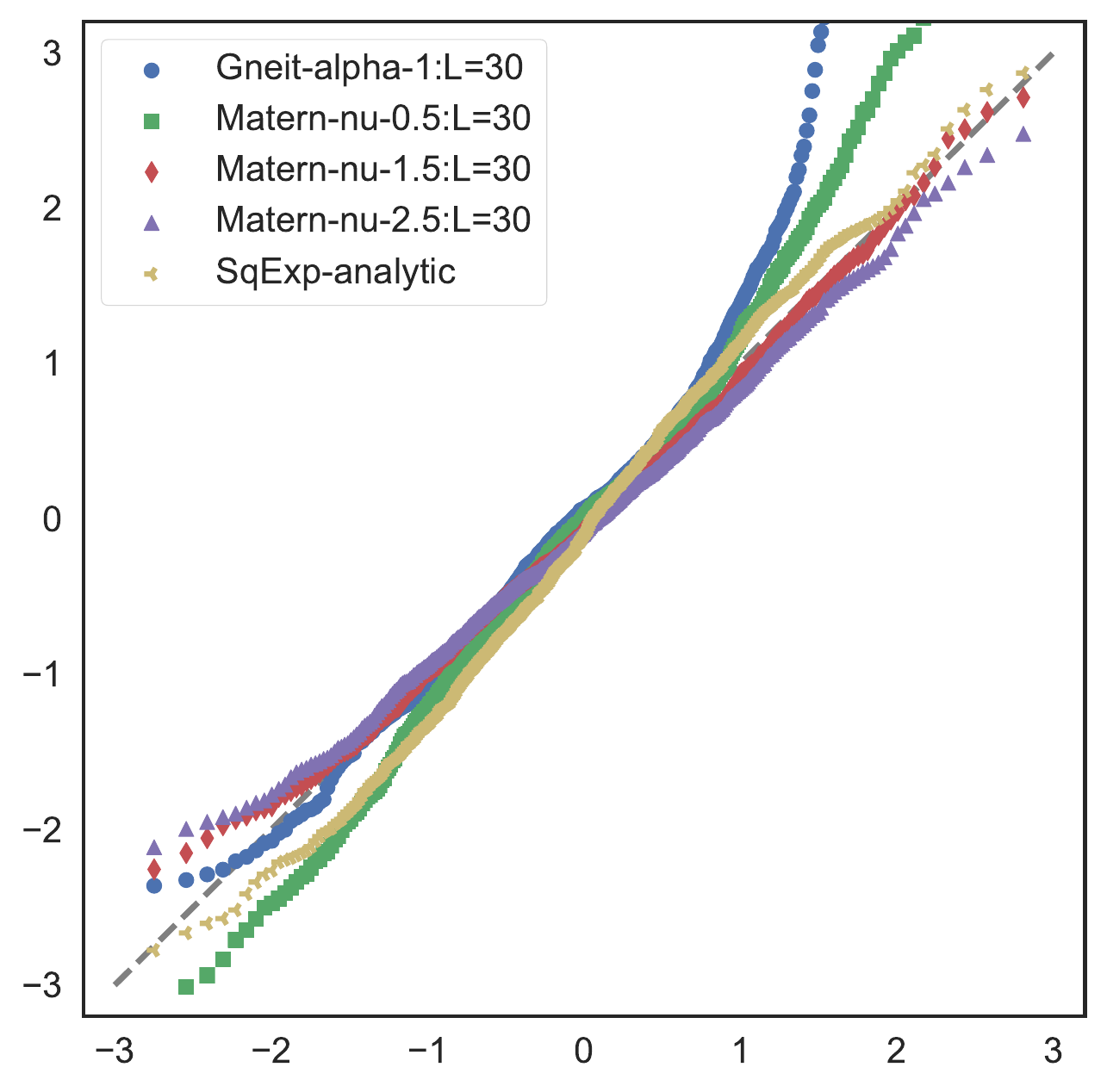}
\caption{Test $\rho(x)$ values}
\end{subfigure}
\caption{Posterior for extinctions $e$ and dust density $\rho$ are well
  calibrated.  QQ-plot for predicted distributions on 2{,}000 held out
  integrated-$\rho$ test points $e_*$.  Theoretical normal quantiles are on
  the horizontal axis, and predictive $z$-scored quantiles are on the
  vertical axis.
}
\label{fig:synthetic-qq}
\end{figure}

\begin{table}[]
    \centering
    \scalebox{.8}{
        \begin{tabular}{lrrrrrr}
\toprule
{} &  Gneit-$\alpha=1$ &  Mat-$\nu=1/2$ &  Mat-$\nu=3/2$ &  Mat-$\nu=5/2$ &  SqExp &  $\mathcal{N}(0, 1)$ \\
$\sigma$ &                   &                &                &                &        &                      \\
\midrule
0.5      &             0.536 &          0.560 &          0.441 &          0.407 &  0.384 &                0.383 \\
1.0      &             0.880 &          0.917 &          0.785 &          0.730 &  0.679 &                0.683 \\
2.0      &             0.997 &          0.999 &          0.999 &          0.988 &  0.971 &                0.955 \\
3.0      &             1.000 &          1.000 &          1.000 &          1.000 &  1.000 &                0.997 \\
\bottomrule
\end{tabular}

    }
    \caption{Coverage fractions at various levels of z-score standard
    deviation $\sigma$ for test extinctions $e_*$. Relative to theoretical
    coverage fractions, the inferences concentrate too many
    predictions near zero --- a sign that we are not confident enough. }
    \label{tab:synthetic-coverage}
\end{table}

\begin{figure}[t!]
\centering
\begin{subfigure}[b]{.4\textwidth}
\centering
\includegraphics[width=\textwidth]{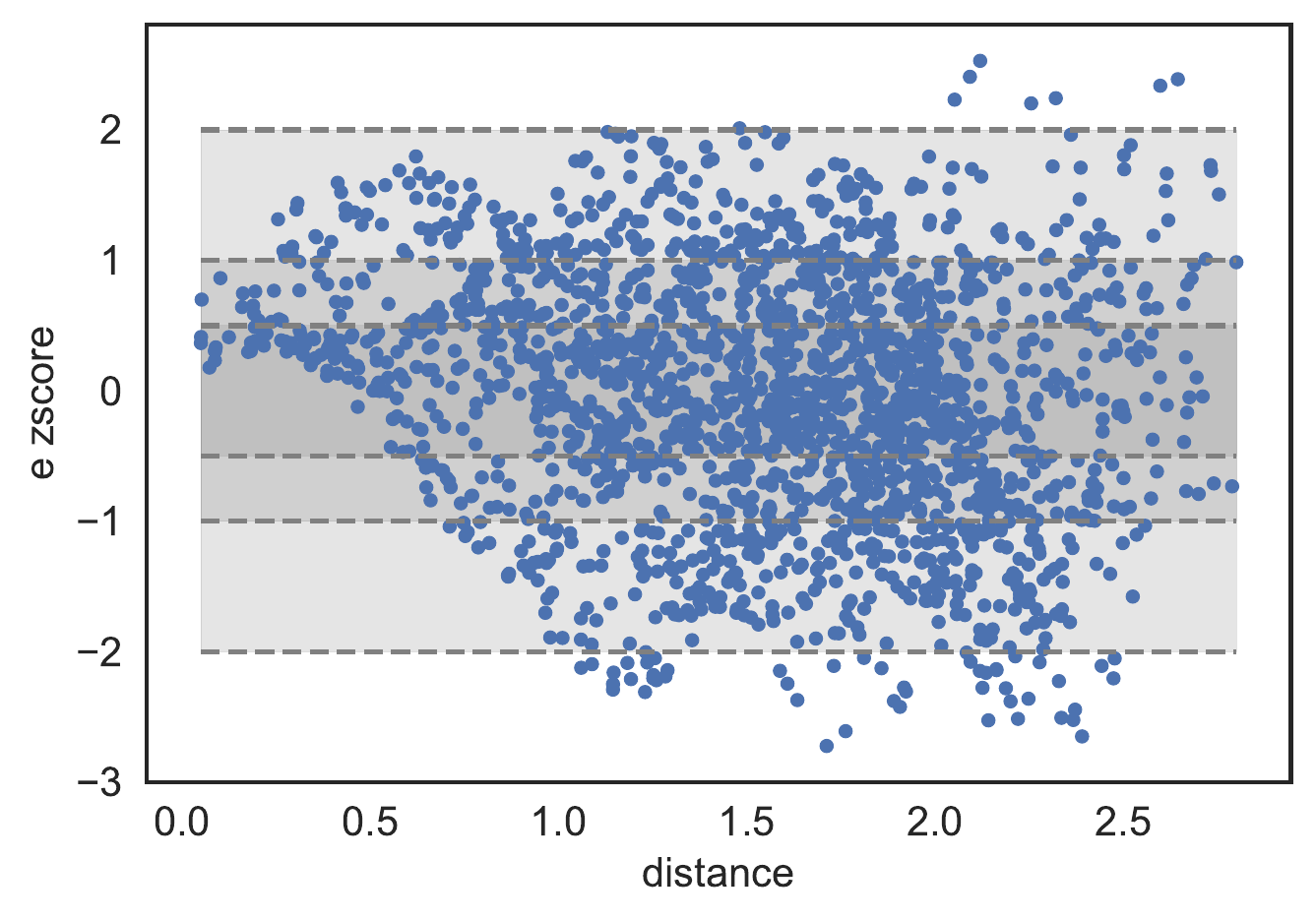}
\caption{Test $e$ z scores by distance}
\label{fig:e-zscore-by-distance}
\end{subfigure}
~
\begin{subfigure}[b]{.4\textwidth}
\centering
\includegraphics[width=\textwidth]{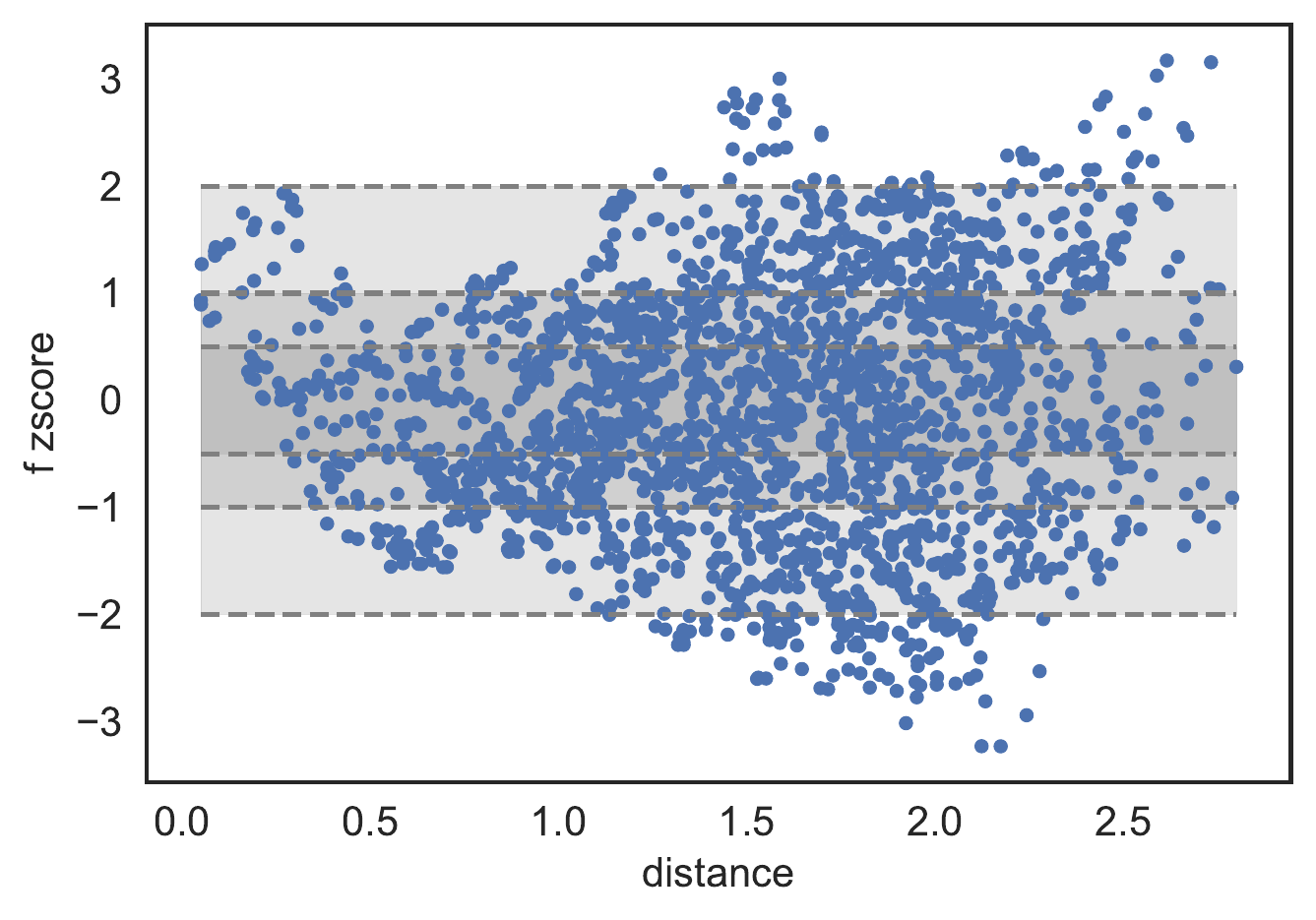}
\caption{Test $\rho(x)$ z scores by distance}
\label{fig:f-zscore-by-distance}
\end{subfigure}

\begin{subfigure}[b]{.4\textwidth}
\centering
\includegraphics[width=\textwidth]{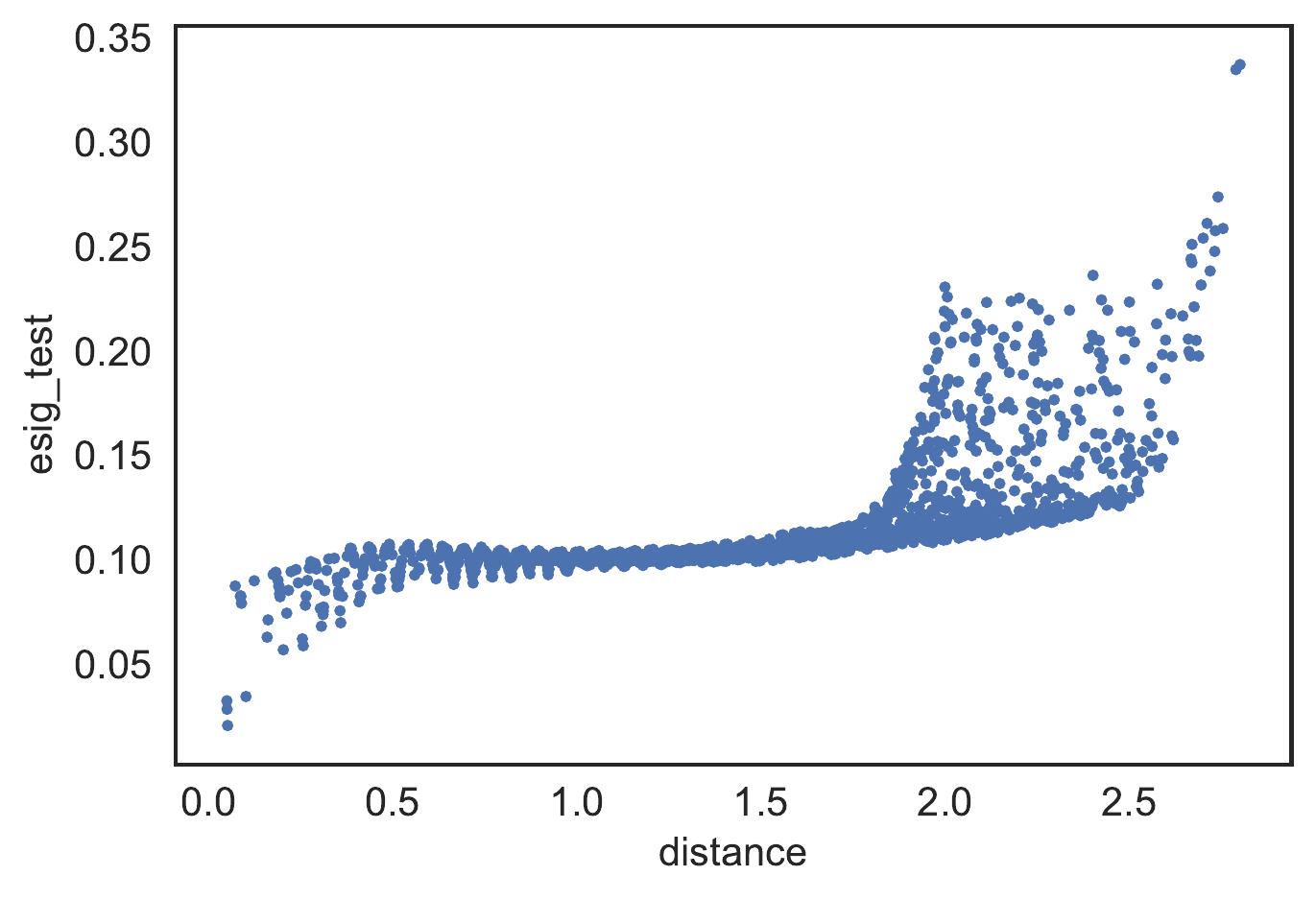}
\caption{Test $e$ posterior standard deviation by distance}
\label{fig:e-posterior-variance-by-distance}
\end{subfigure}
~
\begin{subfigure}[b]{.4\textwidth}
\centering
\includegraphics[width=\textwidth]{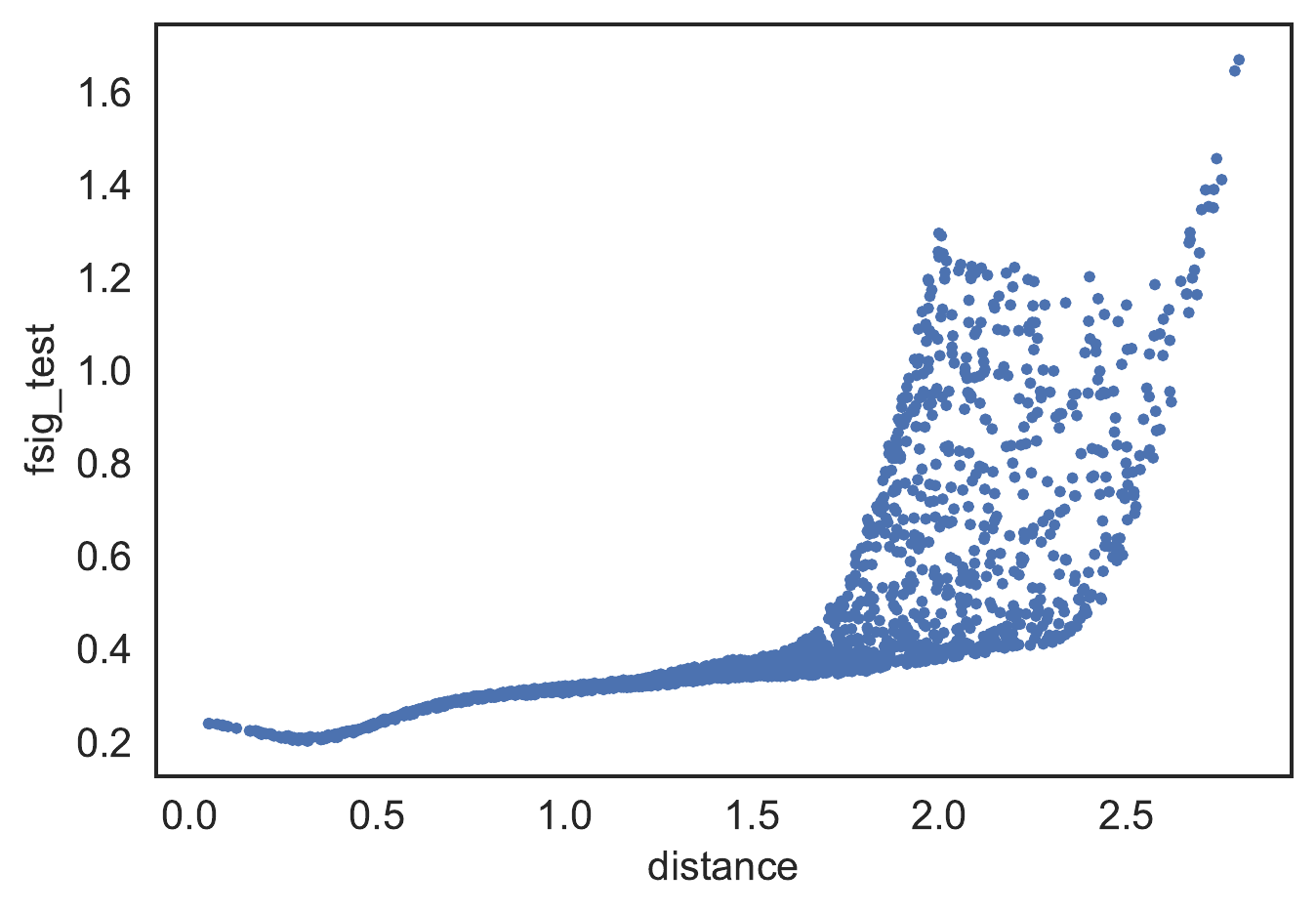}
\caption{Test $\rho(x)$ posterior standard deviation by distance}
\end{subfigure}

\begin{subfigure}[b]{.4\textwidth}
\centering
\includegraphics[width=\textwidth]{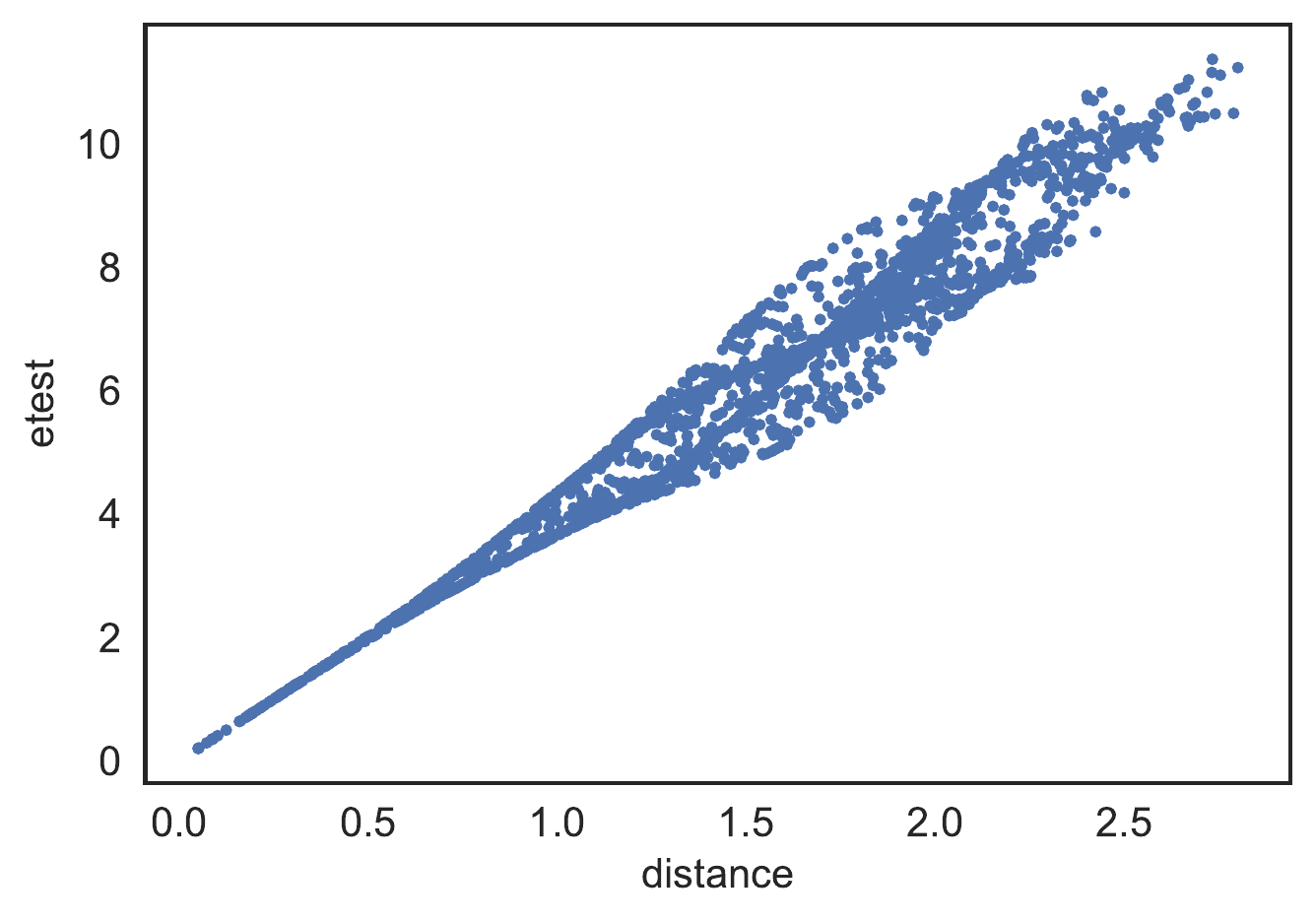}
\caption{Test $e$ values by distance}
\end{subfigure}
~
\begin{subfigure}[b]{.4\textwidth}
\centering
\includegraphics[width=\textwidth]{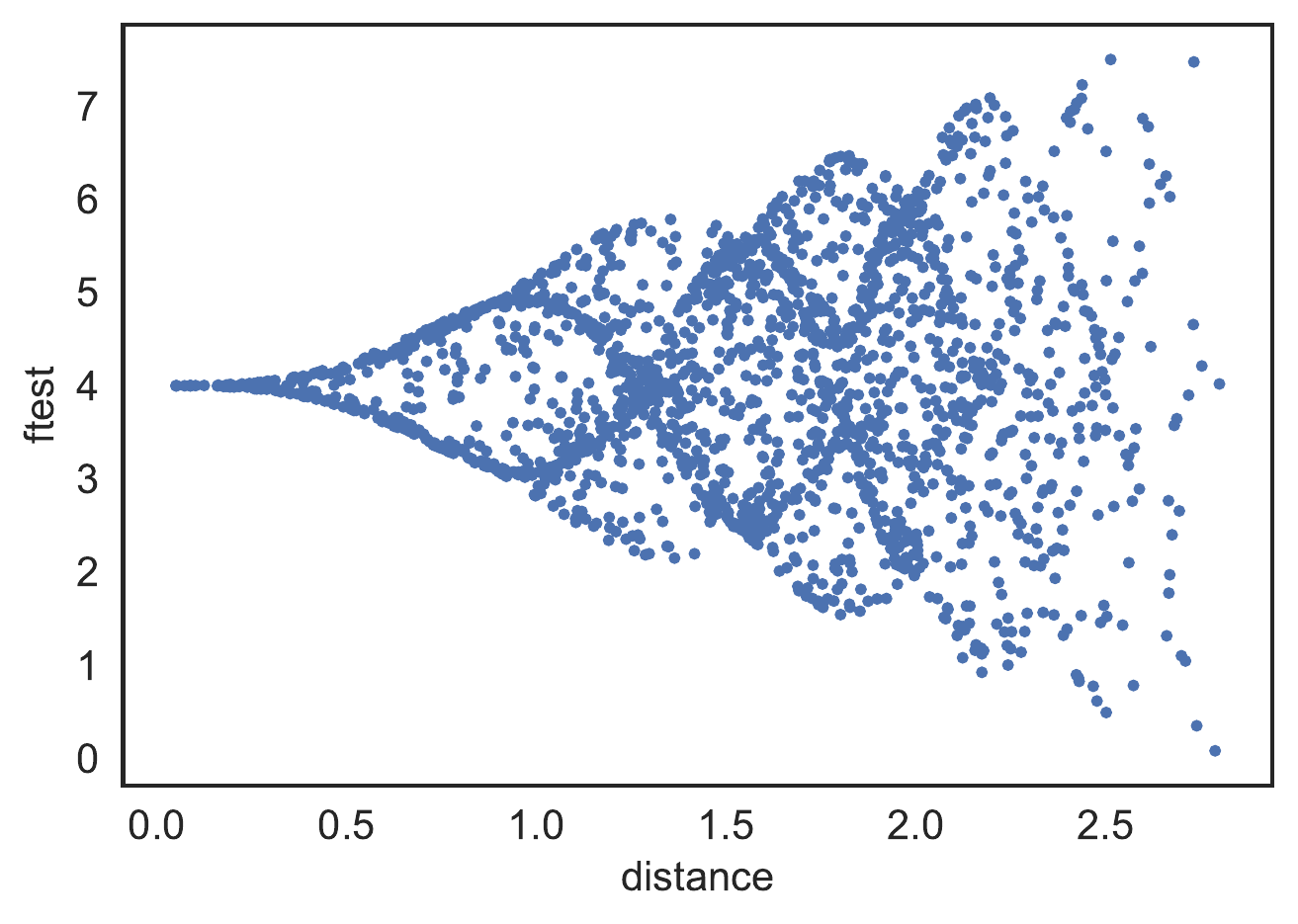}
\caption{Test $\rho(x)$ values by distance}
\end{subfigure}

\caption{Posterior summaries as a function of distance to the origin (observation point). }
\label{fig:summary-by-distance}
\end{figure}

\section{Additional Synthetic Experiments}
\label{sec:app-synthetic-experiments}

In this supplemental section, we examine the effect of hyper parameter
choices on algorithm performance.  Specifically, we study two choices:
(i) the whitened and standard parameterizations and (ii) the effect of
the Monte Carlo semi-integrated covariance estimator on stochastic
variational inference.

\subsection{Algorithm experiments}
\label{sec:algorithm-experiments}
The inference algorithm developed in Section~\ref{sec:scaling-inference}
involves two key choices, optimizing in the whitened space and the number
of Monte Carlo samples used to approximate the semi-integrated kernel.
To validate our choices, we investigate the impact of these choices on the
inference algorithm and predictions.
For simplicity, we assume the number and position of inducing points
is fixed.

\parhead{Whitened parameterization comparison}
\label{sec:whitening-comparison}
We compare the optimization performance of the two parameterizations
detailed above: (i) the \emph{standard} parameterization of variational
parameters $\bm, \bS$ and covariance function parameters $\btheta$; and
(ii) the \emph{whitened} parameterization of variational parameters
$\tilde{\bm}, \tilde{\bS}$ and covariance function parameters $\btheta$. 

We fit the above synthetic dataset with $N=5{,}000$ observations using
stochastic gradient optimization with batches of size $1{,}000$, for 10
epochs.  We fix the learning rate for the variational parameters, and we
vary the learning rate for the covariance function parameters to be $\Delta
\in$ [1e-4, 1e-5, 1e-6].  We also vary the initial value of $\sigma^2$ to
be in $\sigma^2_0 \in [1, 10, 100, 1{,}000]$. We fit 24 models, 12 standard
and 12 whitened. 

We compare the mean squared errors for the test $e$ and $\rho(x)$ values in
Figure~\ref{fig:whitened-comparison}, showing that the models fit in the
whitened space have lower mean squared errors.  Empirically, we also find
that optimizing in the whitened space leads to better predictions in fewer
epochs.

\begin{figure}[t!]
\centering
\begin{subfigure}[b]{.48\textwidth}
\centering
\includegraphics[width=\textwidth]{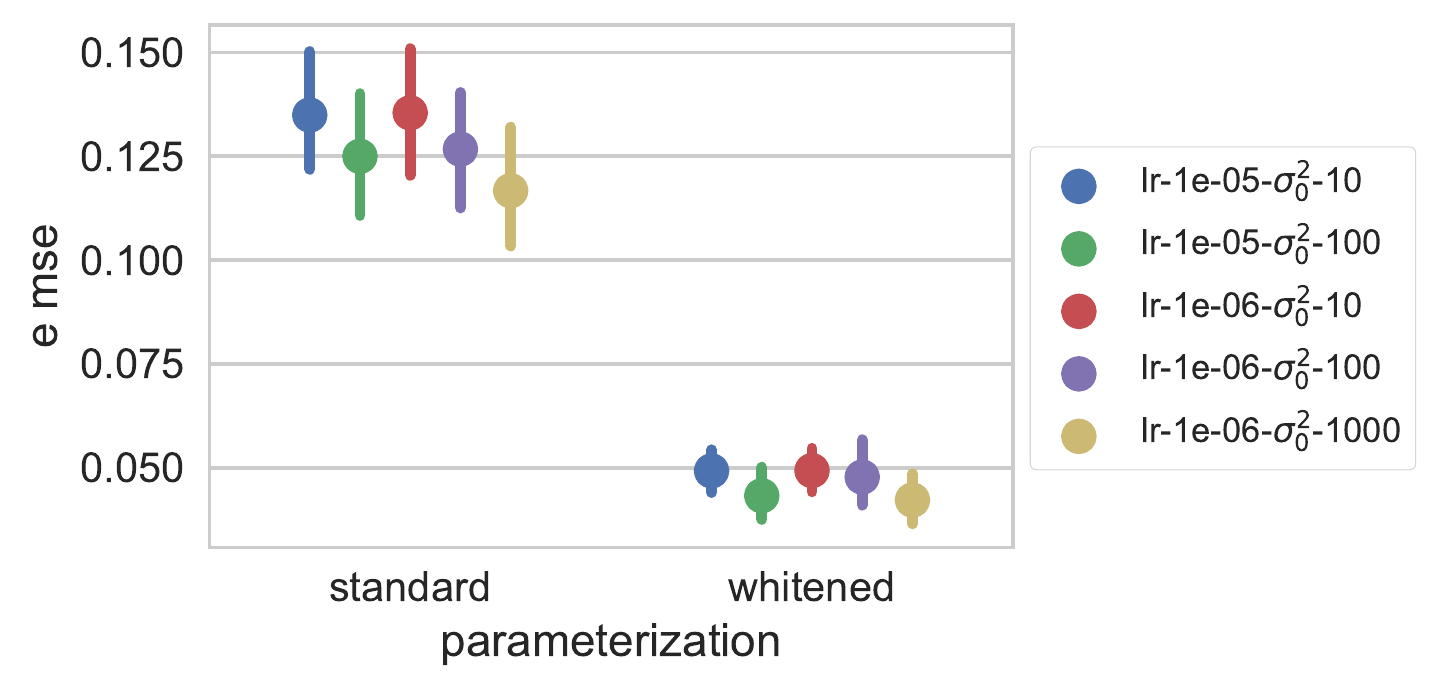}
\caption{Test $e$ MSEs}
\end{subfigure}
~
\begin{subfigure}[b]{.48\textwidth}
\centering
\includegraphics[width=\textwidth]{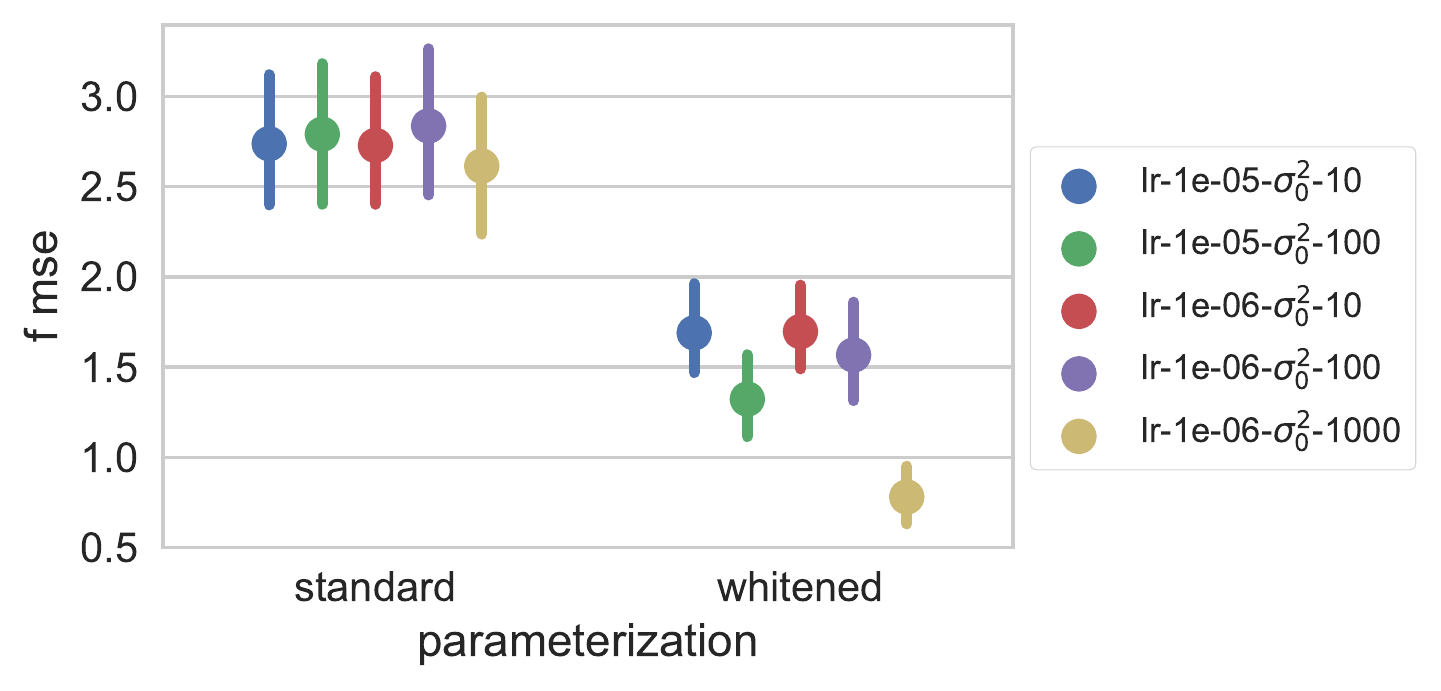}
\caption{Test $\rho(x)$ MSEs}
\end{subfigure}
\caption{Optimizing with the whitened parameterization leads to better
  predictions.  We ran stochastic optimization initialized at many different
  initial settings (the top five shown above) for ten epochs.  The resulting
  predictive mean squared error for integrated values $e$ (left) and
  pointwise values $\rho(x)$ (right) are plotted above.  Optimizing in the
  whitened space dramatically improves predictive accuracy. }
\label{fig:whitened-comparison}
\end{figure}

\parhead{Semi-Integrated Estimators}
Additionally, we examine the affect of the number of Monte Carlo samples used to
estimate the semi-integrated covariance function, $L$. We leverage the fact that
the squared exponential kernel admits a closed form semi-integrated covariance
function so we can directly compare Monte Carlo estimators with different sample
sizes $L$ to the analytic version Further, we use the analytic form of the
semi-integrated covariance function to compute the true variational objective
for all approaches; this true objective value is unavailable when the
semi-integrated covariance function is unavailable.
Figure~\ref{fig:optimization-comparison} compares optimization traces of the
variational objective using different values of $L$, ranging from 5 to 50.  We
see that the performance of the Monte Carlo gradients with $L=20$ closely mimics
the model fit with the exact gradient.  We use $L=30$ when optimizing with the
Gneiting and \matern~covariance functions to be conservative.

\end{document}